\documentclass{amsart}

\usepackage{fullpage}
\usepackage[utf8]{inputenc}
\usepackage[T1]{fontenc}
\usepackage{amsmath,amstext,amsthm,amsfonts,amssymb}
\usepackage{booktabs}
\usepackage[dvipsnames]{xcolor}
\usepackage{tikz}
\usetikzlibrary{arrows.meta,positioning,calc,decorations.pathreplacing}
\usepackage{float}
\usepackage[skins,breakable]{tcolorbox}
\usepackage{microtype}
\usepackage[linesnumbered,ruled]{algorithm2e}
\usepackage[colorlinks=true, linkcolor=blue, citecolor=green, urlcolor=cyan]{hyperref}
\usepackage[capitalize,nameinlink]{cleveref}

\newcommand{\nocolorref}[2]{\hypersetup{linkcolor=black}\hyperref[#1]{#2}\hypersetup{linkcolor=blue}}

\newtheorem*{theor*}{Theorem}
\newtheorem*{prop*}{Proposition}
\newtheorem*{theorem*}{Theorem}
\newtheorem*{lemma*}{Lemma}

\newtheorem{theor}{Theorem}[section]
\newtheorem{assump}[theor]{Assumption}
\newtheorem{lemma}[theor]{Lemma}
\newtheorem{prop}[theor]{Proposition}
\newtheorem{cor}[theor]{Corollary}

\newtheorem{defi}[theor]{Definition}

\theoremstyle{definition}
\newtheorem{rem}[theor]{Remark}

\theoremstyle{plain}

\crefname{theor}{Theorem}{Theorems}
\Crefname{theor}{Theorem}{Theorems}
\crefname{assump}{Assumption}{Assumptions}
\Crefname{assump}{Assumption}{Assumptions}
\crefname{lemma}{Lemma}{Lemmas}
\Crefname{lemma}{Lemma}{Lemmas}
\crefname{prop}{Proposition}{Propositions}
\Crefname{prop}{Proposition}{Propositions}
\crefname{cor}{Corollary}{Corollaries}
\Crefname{cor}{Corollary}{Corollaries}
\crefname{problem}{Problem}{Problems}
\Crefname{problem}{Problem}{Problems}
\crefname{conjecture}{Conjecture}{Conjectures}
\Crefname{conjecture}{Conjecture}{Conjectures}
\crefname{defi}{Definition}{Definitions}
\Crefname{defi}{Definition}{Definitions}
\crefname{definition}{Definition}{Definitions}
\Crefname{definition}{Definition}{Definitions}
\crefname{claim}{Claim}{Claims}
\Crefname{claim}{Claim}{Claims}
\crefname{fact}{Fact}{Facts}
\Crefname{fact}{Fact}{Facts}
\crefname{rem}{Remark}{Remarks}
\Crefname{rem}{Remark}{Remarks}
\crefname{remark}{Remark}{Remarks}
\Crefname{remark}{Remark}{Remarks}

\newcommand{\Prob}{\mathbb{P}}
\newcommand{\E}{\mathbb{E}}

\newcommand{\cal}{\mathcal}

\newcommand{\bfV}{\mathbf{V}}
\newcommand{\D}[2]{{\rm d}({#1},{#2})}
\newcommand{\rmp}{\mathrm{p}}
\newcommand{\Ksg}{{\rm K}_{\mathrm{sg}}}
\renewcommand{\sp}{\mathsf{s}}
\newcommand{\cn}[2]{\nocolorref{eq:def-normalized-average}{{\rm N}_{#1}(#2)}}
\newcommand{\acn}[2]{\nocolorref{def:latent-empirical-link-averages}{{\rmp}_{#1}(#2)}}
\newcommand{\pav}[1]{\overline{N}\!\left(#1\right)}
\newcommand{\apav}[1]{\overline{\rmp}\!\left(#1\right)}

\newcommand{\epsnav}[2]{\nocolorref{eq:def-fluctuation-scale}{\varepsilon_{#1}(#2)}}
\newcommand{\Ept}[2]{\nocolorref{eq:def-lower-occupancy-event}{\mathcal E_{\tr{pt}}(#1,#2)}}
\newcommand{\Enavi}[3]{\nocolorref{eq:def-navigation-event}{\mathcal E_{\rm link}(#1,#2;#3)}}
\newcommand{\Enn}[3]{\nocolorref{eq:def-pair-average-event}{\mathcal E_{\tr{avg}}(#1,#2,#3)}}
\newcommand{\step}[1]{\medskip  \noindent \underline{#1.\,}}
\newcommand{\tr}[1]{\text{\tiny{\rm #1}}}
\newcommand{\tref}[1]{\text{\tiny{\ref{#1}}}}

\newcommand{\keyterm}[1]{\textcolor{PineGreen}{\emph{#1}}}
\newcommand{\termref}[2]{\hyperref[#1]{\keyterm{#2}}}
\newcommand{\ora}{\mathrm{ora}}
\newcommand{\loc}{\mathrm{loc}}
\newcommand{\win}{\mathrm{win}}
\newcommand{\lowerphiregularity}{\termref{def:lower-phi-regularity}{lower \(\phi\)-regularity}}
\newcommand{\lowerregularityfunction}{\termref{def:lower-phi-regularity}{lower-regularity function}}
\newcommand{\latentobservationmodel}{\termref{def:graph_model}{latent-distance observation model}}
\newcommand{\observedweightedgraph}{\termref{rem:observed-weighted-graph}{observed weighted graph}}
\newcommand{\observededgeweight}{\termref{def:graph_model}{observed edge weight}}
\newcommand{\locallinkassumption}{\termref{assump:link-function}{local link assumption}}
\newcommand{\localbilipschitzwindow}{\termref{assump:link-function}{local bi-Lipschitz window}}
\newcommand{\latentlinkaverage}{\termref{def:latent-empirical-link-averages}{latent link average}}
\newcommand{\empiricallinkaverage}{\termref{def:latent-empirical-link-averages}{empirical link average}}
\newcommand{\linkaverageconcentrationevent}{\termref{eq:def-navigation-event}{link-average concentration event}}
\newcommand{\linkaveragethresholdrefinement}{\termref{prop:abstract-score-threshold-refinement}{link-average threshold refinement}}
\newcommand{\linkaverageseparation}{\termref{def:link-average-separated-seed}{link-average separation}}
\newcommand{\latentscorefunction}{\latentlinkaverage}
\newcommand{\normalizedaverage}{\empiricallinkaverage}
\newcommand{\pairaverage}{\termref{def:pair-averages}{pair average}}
\newcommand{\loweroccupancyevent}{\termref{eq:def-lower-occupancy-event}{lower occupancy event}}
\newcommand{\navigationevent}{\linkaverageconcentrationevent}
\newcommand{\uniformpairaverageevent}{\termref{eq:def-pair-average-event}{uniform pair-average event}}
\newcommand{\internalaveragecandidate}{\termref{def:internal-average-candidate}{internal-average candidate}}
\newcommand{\candidatefamily}{\termref{prop:inner-density-candidates}{candidate family}}
\newcommand{\candidatenet}{\termref{prop:candidate-refinement-net}{candidate net}}
\newcommand{\windowsafe}{\termref{def:window-safe}{window-safe}}
\newcommand{\fuzzywindoworacle}{\termref{def:fuzzy-window-oracle}{fuzzy window oracle}}
\newcommand{\scorethresholdrefinement}{\linkaveragethresholdrefinement}
\newcommand{\ORA}{\termref{rem:oracle-local-routes}{Oracle Route}}
\newcommand{\LOC}{\termref{rem:oracle-local-routes}{Local Route}}
\newcommand{\oracleroute}{\termref{rem:oracle-local-routes}{oracle route}}
\newcommand{\localroute}{\termref{rem:oracle-local-routes}{local route}}
\newcommand{\threeblockextraction}{\termref{thm:three-block-extraction}{three-block extraction}}

\newcommand{\chainproperty}{\termref{def:chain-property-rho0}{chain property}}

\title{Denoising Distances in Metric Measure Spaces}

\author{Han Huang}
\address{Department of Mathematics, University of Missouri, Columbia, Columbia, MO 65203}
\email{hhuang@missouri.edu}

\author{Pakawut Jiradilok}
\address{Science Division, Mahidol University International College, Nakhon Pathom 73170, Thailand}
\email{pakawut.jir@mahidol.edu}

\author{Elchanan Mossel}
\address{Department of Mathematics, Massachusetts Institute of Technology, Cambridge, MA 02139}
\email{elmos@mit.edu}

\begin{document}

\begin{abstract}
    Recent work studied the problem of finding clusters and denoising pairwise
distances from noisy distances of points sampled on a manifold. We study the same
problems in more general metric measure spaces under a lower mass condition. We
give an algorithm that extracts large localized clusters around every sampled
point, which can be used to denoise distances, with near-linear running time in
the dense regime for fixed target distance error $r$. When the target distance error \(r\) is
allowed to vanish as \(n\to\infty\), we identify the sharp information-theoretic
scale for achieving distance error \(r\), suggesting a statistical-computational
gap for high-accuracy denoising beyond the Riemannian setting.

\end{abstract}

\maketitle

\section{Introduction}
\label{sec:introduction}
\subsection*{A general model for noisy distance data.}
A common way to model noisy metric data is through latent-distance observations.
There are unobserved points \(X_1,\ldots,X_n\) in a metric space
\((M,{\rm d})\), sampled according to a probability measure \(\mu\). Instead of
observing the distances \({\rm d}(X_u,X_v)\) directly, we observe noisy
pairwise measurements $Z_{u,v}$ between vertices \(u\) and \(v\). The distribution of
these measurements depends on the latent distance, typically through a monotone
link function. This viewpoint includes latent-space network models, random
geometric graphs, and latent-distance graph models
\cite{HoffRafteryHandcock2002LatentSpace,Penrose2003RGG,ArayaDeCastro2019LatentDistance}.
The statistical goal is to use these noisy pairwise observations to recover the
latent metric, or at least to estimate latent distances up to a prescribed
accuracy.

\step{Lower \(\phi\)-regularity}
A standard assumption that makes local metric recovery statistically meaningful
is a lower bound on the mass of small balls. Conditions of this type appear in
nearest-neighbor analysis, nonparametric estimation, random geometric graphs,
graph-based manifold learning, and analysis on metric measure spaces, under
names such as lower mass conditions, small-ball assumptions, minimal mass
assumptions, or lower volume growth
\cite{ChaudhuriDasgupta2014NN,GadatKleinMarteau2016KNN,Penrose2003RGG,TenenbaumDeSilvaLangford2000Isomap,BelkinNiyogi2003LaplacianEigenmaps,HeinAudibertLuxburg2007GraphLaplacians,Heinonen2001LecturesMetricSpaces}.
At a high level, they ensure that local
neighborhoods contain enough probability mass for statistical estimates to
concentrate uniformly across the space.

In our notation, \lowerphiregularity{} means that there is a nondecreasing
\lowerregularityfunction{} \(\phi\) such that
\[
\mu(B(x,r))\ge \phi(r)
\]
for every support point \(x\) and every relevant small radius \(r\), where $B(x,r)$ denotes the open ball of radius $r$ centered at $x$. In a
\(d\)-dimensional space, one can think of \(\phi(r)\) as proportional to
\(r^d\), meaning that every radius-\(r\) ball has mass at least a constant times
\(r^d\).

\step{Local link assumption}
The second structural assumption concerns the link between latent distance and
observed edge statistics. By link function, we mean a function $\rmp$ such that our noisy observation $Z_{u,v}$ has expectation $\E[Z_{u,v}\mid X_u,X_v] =  \rmp(\D{X_u}{X_v})$.
Monotonicity says that the observations are ordered by
distance: in similarity graphs, nearby points tend to produce stronger or more
frequent edges, while in noisy distance models the expected observation
increases with distance. Monotonicity alone, however, is not enough for metric
recovery at small scales.
Close points may become statistically indistinguishable from one another.

The \locallinkassumption{} is the following local identifiability condition:
the link function is bi-Lipschitz on a neighborhood \([0,r_\rmp]\) of zero.
We refer to this neighborhood as the \localbilipschitzwindow{}. This means
that, for small distances, a change in latent distance of order \(r\) produces
a change in the expected observation of order \(r\).

The assumption is deliberately local. We do not require the link to remain
informative at large distances; it may flatten, saturate, or otherwise lose
metric information outside the \localbilipschitzwindow{}. This is natural in similarity
graphs, where faraway points may all have nearly indistinguishable connection
probabilities.

\step{Latent-distance observation model}
We work with a general \latentobservationmodel{}. The latent points
\(X_v\) are sampled independently from a metric probability space
\((M,{\rm d},\mu)\) satisfying the \lowerphiregularity{} condition above. For each pair
\(u\neq v\), there is a noisy weighted measurement whose mean depends only on
the latent distance \({\rm d}(X_u,X_v)\). Informally,
\[
Z_{u,v} = B_{u,v} \tilde{Z}_{u,v},
\]
where $\tilde{Z}_{u,v}$ is a subgaussian random variable with mean $\rmp({\rm d}(X_u,X_v))$, and $B_{u,v}$ is a Bernoulli random variable with parameter $\sp$, independent of $\tilde{Z}_{u,v}$, which models the sparsity of the \observedweightedgraph{} ($\sp$ can be $n$ dependent). Meaning that with probability $\sp$, the noisy measurement is observed, and with probability $1-\sp$, it is $0$. The link function \(\rmp\) is a monotone link function satisfying the \locallinkassumption{} described above, and bounded on $[0, {\rm diam}(M)]$. Further, we assume $Z_{u,v}$ are independent across pairs
up to symmetry $Z_{u,v}=Z_{v,u}$, conditional on the latent points.
The precise formulation is given in
Definition~\ref{def:graph_model} and
Assumptions~\ref{assump:lower-regularity}, \ref{assump:link-function}, and
\ref{assump:graph_model}.

For terminology, we view the observed data as an \observedweightedgraph{}:
a weighted graph \(G\) on vertex set \(\{1,\ldots,n\}\), where the weight of
the edge \((u,v)\) is \(Z_{u,v}\).  The latent points \(X_v\) are unobserved.
This leads to the central question of the paper:

\emph{Under these two assumptions, and without any knowledge of the structure of the latent metric
space, how much can we recover latent distance between the latent points from the \observedweightedgraph{}?}

\subsection*{Examples}
Before stating our results, we discuss several examples that fit into this framework.

\step{Noisy distance data on a manifold}
The first motivating example is noisy distance data on a manifold. Suppose that
\(M\) is a \(d\)-dimensional Riemannian manifold, the latent points
\(X_1,\ldots,X_n\) are sampled from a regular probability measure on \(M\), and
we observe
\[
Z_{u,v}
=
q({\rm d}(X_u,X_v))+\xi_{u,v},
\]
where the raw mean link is monotone increasing and the noises \(\xi_{u,v}\)
are centered subgaussian random variables. For instance, the natural raw link
is \(q(t)=t\).  The usual regularity assumptions on a \(d\)-dimensional manifold
imply that small balls have
mass at least a constant times \(r^d\), so in our notation one may take
\(\phi(r)\gtrsim r^d\).

\step{Soft random geometric graphs}
The Bernoulli soft random geometric graph is a basic example
\cite{Penrose2003RGG,DucheminDeCastro2022RGGSurvey}. In this model,
\(Z_{u,v}\in\{0,1\}\), and
\[
\mathbb P(Z_{u,v}=1\mid X_u,X_v)
=
\mathsf{s}\,\rmp({\rm d}(X_u,X_v)),
\]
where \(\rmp:\mathbb R_+\to[0,1]\) is non-increasing. Thus nearby latent points
are more likely to be connected by an edge. This captures, for instance, models
in which two devices are more likely to communicate when their latent distance
is small, or more generally similarity graphs in which edge probabilities
decrease with distance.

These two examples are the main objects in recent work on manifold distance
denoising and random geometric graphs. Huang, Jiradilok, and
Mossel~\cite{HuangJiradilokMossel2024GeometryRGG,HuangJiradilokMossel2025RiemannianMetrics,HuangJiradilokMossel2026BeyondVolumetric}
and Fefferman, Marty, and Ren~\cite{FeffermanMartyRen2025NoisyDistances} study
closely related problems for points sampled from \(d\)-dimensional Riemannian
manifolds or similar geometric spaces, under regularity assumptions that imply
\lowerphiregularity{} of the form
\[
\mu(B(x,r))\gtrsim r^d.
\]
Their algorithms recover latent distances at polynomial rates $n^{-\Theta(1/d)}$, where $n$ is the number of sampled points and $d$ is the intrinsic dimension. Latent-distance estimation in random geometric graphs has also been studied through spectral and graph-distance methods in Euclidean or
spherical settings
\cite{ArayaDeCastro2019LatentDistance,AriasCastroChannarondPelletierVerzelen2021GraphDistances}, though these methods usually assume concrete knowledge of the latent geometry but in the hard disc setting, where the link function is an indicator of whether the distance is below a threshold.

\step{Beyond manifolds}
Beyond manifold models from~\cite{HuangJiradilokMossel2024GeometryRGG,HuangJiradilokMossel2025RiemannianMetrics,HuangJiradilokMossel2026BeyondVolumetric,FeffermanMartyRen2025NoisyDistances}, our framework also covers a variety of non-Euclidean latent
geometries, provided the sampling measure satisfies appropriate
\lowerphiregularity{}.

This includes more singular metric measure spaces, such as compact
metric trees, finite metric graphs, stratified spaces, or Ahlfors-regular
fractal sets, the Heisenberg group with the Carnot–Carath\'{e}odory metric as well as subsets of Euclidean spaces with rough boundaries.
In such examples, the \lowerregularityfunction{} \(\phi\) records the relevant local
mass growth; for instance, an Ahlfors-regular fractal with dimension
\(d_f\) has \(\phi(r)\) comparable to \(r^{d_f}\), where \(d_f\) need not be an
integer.

\step{Stochastic block models}
Here we illustrate an example at the opposite extreme. Let \(M=\{a_1,\ldots,a_k\}\) be a finite metric space with
\[
{\rm d}(a_i,a_i)=0,
\qquad
{\rm d}(a_i,a_j)=1
\quad (i\neq j),
\]
and let \(\mu\) be a probability distribution on \(M\). Sampling
\(X_v\sim\mu\) assigns each vertex to one of \(k\) latent classes. If
\[
\mathbb P(Z_{u,v}=1\mid X_u,X_v)
=
\mathsf{s}\,\rmp({\rm d}(X_u,X_v)),
\]
then vertices in the same class connect with probability
\(
\mathsf{s}\,\rmp(0),
\)
while vertices in different classes connect with probability
\(
\mathsf{s}\,\rmp(1).
\)
Thus the model reduces to a stochastic block model with \(k\) blocks, within
block probability \(\mathsf{s} \, \rmp(0)\), and between-block probability
\(\mathsf{s} \, \rmp(1)\). This connects the finite-metric special case to the
classical stochastic block model
\cite{HollandLaskeyLeinhardt1983SBM,Abbe2017SBMSurvey}. In this finite
example, \lowerphiregularity{} is simply a lower bound on the smallest block
mass. This example illustrates that the metric measure framework is not
restricted to smooth spaces; it also includes highly discrete latent
geometries.

\step{Monotonicity convention}
The preceding examples illustrate the two possible orientations:
noisy-distance observations naturally have a non-decreasing mean link, while
similarity graphs naturally have a non-increasing link. For real-valued
observations this orientation is immaterial. If the raw observations have
conditional mean
\[
\mathbb E[Z_{u,v}\mid X_u,X_v]
=
q(\D{X_u}{X_v})
\]
with \(q\) non-decreasing, we replace \(Z_{u,v}\) by \(-Z_{u,v}\) and use the
link function \(\rmp=-q\). Then \(\rmp\) is non-increasing, with the same local
bi-Lipschitz constants up to sign. In all formal statements and proofs below,
we adopt this non-increasing similarity convention, relabeling the transformed
observations as \(Z_{u,v}\).

\subsection*{Results}
\step{From distance recovery to local cluster extraction}
A natural way to estimate latent distances is to first construct local
neighborhoods around the sampled points.  Suppose, informally, that for each
vertex \(v\) we could find (by examining the \observedweightedgraph{}) a large set \(U_v\), whose latent points all lie
within a small radius \(r\) of \(X_v\).  Then \(U_v\) can be used as a local
statistical proxy for \(X_v\).  For two vertices \(v,w\), the average observed
weight between \(U_v\) and \(U_w\) concentrates around its conditional
expectation; because the clusters are localized, this expectation is close to
\(\rmp({\rm d}(X_v,X_w))\).  Thus, once such local clusters are available,
distance recovery can be reduced to averaging edge weights between clusters and
inverting the link function if it is known.

Our main technical result is a cluster-extraction theorem: under
\lowerphiregularity{} and the \locallinkassumption{}, we construct large localized clusters
around all target vertices.  The corresponding distance-denoising statements
then follow as consequences.

In the context below, we always assume any algorithm has access to the \observedweightedgraph{},
the parameters of $\rmp$ (bi-Lipschitz constant, the radius where $\rmp$ is bi-Lipschitz, and the $\max_{t \in [0,\operatorname{diam}(M)]} |\rmp(t)|$), the \lowerregularityfunction{} \(\phi\), and the sparsity parameter \(\mathsf{s}\), but not the latent points, any other structural information about the latent metric space, or the link function \(\rmp\) beyond the stated assumptions.

\begin{theor}[Informal statement of Theorem \ref{theor:main}]
There exists an algorithm that takes as input the \observedweightedgraph{} and a target scale \(r\), and outputs, with high
probability, for every target vertex \(v\) a set \(U_v\) satisfying
\[
v\in U_v,
\qquad
X_{U_v}\subseteq B(X_v,Cr),
\qquad
|U_v|\ge c\,\phi(r)n,
\]
provided \(0<\sqrt r\le c\min\{1,r_\mu,r_\rmp\}\) and
\[
\mathsf{s}n\,\phi(r)r^2
\ge
C\Lambda \log n\,.
\]
The algorithm has running time
$$
    \exp\left( C \log^2\!\left(\frac{3}{r\phi(r)}\right)
    \frac{1}{\sp r^2}\right)n\Lambda\log n\,.
$$
Further, the collection
$$
    \{U_v\}_{v=1}^n
$$
without counting multiplicity is bounded by $$
    \exp\left( C \log^2\!\left(\frac{3}{r\phi(r)}\right)
    \frac{1}{\sp r^2}\right).
$$
The constants \(C,c\), and the required lower bound on \(\Lambda\), depend only
on the model parameters.
\end{theor}

Usually, a cluster result of this type can be translated into distance estimates by considering the average observed edge weight between two clusters. These ideas are standard in the literature on clustering and community detection, and they also appear in recent work on distance denoising under manifold assumptions \cite{HuangJiradilokMossel2024GeometryRGG,HuangJiradilokMossel2025RiemannianMetrics,HuangJiradilokMossel2026BeyondVolumetric,FeffermanMartyRen2025NoisyDistances}.
Here we state a reader-friendly version of this translation.
\begin{cor}[Informal distance recovery from extracted clusters]
    \label{cor:distance-recovery-from-clusters}
Assume the hypotheses of the cluster-extraction theorem.

\begin{enumerate}
\item If \(\rmp\) is known and bi-Lipschitz on
\([0,\operatorname{diam}(M)]\), then inverting cluster averages recovers all pairwise
distances \(\D{X_v}{X_w}\) with error \(O(r)\).

\item If \(\rmp\) is known but bi-Lipschitz only on \([0,r_\rmp]\), then the
same argument recovers distances below the \localbilipschitzwindow{} with error $O(r)$. Large distances recovery is in general not possible. (See Figure \ref{fig:flat-link-large-distance}) However, if $M$ is a geodesic space, then distance recovery still holds for all distances.
\end{enumerate}
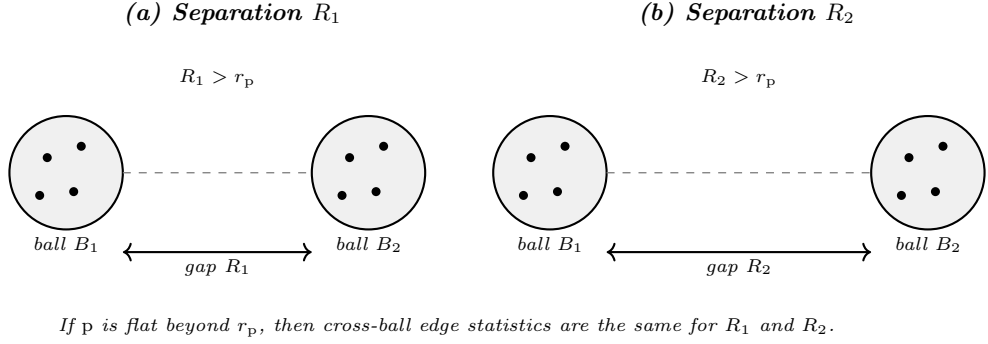
\begin{figure}[t]
\centering
\begin{tikzpicture}[
    scale=1.0,
    ball/.style={draw=black, fill=gray!12, line width=0.8pt},
    point/.style={circle, fill=black, inner sep=1.2pt},
    faintedge/.style={gray, dashed, line width=0.5pt},
    lab/.style={font=\scriptsize},
    title/.style={font=\small\bfseries}
]

\begin{scope}[shift={(0,0)}]
\node[title] at (2.2,2.1) {(a) Separation \(R_1\)};

\draw[ball] (0,0) circle (0.75);
\draw[ball] (4.0,0) circle (0.75);

\node[point] at (-0.25,0.20) {};
\node[point] at (0.20,0.35) {};
\node[point] at (0.10,-0.25) {};
\node[point] at (-0.35,-0.30) {};

\node[point] at (3.75,0.20) {};
\node[point] at (4.20,0.35) {};
\node[point] at (4.10,-0.25) {};
\node[point] at (3.65,-0.30) {};

\draw[faintedge] (0.75,0) -- (3.25,0);

\draw[<->, thick] (0.75,-1.05) -- (3.25,-1.05);
\node[lab] at (2.0,-1.25) {gap \(R_1\)};

\node[lab] at (0,-0.95) {ball \(B_1\)};
\node[lab] at (4.0,-0.95) {ball \(B_2\)};
\node[lab, align=center] at (2.0,1.25) {\(R_1>r_\rmp\)};
\end{scope}

\begin{scope}[shift={(6.4,0)}]
\node[title] at (2.6,2.1) {(b) Separation \(R_2\)};

\draw[ball] (0,0) circle (0.75);
\draw[ball] (5.0,0) circle (0.75);

\node[point] at (-0.25,0.20) {};
\node[point] at (0.20,0.35) {};
\node[point] at (0.10,-0.25) {};
\node[point] at (-0.35,-0.30) {};

\node[point] at (4.75,0.20) {};
\node[point] at (5.20,0.35) {};
\node[point] at (5.10,-0.25) {};
\node[point] at (4.65,-0.30) {};

\draw[faintedge] (0.75,0) -- (4.25,0);

\draw[<->, thick] (0.75,-1.05) -- (4.25,-1.05);
\node[lab] at (2.5,-1.25) {gap \(R_2\)};

\node[lab] at (0,-0.95) {ball \(B_1\)};
\node[lab] at (5.0,-0.95) {ball \(B_2\)};
\node[lab, align=center] at (2.5,1.25) {\(R_2>r_\rmp\)};
\end{scope}

\node[lab, align=center] at (5.1,-2.05) {
If \(\rmp\) is flat beyond \(r_\rmp\), then cross-ball edge statistics are the same for \(R_1\) and \(R_2\).
};

\end{tikzpicture}
\caption{Large-distance non-identifiability for a locally informative link. Two latent spaces may have identical local geometry but different separation between components. If the link function is flat beyond the local scale \(r_\rmp\), cross-component observations cannot distinguish the gaps \(R_1\) and \(R_2\).}
\label{fig:flat-link-large-distance}
\end{figure}
\end{cor}
\begin{rem}
    For the second case, the geodesic assumption can be relaxed; see Corollary~\ref{cor:global-distance-by-chaining}. It is enough to have chains whose consecutive distances are below \(r_\rmp\) and whose total length approximates the original distance; the approximation error is reflected in the final recovery error.
\end{rem}
\begin{rem}[Unknown link]
When the link function \(\rmp\) is unknown, one can still adapt the bisection/calibration strategy of \cite{FeffermanMartyRen2025NoisyDistances} after the localized clusters have been extracted. This would recover the metric up to an unknown global dilation factor, with the corresponding logarithmic loss. We do not pursue this direction in the present manuscript.
\end{rem}

\subsection*{Efficiency at fixed accuracy.}
In general, our algorithm is inefficient because it includes an exhaustive search. However, we still have an efficient algorithm when we only aim for a fixed accuracy in the dense regime.
\begin{cor}[Efficiency Statement]
    When $r$ is fixed and in the dense regime $\sp$ is a constant, the cluster-extraction algorithm runs in time
    \[
    n \, \log(n),
    \]
    and the same running time applies to distance recovery when $\rmp$ is known and bi-Lipschitz on the full distance range.
    The additional running time for distance recovery is polynomial in the number of distinct extracted clusters, which is bounded by \(\exp(C\log^2(1/\phi(r)))\) and is therefore constant when \(r\) is fixed. The distance recovery step only computes averages between pairs of these distinct clusters, and vertices in the same cluster have latent distance at most \(O(r)\), which is already within the recovery error scale.
\end{cor}

\subsection*{Information-theoretic bounds.}
In the case where $\rmp$ is Bi-Lipschitz on the full distance range, or say the underlying metric space is geodesic, the cluster-extraction algorithm from Theorem~\ref{theor:main} together with cluster-to-distance recovery Corollary~\ref{cor:distance-recovery-from-clusters} can be used to recover all pairwise distances with error $O(r)$ when
$$
\mathsf{s}n\,\phi(r)r^2 \gtrsim \log n\,.
$$
Below we provide a matching lower bound on the order.

\begin{lemma}[Lower bound]
\label{lem:lower-bound}
There are universal constants \(c_\phi,C_0,c_0,c_1>0\) such that the following
holds for all sufficiently large \(n\). Let \(\phi\) be nondecreasing, let
\(\mathsf{s}\in(0,1]\), let \(r_0\in(0,1]\), and assume
\(0<\phi(r_0)\le c_\phi\). If
\[
\mathsf{s}n\phi(r_0)> C_0
\qquad\text{and}\qquad
\mathsf{s}n\phi(r_0)r_0^2<c_0\log n,
\]
then there is a finite metric probability space that is lower-\(\phi\)-regular
up to scale \(r_0\), together with a locally bi-Lipschitz observation model
satisfying the standing assumptions with sparsity parameter \(\mathsf{s}\), for
which every estimator \(\widehat d\) satisfies
\[
\mathbb P\!\left(
\max_{u,v}|\widehat d(u,v)-\D{X_u}{X_v}|\ge r_0/2
\right)
\ge c_1.
\]
\end{lemma}

The idea is to consider a discrete metric space with \(k\simeq1/\phi(r_0)\)
points, all at mutual distance \(r_0\). The corresponding latent-distance
observation model is a symmetric stochastic block model with \(k\) blocks, within
block probability \(p=\sp\rmp(0)\), and between-block probability
\(q=\sp\rmp(r_0)\). Recovering all pairwise distances with error strictly less
than \(r_0/2\) would recover the induced block partition exactly. Thus the lower
bound is the metric translation of the exact-recovery threshold for stochastic
block models: in the weak-signal regime relevant here, the decisive quantity is
of order
\[
    \frac{(p-q)^2}{p}\frac{n}{k}
    \simeq
    \sp n\phi(r_0)r_0^2.
\]
The impossibility side of this threshold is classical for the symmetric
two-block model and its exact-recovery variants
\cite{MosselNeemanSly2016Consistency,AbbeBandeiraHall2016Exact}, with minimax
multi-block formulations in \cite{zhang2016minimax}; see also the survey
\cite{Abbe2017Community}. We include the proof of Lemma~\ref{lem:lower-bound}
because our reduction requires a growing number of blocks, random block-size
fluctuations, and the precise normalization used in the present metric model. The proof is given in Appendix~\ref{sec:lower-bound}.

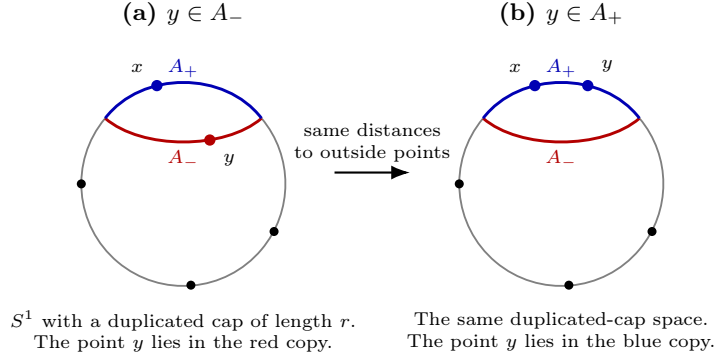
\begin{figure}[t]
\centering
\begin{tikzpicture}[
    scale=1.0,
    point/.style={circle, fill=black, inner sep=1.2pt},
    redpoint/.style={circle, fill=red!70!black, inner sep=1.5pt},
    bluepoint/.style={circle, fill=blue!70!black, inner sep=1.5pt},
    thickblue/.style={line width=1.1pt, blue!70!black},
    thickred/.style={line width=1.1pt, red!70!black},
    faint/.style={gray, line width=0.7pt},
    lab/.style={font=\scriptsize},
    title/.style={font=\small\bfseries}
]

\begin{scope}[shift={(0,0)}]
\node[title] at (0,2.25) {(a) \(y\in A_-\)};

\draw[faint] (140:1.35) arc[start angle=140,end angle=400,radius=1.35];

\coordinate (L) at (140:1.35);
\coordinate (R) at (40:1.35);

\draw[thickred] (L) .. controls (-0.55,0.45) and (0.55,0.45) .. (R);
\draw[thickblue]  (L) .. controls (-0.55,1.50) and (0.55,1.50) .. (R);

\node[lab, red!70!black] at (0,0.33) {\(A_-\)};
\node[lab, blue!70!black] at (0,1.53) {\(A_+\)};

\node[bluepoint,label=above left:{\scriptsize \(x\)}] at (-0.35,1.30) {};
\node[redpoint,label=below right:{\scriptsize \(y\)}] at (0.35,0.58) {};
\node[point] at (-1.35,0.0) {};
\node[point] at (1.2,-0.63) {};
\node[point] at (0.1,-1.34) {};

\node[lab, align=center] at (0,-1.95) {
$S^1$ with a duplicated cap of length $r$.\\
The point \(y\) lies in the red copy.
};
\end{scope}

\draw[-{Latex[length=2.5mm]}, thick] (2.0,0.15) -- (3.0,0.15);
\node[lab, align=center] at (2.5,0.55) {same distances\\to outside points};

\begin{scope}[shift={(5.0,0)}]
\node[title] at (0,2.25) {(b) \(y\in A_+\)};

\draw[faint] (140:1.35) arc[start angle=140,end angle=400,radius=1.35];

\coordinate (Lr) at (140:1.35);
\coordinate (Rr) at (40:1.35);

\draw[thickred] (Lr) .. controls (-0.55,0.45) and (0.55,0.45) .. (Rr);
\draw[thickblue]  (Lr) .. controls (-0.55,1.50) and (0.55,1.50) .. (Rr);

\node[lab, red!70!black] at (0,0.33) {\(A_-\)};
\node[lab, blue!70!black] at (0,1.53) {\(A_+\)};

\node[bluepoint,label=above left:{\scriptsize \(x\)}] at (-0.35,1.30) {};
\node[bluepoint,label=above right:{\scriptsize \(y\)}] at (0.35,1.30) {};

\node[point] at (-1.35,0.0) {};
\node[point] at (1.2,-0.63) {};
\node[point] at (0.1,-1.34) {};

\node[lab, align=center] at (0,-1.95) {
The same duplicated-cap space.\\
The point \(y\) lies in the blue copy.
};
\end{scope}

\end{tikzpicture}
\caption{A folding obstruction at the volumetric scale, illustrated for \(\phi(r)=r\). Both panels show the same space: a circle with a duplicated cap of length \(r\), equipped with the intrinsic metric. The only displayed change is the location of \(y\): in panel (a), \(y\) lies in the red copy \(A_-\); in panel (b), it lies in the matched blue copy \(A_+\). Every point outside \(A_-\cup A_+\) has the same distance to these two matched locations. Therefore the two configurations are distinguishable only through interactions between the two duplicated caps. There are about \(n\phi(r)\) sampled points in each cap, and cross-cap edge probabilities change by \(O(\sp r)\), giving the heuristic obstruction \(O(\sp n\phi(r)r^2)\). The same idea extends to duplicated caps in \(S^d\), where \(\phi(r)\asymp r^d\), and to analogous spaces adapted to a prescribed \(\phi\).
\label{fig:continuous-folding-lower-bound}
}
\end{figure}

\subsection*{Related work}
The closest works to ours are the recent distance-denoising results of Huang,
Jiradilok, and Mossel~\cite{HuangJiradilokMossel2024GeometryRGG,HuangJiradilokMossel2025RiemannianMetrics,HuangJiradilokMossel2026BeyondVolumetric} and Fefferman, Marty, and
Ren~\cite{FeffermanMartyRen2025NoisyDistances}. These works study closely related latent-distance observation
models under geometric assumptions, most prominently when the latent points are
sampled from a \(d\)-dimensional Riemannian manifold. In such settings, the
sampling measure satisfies \lowerphiregularity{} of the form
\[
\mu(B(x,r))\gtrsim r^d
\]
at small scales, and the intrinsic dimension \(d\) governs the achievable
distance-denoising rates.

The work of Fefferman, Marty, and Ren~\cite{FeffermanMartyRen2025NoisyDistances} also treats settings beyond
smooth manifolds. They do so by abstracting the conditions needed for their manifold estimation algorithm.
However, the results of~\cite{FeffermanMartyRen2025NoisyDistances} do not apply for fractal or discrete settings and many other spaces where our results hold.

\subsection*{Proof ideas}
Here in this paper we simply assume the link function $\rmp$ is monotone non-increasing, as in similarity graphs. The other case is similar after $Z_{u,v}$ by $-Z_{u,v}$ and $\rmp$ by $-\rmp$.
Let $\bfV$ denote the vertex set of the \observedweightedgraph{}, so $\bfV$ is the index set of the latent points. For a subset \(U\subseteq\bfV\), we define the \latentscorefunction{} \(\acn{U}{\cdot}:M\to\mathbb R\) by
\[
\acn{U}{y}
:=
\frac1{|U|}
\sum_{u\in U}
\rmp({\rm d}(X_u,y)),
\qquad y\in M.
\]
When the argument is a vertex \(v\), we use the convention
\(\acn{U}{v}:=\acn{U}{X_v}\). For a vertex $v \notin U$, the observable
counterpart is the \normalizedaverage{}
\[
\cn{U}{v}:=
\frac{1}{\mathsf{s}|U|}
\sum_{u\in U}Z_{u,v}.
\]
Conditionally on the latent positions \(X_U\) and \(X_v\),
\[
\mathbb E\bigl[\cn{U}{v}\mid X_U,X_v\bigr]
=
\acn{U}{v}.
\]
Thus \(\cn{U}{v}\) is a noisy empirical evaluation of the latent link average
at \(v\).
When \(U\) is a cluster with center \(x\in M\), we expect
\(\acn{U}{y}\) to be a good approximation of
\(\rmp(\D{x}{y})\): averaging over points \(X_u\) near \(x\) should nearly
produce the link value from \(x\) to \(y\). So if $U$ is a cluster around $X_w$, then \(\cn{U}{v}\) is a good proxy for \(\rmp(\D{X_w}{X_v})\). This is the key to distance recovery from clusters.

The algorithm is built from two ingredients. The first produces many
\internalaveragecandidate{}s by exhaustive search. The second refines an \internalaveragecandidate{}
into a cleaner cluster by thresholding its edge averages to fresh vertices
(see Definitions~\ref{def:cluster} and~\ref{def:normalized-averages}, and
Proposition~\ref{prop:abstract-score-threshold-refinement}).

\subsubsection*{Ingredient 1: internal-average search}
The starting point is the monotonicity of the link function. Since \(\rmp\) is
monotone decreasing, a set of vertices with unusually large internal average
edge weight should correspond to latent points that are mostly close to one
another. Conversely, by \lowerphiregularity{}, every small latent ball
contains many sampled vertices, and those vertices form a set with high
internal average; this uses Assumptions~\ref{assump:lower-regularity}
and~\ref{assump:link-function}, together with the lower-occupancy estimate in
Lemma~\ref{lem:uniform-lower-occupancy}.

Fix a scale \(\rho\). When $n$ is large enough, we expect for each latent point $x$ that the ball \(B(x,\rho)\) contains about \(n\phi(\rho)\) sampled points, let us denote it by $U_x$ for the moment. For any pair of points \(u,v\in U_x\), $\rmp(\D{X_u}{X_v})$ is at least $\rmp(0) - L_\rmp \rho$.
Thus, we can search for subsets $U$ of with
$$
    \pav{U} := \frac{1}{\sp}\frac{1}{|U|^2}\sum_{u,v\in U}Z_{u,v}
$$
whose expected value is at least $\rmp(0) - L_\rmp \rho$. (The $\rmp(0)$ is not known, but can be well estimated from the maximum internal average of subsets of size about $n\phi(\rho)$, see Lemma~\ref{lem:p0-from-max-internal-average}.)

In order for $\pav{U}$ to be realistically close to its expectation, we need its fluctuation to be smaller than $\rho$ for every $U$ of size about $n\phi(\rho)$. This requires a search-scale condition of the form
\[
\sp n \phi(\rho)\rho^2\gtrsim \log n.
\]
In the final two-round construction, the internal search scale is chosen so that this reduces to the main sample-size condition at target scale \(r\).

This exhaustive search produces a \candidatefamily{} which we denote by \({\cal C}_\rho\). (See Proposition~\ref{prop:inner-density-candidates}, with the estimate of \(\rmp(0)\) supplied by Lemma~\ref{lem:p0-from-max-internal-average}.)
Then, this \candidatefamily{} has two useful properties.
\begin{itemize}
    \item
First, it covers the latent space:
every ball of radius \(\rho\) in the latent space contributes at least one
\internalaveragecandidate{}, as in
Proposition~\ref{prop:inner-density-candidates}. That is, for every point
\(x\in M\), there is a candidate \(U\in{\cal C}_\rho\) such that all latent
points in \(U\) lie within distance \(\rho\) of \(x\).
\item
Second, every such candidate localizes after Markov. Namely, a candidate \(U\)
has a representative point \(\theta_U \in M\) such
that, for larger radii \(R\), all but an \(O(\rho/R)\)-fraction of the latent
points in \(U\) lie within distance \(R\) of \(\theta_U\), as long as $R$ is
within the \localbilipschitzwindow{} of the link function $\rmp$; see
Lemma~\ref{lem:internal-gap-average-distance} and
Remark~\ref{rem:internal-average-candidate-cluster}.
\end{itemize}

\step{Simplified when $\rmp$ is bi-Lipschitz on the full distance range}
If $\rmp$ were bi-Lipschitz on the full distance range, then the second property would be enough for distance recovery: Each candidate set $U$ has a good representative point $\theta_U$ such that
$$
|\acn{U}{y} - \rmp({\rm d}(\theta_U,y))| = O(\rho)\,.
$$

At a higher level, this is already enough to give a distance-recovery procedure, ignoring some non-essential technicalities. The final statement in our main theorem that each $v$ has a no-error cluster $U_v$ with $X_{U_v} \subseteq B(X_v,Cr)$ is more than what is needed for distance recovery.

So our main algorithm resolves two issues:  1. handle the fact that the link function $\rmp$ is only bi-Lipschitz on a \localbilipschitzwindow{}, and 2. avoid running an exhaustive search on the whole vertex set when $n$ is larger than the minimum assumption $\sp n \phi(\rho)\rho^2\gtrsim \log n$, improving the running time.

\step{Presence of a local bi-Lipschitz window}
In the case that $\rmp$ is only bi-Lipschitz on a \localbilipschitzwindow{}, one can treat each $U$ as a $(R,c\rho/R)$ cluster in the sense of
Definition~\ref{def:cluster}, meaning that it is mostly contained in a ball of
radius $R$ around $\theta_U$, except for a fraction of $c\rho/R$ of its points.
Such a cluster is in general insufficient for distance recovery, as a portion
of its points may be far from the representative point $\theta_U$ with distance
much larger than the \localbilipschitzwindow{} of the link function $\rmp$. Still,
for a fresh vertex \(v\), the \normalizedaverage{} \(\cn{U}{v}\) concentrates
around $\acn{U}{v}$ by Lemma~\ref{lem:fixed-v-navigation} (and uniformly by
Lemma~\ref{lem:uniform-navigation-event}). For each $U \in {\cal C}_\rho$,
\(\acn{U}{v}\) therefore serves as a proxy for
$\rmp({\rm d}(\theta_U,X_v))$ with error $O(\sqrt{\rho})$, instead of
\(O(\rho)\) in the ``global'' bi-Lipschitz case; see
Remark~\ref{rem:internal-average-to-tau-tau}. In short, there is a tradeoff
between the score error and the fraction of points that are far from
\(\theta_U\).

\begin{center}
\small
\begin{tabular}{@{}p{0.28\textwidth}p{0.36\textwidth}p{0.25\textwidth}@{}}
\toprule
Regime & Applicable \(y\)'s & Error bound \\
\midrule
Global bi-Lipschitz
&
all \(y\in M\)
&
\(O(\rho)\)
\\
Local bi-Lipschitz window
&
\(\D{\theta_U}{y}+O(\sqrt{\rho})\le r_\rmp\)
&
\(O(\sqrt{\rho})\)
\\
\bottomrule
\end{tabular}
\end{center}
\subsubsection*{Ingredient 2: link-average threshold refinement}
The second ingredient is \scorethresholdrefinement{}, which cleans up a seed set using fresh vertices. Thresholding the observed averages \(\cn{U}{v}\) selects vertices whose latent points are near the seed representative; the abstract refinement statement is
Proposition~\ref{prop:abstract-score-threshold-refinement}.

This refinement step turns a seed into a cleaner cluster. A first refinement may
still leave a small exceptional fraction of mistakes, but a second refinement can
be made strong enough to give exact inclusions: all vertices in an inner ball
are selected, and no vertices outside a larger ball are selected. The needed
link-average gap is packaged in the unified
Proposition~\ref{prop:link-average-separation}, whose two cases are the
\oracleroute{} and the \localroute{}.

\subsubsection*{Algorithm outline}
\step{Vertex Partition}
To keep the concentration arguments independent, we apply the two ingredients
on separate vertex blocks. For the basic three-block procedure, we split the
vertices into
\[
V_1,\quad V_2,\quad V_3.
\]
The exhaustive search and \candidatenet{} construction are performed in \(V_1\).
The selected cluster candidates (a subset of ${\cal C}_\rho$) are then refined
into \(V_2\), producing large intermediate clusters. These intermediate
clusters are refined once more into
\(V_3\), which is the final output block. This is formalized in
Theorem~\ref{thm:three-block-extraction} and
Remark~\ref{rem:three-block-procedure}; the \candidatenet{} construction is the
unified Proposition~\ref{prop:candidate-refinement-net}, with route-specific
proofs for the \oracleroute{} and the \localroute{}. At a high level,
\[
\text{internal-average search}
\longrightarrow
\text{candidate net}
\longrightarrow
\text{intermediate refinement}
\longrightarrow
\text{exact refinement}.
\]

The reason for using disjoint blocks is simple: the edges used to construct
a cluster candidate are independent of the edges used to test that cluster
candidate against fresh vertices. This lets us condition on the cluster
candidate and apply concentration to the
next block, which allows us to exploit conditional independence; this is the
conditioning setup used in Proposition~\ref{prop:abstract-score-threshold-refinement}.

\step{Accuracy/efficiency tradeoffs by choosing smaller subsets}
Although the theorem is stated with three ambient blocks of size \(n\), the
first two stages do not need to use all vertices in their blocks.  The expensive
exhaustive search is performed only on a smaller working subset
\[
V_1'\subseteq V_1,\qquad |V_1'|=n_1,
\]
where \(n_1\) is chosen just large enough for the internal-average search
and the first refinement-net construction to succeed. In the \localroute{} this
internal scale is \(\rho\asymp r^2\), while in the \oracleroute{} it is
\(\rho\asymp r\), up to constants depending only on the model parameters. The
working size \(n_1\) is chosen near the smallest value for which the first-block
occupancy and pair-average concentration events hold, with the additional mild
lower bound \(n_1\ge\log\log n\).

Similarly, the first refinement is performed only into a smaller working subset
\[
V_2'\subseteq V_2,\qquad |V_2'|=n_2.
\]
The role of \(V_2'\) is to produce intermediate clusters large enough to seed
the exact refinement into the full output block \(V_3\).  Thus it is enough to
choose \(n_2\) so that
\[
\mathsf{s}n_2\phi(cr)r^2
\gtrsim
\Lambda\log n.
\]
The final refinement is still performed into the full block \(V_3\), so the
output clusters have size of order \(n\phi(r)\).

This separation of working sizes is useful computationally.  In the dense
fixed-scale regime, where \(\mathsf{s}\asymp1\), \(r\asymp1\), and
\(\phi(cr)\asymp1\), one may take
\[
n_1\asymp \log\log n,
\qquad
n_2\asymp \Lambda\log n.
\]
Then the exhaustive search over \(V_1'\) is only polylogarithmic in \(n\), and
the refinement stages cost \(n\,\operatorname{polylog}(n)\).  Thus the same
statistical construction yields a near-linear-time algorithm for fixed-accuracy
recovery in dense regimes.

\step{The local-window difficulty}
If the link function were bi-Lipschitz on all of
\([0,\operatorname{diam}(M)]\), the \internalaveragecandidate{}s would
behave like ordinary small-radius clusters. However, we only assume that
\(\rmp\) is bi-Lipschitz on a \localbilipschitzwindow{} \([0,r_\rmp]\), as in
Assumption~\ref{assump:link-function}. Outside this window,
the link may flatten or lose metric information. As a result, an internal-average
candidate is not automatically safe for all comparisons.

This is where the \windowsafe{} idea enters. A set is \windowsafe{} if its latent
diameter lies safely inside the \localbilipschitzwindow{}
(Definition~\ref{def:window-safe}). Within such a set, the local
bi-Lipschitz assumption behaves like a global one. If a \fuzzywindoworacle{}
is available on the first block (Definition~\ref{def:fuzzy-window-oracle}),
we restrict the exhaustive search to candidates certified by the oracle. This removes the
square-root loss that appears in the \localroute{} and allows
the candidate search to operate at the target scale \(r\), as reflected in the
two routes of Theorem~\ref{thm:three-block-extraction}.

\step{Two-round bootstrap}
The final algorithm constructs the fuzzy window oracle itself; this is the content of
Theorem~\ref{thm:two-round-extraction}. It uses six blocks.
First, we run the \localroute{} of the \threeblockextraction{} at the coarser scale
\(r_0=\sqrt r\) on the first three blocks. This produces exact coarse clusters. These coarse
clusters are then used as local probes: by averaging their edges to vertices in
the fourth block, we certify which vertices lie in a common \localbilipschitzwindow{}.
This gives a \fuzzywindoworacle{} on the fourth block, via
Lemma \ref{lem:coarse-output-generates-window-oracle}.
Then we run the \oracleroute{} of the \threeblockextraction{} at the target scale
\(r\) on the last three blocks (see Theorem \ref{thm:three-block-extraction}). The local-route run at scale \(\sqrt r\) and the oracle-route run at scale \(r\) both use internal candidate scales comparable to \(r\), so their statistical requirements reduce to
\[
\mathsf{s}n\phi(c r)r^2\gtrsim \log n
\]
for a model-dependent constant \(c>0\). The final repacking argument removes this
constant inside \(\phi\) in the main theorem. Thus the square-root loss is paid
only in the preliminary oracle-construction round, not in the final resolution.

\subsection*{Open questions}
We close with several questions left open by the present work.
\begin{enumerate}
\item Can the efficient regime be improved beyond the fixed-accuracy dense
setting considered here?
\item Is there a genuine statistical-computational gap for high-accuracy
distance denoising in general metric measure spaces?
\item It seems the result could be extended to the case where \(\rmp\) is not necessarily bi-Lipschitz, but only satisfies some weaker regularity condition. For example, if \(\rmp\) satisfies a local H\"older condition of the form
$$
\ell_\rmp |t-t'|^\alpha
\le |\rmp(t)-\rmp(t')| \le L_\rmp |t-t'|^\alpha
$$
for some $\alpha \in (0,1]$, it seems that the same higher level strategy might still work both on the upper and lower bounds, with a different transition.
\end{enumerate}

\subsection*{Acknowledgments}
The authors were partially supported by Vannevar Bush Faculty Fellowship ONR-N00014-
20-1-2826 and by Simon Investigator award (622132). E.M. was also partially supported by
ARO MURI W911NF1910217, NSF DMS-2031883, and NSF award CCF 1918421.

Disclosure of AI-assisted work. H.H. used OpenAI Codex and ChatGPT as auxiliary tools while preparing this manuscript from an older draft, where H.H. provided an older statement or a proof sketch and asked Codex to carry out some of the details. All statements, arguments, and proofs were verified and revised by the authors through multiple rounds of editing.

\section{Graph model and standing assumptions}
\label{sec:graph-model-standing-assumptions}
\begin{defi}[Metric measure space / metric probability space]
A \emph{metric measure space} is a triple \((M,{\rm d},\mu)\), where
\((M,{\rm d})\) is a metric space and \(\mu\) is a Borel measure on \(M\).

If in addition
\[
\mu(M)=1,
\]
then \((M,{\rm d},\mu)\) is called a \emph{metric probability space}.
\end{defi}

Throughout, balls are open:
\[
B(x,r):=\{y\in M:{\rm d}(x,y)< r\}.
\]

\begin{defi}[Centered subgaussian norm]
The \emph{centered subgaussian norm} of a random variable \(X\) is
defined as
\[
\|X\|_{\psi_2}
:= \inf \left\{ \sigma > 0 \;:\;
\mathbb{E}\left[\exp\left(\lambda (X - \mathbb{E}X)\right)\right]
\le \exp\left(\frac{\sigma^2 \lambda^2}{2}\right)
\;\; \text{for all } \lambda \in \mathbb{R}
\right\}.
\]

For \(K\in[0,\infty)\), we say that \(X\) is \emph{\(K\)-subgaussian} if
\[
\|X\|_{\psi_2}\le K.
\]
\end{defi}

\begin{defi}[Random graph model]
\label{def:graph_model}
Let \((M,{\rm d},\mu)\) be a metric probability space, and consider a
non-increasing function
\[
\rmp:[0,\infty)\to\mathbb R.
\]
We formulate the model in the non-increasing, similarity-oriented convention;
non-decreasing real-valued observation models are reduced to this case by
multiplying all observations and the link by \(-1\).
Let \(\bfV\) be a finite vertex set of \(n\) vertices, and let
\(\sp = \mathsf{s}_n \in(0,1]\) be a sparsity parameter.

First, let \(\{X_v\}_{v\in\bfV}\) be i.i.d.\ samples of latent points in \(M\)
according to \(\mu\). For each unordered pair \(\{u,v\}\subseteq\bfV\) with
\(u\neq v\), let \(\mathcal U_{u,v}\sim \mathrm{Unif}[0,1]\), independently
over unordered pairs and independently of \(\{X_v\}_{v\in\bfV}\), and set
\(\mathcal U_{v,u}:=\mathcal U_{u,v}\).

Assume there is a measurable function
\[
\mathbf F:[0,\infty)\times[0,1]\to\mathbb R
\]
such that for every \(t\ge0\),
\[
\int_0^1 \mathbf F(t,s)\,ds=\rmp(t),
\]
and such that
\[
\mathbf F(t,\mathcal U)-\rmp(t)
\]
is \(\Ksg\)-subgaussian uniformly in \(t\), where
\(\mathcal U\sim\mathrm{Unif}[0,1]\). Equivalently,
\[
\|\mathbf F(t,\mathcal U)-\rmp(t)\|_{\psi_2}\le \Ksg
\qquad\text{for every }t\ge0.
\]
Thus, for \(u\neq v\),
\[
\widetilde Z_{u,v}
:=
\mathbf F\!\bigl({\rm d}(X_u,X_v),\mathcal U_{u,v}\bigr)
\]
has conditional mean \(\rmp({\rm d}(X_u,X_v))\) and conditional subgaussian
norm at most \(\Ksg\).

For each unordered pair \(\{u,v\}\subseteq\bfV\) with \(u\neq v\), let
\(B_{u,v}\sim\mathrm{Bernoulli}(\sp)\), independently over \(u<v\), and
independently of \(\{X_v\}_{v\in\bfV}\) and \(\{\mathcal U_{u,v}\}_{u<v}\).
Set
\[
B_{v,u}:=B_{u,v},
\qquad
Z_{u,v}:=B_{u,v}\widetilde Z_{u,v}.
\]
Finally, set \(Z_{u,u}=0\).
\end{defi}

Under this model, for \(U\subseteq\bfV\), write
\[
X_U:=(X_u)_{u\in U}
\]
for the corresponding latent point family. All counts of sampled points are
understood with multiplicity; for example,
\[
|\{u\in U:X_u\in A\}|
\]
counts vertices whose latent points lie in \(A\).

\begin{rem}
\label{rem:observed-weighted-graph}
We refer to \(Z_{u,v}\) as the \observededgeweight{} between \(u\) and \(v\).
The collection of these weights is the \observedweightedgraph{}.
In the Bernoulli soft random geometric graph special case, this observed edge
weight is the edge indicator.
\end{rem}

\begin{rem}
When \(0\le \rmp\le 1\) and \(\mathbf F(t,s) = {\bf 1}\{s \le \rmp(t)\}\), the model reduces to a sparse soft random
geometric graph with edge probability
\(\sp\rmp({\rm d}(X_u,X_v))\). The formulation above allows more general
weighted subgaussian observations.
\end{rem}

\begin{defi}[Lower \(\phi\)-regularity]
\label{def:lower-phi-regularity}
Let \(\mu\) be a Borel probability measure on a metric space
\((M,{\rm d})\), and define
\[
\mu_{\min}(r) := \inf_{p \in \operatorname{supp}(\mu)} \mu(B(p,r)).
\]
Let \(I\) be either \([0,r_*)\) for some \(r_*>0\), or \((0,\infty)\), and
let \(\phi:I\to[0,\infty)\) be nonnegative and nondecreasing. We say that
\((M,{\rm d},\mu)\) satisfies \lowerphiregularity{} on \(I\) if
\[
\mu_{\min}(r)\ge \phi(r)
\qquad\text{for all } r\in I.
\]
When \(I=(0,r_\mu]\), we also say that \((M,{\rm d},\mu)\) is
\lowerphiregularity{} up to scale \(r_\mu\). When
\(I=(0,\infty)\), we simply say that it is \emph{lower \(\phi\)-regular}.
Equivalently, throughout the extraction argument one may replace \(M\) by
\(\operatorname{supp}(\mu)\).
\end{defi}

\begin{assump}[Support convention]
\label{assump:support-convention}
We assume, after replacing the ambient metric space by
\(\operatorname{supp}(\mu)\) if necessary, that
\[
M=\operatorname{supp}(\mu).
\]
\end{assump}

\begin{assump}[Lower regularity of the measure]
\label{assump:lower-regularity}
Let \(\mu\) be a Borel probability measure on a metric space \((M,{\rm d})\).
We assume that there exist \(r_{\mu}>0\) and a nonnegative nondecreasing
function \(\phi:(0,r_\mu]\to[0,1]\) such that \((M,{\rm d},\mu)\) is
\lowerphiregularity{} up to scale \(r_\mu\).
\end{assump}

\begin{assump}[Local link function]
\label{assump:link-function}
Let \(\rmp:[0,\infty)\to\mathbb R\) be a non-increasing function. We assume
there exist positive constants \(L_{\rmp}, \ell_\rmp, r_\rmp>0\) such that
for all \(a,b\in[0,r_\rmp]\),
\[
\ell_\rmp |a-b|\le |\rmp(a)-\rmp(b)|\le L_\rmp |a-b|.
\]
This is the \locallinkassumption{}; the interval \([0,r_\rmp]\) is the
\localbilipschitzwindow{}. Given a metric space
\((M,{\rm d})\), we also assume that
\[
M_{\rmp}:=\sup_{t\in[0,{\rm diam}(M)]}|\rmp(t)|<+\infty.
\]
\end{assump}

\begin{assump}[Standing model assumptions]
\label{assump:graph_model}
Throughout this paper, we consider the random graph model defined in
Definition~\ref{def:graph_model} with link function \(\rmp\) satisfying
Assumption~\ref{assump:link-function} and measure \(\mu\) satisfying
Assumptions~\ref{assump:support-convention} and
\ref{assump:lower-regularity}. We may consider the same model with different
vertex-set sizes and sparsity parameters \(\sp\).
\end{assump}

\section{Basic definitions, events, and concentration}
\label{sec:basic-definitions-events-concentration}
Here we define several basic events and estimates that will be used repeatedly in the main construction. These estimates are standard consequences of occupancy bounds, Chernoff--Bernstein inequalities, subgaussian concentration, and union bounds. To keep the main construction readable, their proofs are postponed to Appendix~\ref{app:auxiliary-probability-proofs}.

Here we introduce a global large parameter
\begin{equation}
    \label{def:global-large-parameter}
\Lambda \ge 1,
\end{equation}
whose value will be determined by the needs of the main construction. The parameter \(\Lambda\) is used to control the probability of failure of various events.
Unless explicitly stated otherwise, throughout the rest of this paper we work
under Assumption~\ref{assump:graph_model}. In particular, the latent points
\((X_v)_{v\in\bfV}\) are sampled from a metric probability space
\((M,{\rm d},\mu)\) satisfying \lowerphiregularity{} with
\(M=\operatorname{supp}(\mu)\), the \observededgeweight{}s \(Z_{u,v}\) are generated by
Definition~\ref{def:graph_model}, and the link function \(\rmp\) satisfies the
\locallinkassumption{} of Assumption~\ref{assump:link-function}.

\subsection{Occupancy event}
The first event is a uniform lower occupancy event, which ensures that every latent ball contains a sufficiently large number of sampled points. This is a consequence of the \lowerphiregularity{} assumption and standard occupancy bounds.

\begin{defi}[Lower occupancy event]
\begin{equation}
\label{eq:def-lower-occupancy-event}
\Ept{W}{r}
:=
\left\{
|\{v\in W:X_v\in B(x,r)\}|
\ge
\frac{|W|\phi(r/3)}{2}
\text{ for every }x\in M
\right\}.
\end{equation}
We call \(\Ept{W}{r}\) the \loweroccupancyevent{}.
\end{defi}

\begin{lemma}[Uniform lower occupancy]
\label{lem:uniform-lower-occupancy}
Let \((M,{\rm d},\mu)\) satisfy \lowerphiregularity{} up to scale \(r_\mu\).
Thus \(\phi\) is nonnegative and nondecreasing, as in
Definition~\ref{def:lower-phi-regularity}.
Let \(W\) be a finite vertex set, assume \(n_\star:=|W| \ge 2\), and let
\(\{X_v\}_{v\in W}\) be i.i.d.\ samples from \(\mu\). If
\[
r\in(0,r_\mu],
\qquad
\phi(r/6)\ge \Lambda \frac{\log n_\star}{n_\star},
\]
then
\[
\mathbb P\bigl(\Ept{W}{r}\bigr)\ge 1-\exp\left(
     - \tfrac{1}{16}\Lambda \log(n_\star)
\right)\,,
\]
provided $\Lambda$ is greater than some universal constant.
\end{lemma}

\subsection{Clusters and empirical averages}
We first record the basic cluster notion and the empirical averages used to test clusters against vertices.

\begin{defi}[\((r,\eta)\)-cluster]
\label{def:cluster}
Assume the graph model of Definition~\ref{def:graph_model}. Let \(r>0\) and
\(\eta\in[0,1]\). A nonempty subset \(U\subseteq\bfV\) is called an
\emph{\((r,\eta)\)-cluster with center \(x\in M\)} if
\[
\left|\left\{u\in U:X_u\in B(x,r)\right\}\right|
\ge
(1-\eta)|U|.
\]
\end{defi}

\begin{defi}[Latent and empirical link averages]
\label{def:normalized-averages}
\label{def:latent-empirical-link-averages}
For a nonempty set \(U\subseteq\bfV\), define its \emph{latent link average}
\begin{equation}
\label{eq:def-latent-link-average}
\acn{U}{y}
:=
\frac1{|U|}
\sum_{u\in U}\rmp\bigl({\rm d}(X_u,y)\bigr),
\qquad y\in M.
\end{equation}
When the argument is a vertex \(v\), we write, by abuse of notation,
\[
\acn{U}{v}:=\acn{U}{X_v}.
\]
For \(v\in\bfV\setminus U\), define the \emph{empirical link average}
\begin{equation}
\label{eq:def-normalized-average}
\cn{U}{v}
:=
\frac{\sum_{u\in U}Z_{u,v}}{\sp |U|},
\end{equation}
so that
\[
\mathbb E\bigl[\cn{U}{v}\mid X_v,X_U\bigr]=\acn{U}{v}.
\]
Thus \(\cn{U}{v}\) is the empirical link average from \(U\) to \(v\), and
\(\acn{U}{v}\) is its conditional mean.
\end{defi}

\subsection{Link-average concentration events}
The key object is the \empiricallinkaverage{} \(\cn{U}{v}\) from a set \(U\)
to a vertex \(v\). The link-average concentration events ensure that
\(\cn{U}{v}\) concentrates around its conditional mean \(\acn{U}{v}\). Thus,
if \(U\) is a cluster with center \(x\), then \(\cn{U}{v}\) is a good proxy for
\(\rmp({\rm d}(x,X_v))\), which is the key to distance recovery from clusters.

\begin{lemma}[Fluctuation of empirical link averages for a fixed vertex]
\label{lem:fixed-v-navigation}
Consider the graph model of Assumption~\ref{assump:graph_model}. Let
\(U\subseteq\bfV\) be nonempty, let \(v\in\bfV\setminus U\), and set
\(m:=|U|\). There exist constants \(c_{\rm link},C_{\rm link}>0\), depending only on \(M_\rmp\) and
\(\Ksg\), such that for every fixed realization \(X_U=x_U\) and \(X_v=x_v\),
\[
\mathbb P\!\left(
\left|\cn{U}{v}-\acn{U}{v}\right|>t
\ \middle|\ X_U=x_U,\ X_v=x_v
\right)
\le
C_{\rm link}\exp\!\bigl(-c_{\rm link}\,\sp m\,\min\{t^2,1\}\bigr)
\qquad\text{for all }t>0.
\]
\end{lemma}

\begin{defi}[Link-average concentration event]
We associate to \(\cn{U}{v}\) the fluctuation scale
\begin{equation}
\label{eq:def-fluctuation-scale}
\epsnav{U}{n_\star}
:=
\sqrt{\Lambda \frac{\log n_\star}{\sp |U|}},
\qquad n_\star\ge2,
\end{equation}
where \(n_\star\) is a reference size parameter for logarithmic union bounds.

\begin{equation}
\label{eq:def-navigation-event}
\Enavi{U}{V}{n_\star}
:=
\left\{
\forall v\in V:
\left|\cn{U}{v}-\acn{U}{v}\right|
\le
\epsnav{U}{n_\star}
\right\}.
\end{equation}
We call \(\Enavi{U}{V}{n_\star}\) the \navigationevent{}.
\end{defi}

\begin{lemma}[Uniform link-average concentration]
\label{lem:uniform-navigation-event}
Consider the graph model of Assumption~\ref{assump:graph_model}. Let
\(U,V\subseteq\bfV\) be disjoint, with \(U\) nonempty. Let \(n_\star\ge2\) be a
reference size parameter. Assume
\[
\sp |U|\ge \Lambda \log n_\star,
\qquad
|V|\le n_\star\,.
\]
 Then for all fixed realizations
\(X_U=x_U\) and \(X_V=x_V\),
\[
\mathbb P\!\left(
\Enavi{U}{V}{n_\star}^{\,c}
\ \middle|\ X_U=x_U,\ X_V=x_V
\right)
\le
\exp( - \tfrac{1}{2}c_{\rm link}\,\Lambda \log n_\star)\,,
\]
provided that $\Lambda$ is greater than some universal constant depending on $C_{\rm link}, c_{\rm link}$ (and thus depends on $M_\rmp$ and $\Ksg$).
\end{lemma}

\subsection{Pair averages}
The next definition generalizes the empirical link average: it averages observed edge weights over pairs of vertices in two sets \(U\) and \(V\). The corresponding expected value is an average of the link function over pairs of latent points in \(X_U\) and \(X_V\).

\begin{defi}[Pair averages]
\label{def:pair-averages}
For subsets \(U,V\subseteq\bfV\), define the ordered off-diagonal pair set
\[
\mathcal D(U,V)
:=
\{(u,v)\in U\times V:\ u\neq v\}.
\]
Then
\[
|\mathcal D(U,V)|=|U||V|-|U\cap V|.
\]
When \(|\mathcal D(U,V)|>0\), define
\[
\pav{U,V}
:=
\frac{1}{\sp |\mathcal D(U,V)|}
\sum_{(u,v)\in\mathcal D(U,V)}Z_{u,v},
\quad\mbox{and}\quad
\apav{U,V}
:=
\frac{1}{|\mathcal D(U,V)|}
\sum_{(u,v)\in\mathcal D(U,V)}
\rmp\!\bigl(\D{X_u}{X_v}\bigr).
\]
With this convention,
\[
\mathbb E\!\left[\pav{U,V}\mid X_{\bfV}\right]=\apav{U,V}.
\]
For \(U=V\) with \(|U|\ge2\), write
\[
\pav{U}:=\pav{U,U},
\qquad
\apav{U}:=\apav{U,U}.
\]
We call \(\pav{U,V}\) the \pairaverage{} between \(U\) and \(V\).
Since \(\mathcal D(U,U)\) consists of ordered off-diagonal pairs, each
unordered pair is counted twice. This agrees with the usual unordered average
because \(Z_{u,v}=Z_{v,u}\) and \(\D{X_u}{X_v}=\D{X_v}{X_u}\).
\end{defi}

\begin{defi}[Uniform pair-average event]
\begin{equation}
\label{eq:def-pair-average-event}
\Enn{W}{\lambda}{m}
:=
\left\{
\begin{array}{l}
\text{for all }U_1,U_2\subseteq W\text{ with }|U_1|\ge m,\ |U_2|\ge m,\\[0.3ex]
\left|\pav{U_1,U_2}-\apav{U_1,U_2}\right|\le\lambda
\end{array}
\right\}.
\end{equation}
We call \(\Enn{W}{\lambda}{m}\) the \uniformpairaverageevent{}.
\end{defi}

\begin{rem}
The event \(\Enn{W}{\lambda}{m}\) allows \(U_1\) and \(U_2\) to overlap. The
diagonal is omitted by \(\mathcal D(U_1,U_2)\), and the appendix proof handles
the resulting duplicate appearances of unordered edges by assigning
multiplicities in \(\{0,1,2\}\). The union bound is taken over all cardinalities
\(|U_1|=a\), \(|U_2|=b\) with \(a,b\ge m\).
\end{rem}

\begin{lemma}[Uniform concentration of pair averages]
\label{lem:uniform-pair-average}
Let \(W\) be a vertex set, and set \(n_\star:=|W|\). Fix
\(\varphi\in(0,1)\), set
\[
m:=\lceil\varphi n_\star\rceil,
\]
and assume \(m\ge2\). There exist constants \(c_{\tref{lem:uniform-pair-average}},C_{\tref{lem:uniform-pair-average}}>0\), depending only on
\(\Ksg\) and \(M_\rmp\), such that for every \(\lambda\in(0,1)\), if
\[
\sp n_\star\varphi\lambda^2
\ge
C_{\tref{lem:uniform-pair-average}}\log(e/\varphi),
\]
then, conditionally on \(X_W\),
\[
\mathbb P\!\left(
\Enn{W}{\lambda}{m}^{\,c}
\ \middle|\ X_W
\right)
\le
\exp\{-c_{\tref{lem:uniform-pair-average}}\varphi n_\star\log(e/\varphi)\}
\]
\end{lemma}
\begin{rem}
\label{rem:uniform-pair-average-condition}
Later we will often verify the stronger sufficient condition
\[
\sp \frac{n_\star}{\Lambda\log n_\star}\varphi\lambda^2
\ge
C_{\tref{lem:uniform-pair-average}}.
\]
Indeed, since \(\Lambda\ge1\), this implies the same condition with
\(\Lambda\) removed from the denominator. For \(a,b\ge e\) and \(k>0\),
\[
a\ge b\log^k a
\quad\Longrightarrow\quad
a\ge b\log^k b,
\]
because \(\log(a)\ge1\) first gives \(a\ge b\), and hence
\(\log(a)\ge\log(b)\). Applying this with
\[
a=n_\star,
\qquad
b=\frac{C_{\tref{lem:uniform-pair-average}}}{\sp\varphi\lambda^2},
\qquad
k=1,
\]
we obtain
\[
\sp n_\star\varphi\lambda^2
\ge
C_{\tref{lem:uniform-pair-average}}
\log\!\left(
\frac{C_{\tref{lem:uniform-pair-average}}}{\sp\varphi\lambda^2}
\right)
\ge
C_{\tref{lem:uniform-pair-average}}\log(e/\varphi).
\]
The last inequality follows after increasing
\(C_{\tref{lem:uniform-pair-average}}\) if necessary, using
\(\sp\le1\) and \(\lambda\in(0,1)\).
\end{rem}

\subsection{Window-safe sets and fuzzy window oracles}
The next definition is the notion of \windowsafe{} sets, which are sets whose latent diameter lies safely inside the \localbilipschitzwindow{}. This ensures that the local bi-Lipschitz assumption behaves like a global one on such sets. The second part of the definition is the notion of a \fuzzywindoworacle{}, which is an oracle that certifies pairs of vertices whose latent points lie within a smaller window. This allows us to restrict the candidate search to certified sets, which removes the square-root loss in the main construction.

\begin{defi}[Window-safe sets and fuzzy window oracles]
\label{def:window-safe}
\label{def:fuzzy-window-oracle}
For a possibly multiset \(A\subseteq M\), define the latent diameter
\[
\operatorname{diam}(A)
:=
\sup_{p,p'\in A}\D{p}{p'}.
\]
For \(\lambda_{\win}\in(0,1]\), we say a multiset \(A\) of \(M\) is
\(\lambda_{\win}\)-\windowsafe{} if
\[
\operatorname{diam}(A)\le \lambda_{\win}r_\rmp.
\]
By abuse of notation, a vertex set \(U\subseteq\bfV\) is
\(\lambda_{\win}\)-\windowsafe{} when the corresponding latent multiset
\(X_U\) has that property. Similarly, \((U,y)\) is
\(\lambda_{\win}\)-\windowsafe{} when \(X_U\cup\{y\}\) has that property.
The case \(\lambda_{\win}=1\) means that the relevant latent diameter lies
inside the local bi-Lipschitz window.

For \(0<\alpha_{\win}<\lambda_{\win}<1\),
a map
\[
\mathcal O_{\win}:\bfV\times\bfV\to\{0,1\}
\]
is called an \((\alpha_{\win},\lambda_{\win})\)-\fuzzywindoworacle{}
if, for all \(u,v\in\bfV\),
\[
\D{X_u}{X_v}\le \alpha_{\win}r_\rmp
\quad\Longrightarrow\quad
\mathcal O_{\win}(u,v)=1,
\]
and
\[
\D{X_u}{X_v}>\lambda_{\win}r_\rmp
\quad\Longrightarrow\quad
\mathcal O_{\win}(u,v)=0.
\]
We say that the oracle certifies a set \(U\subseteq\bfV\) if
\[
\mathcal O_{\win}(u,v)=1
\qquad\text{for every }u,v\in U.
\]
\end{defi}

\begin{rem}
If a fuzzy window oracle certifies \(U\), then \(U\) is
\(\lambda_{\win}\)-\windowsafe{}. Conversely, if
\[
\operatorname{diam}(X_U)\le \alpha_{\win}r_\rmp,
\]
then the oracle certifies \(U\). Similarly, if
\[
\operatorname{diam}(X_U\cup X_V)\le \alpha_{\win}r_\rmp,
\]
then all pairs between \(U\) and \(V\) are certified by the oracle.
\end{rem}

\section{Internal-average candidates and link-average separation}
\label{sec:internal-average-link-separation}
\label{sec:internal-average-cluster-candidates}

We first introduce the notion of an \internalaveragecandidate{}, which is a
set \(U\) whose internal pair average \(\apav{U}\) is close to the maximum
possible value \(\rmp(0)\). These sets can be extracted from the observed
graph, and the key point is that such a set behaves like a cluster with a
center \(\theta_U\).

We then introduce
\linkaverageseparation{}, the deterministic condition used by the refinement
step to separate points close to the center \(\theta_U\) from points that are
far away.

\begin{defi}[Internal-average candidate]
\label{def:internal-average-candidate}
Let \(U\subseteq\bfV\) and \(\Delta\ge0\). We say that \(U\) is a
\(\Delta\)-\keyterm{internal-average candidate} if
\[
\apav{U}\ge \rmp(0)-\Delta.
\]
\end{defi}

\begin{defi}[Average distance]
\label{def:average-distance}
\label{def:cluster-score}
For a nonempty set \(U\subseteq\bfV\) and a point \(x\in M\), define
\[
\bar d_U(x)
:=
\frac1{|U|}\sum_{u\in U}\D{X_u}{x}.
\]
\end{defi}

\begin{lemma}[Internal average gap gives a center representative]
\label{lem:internal-gap-average-distance}
Let \(U\subseteq\bfV\) with \(|U|\ge2\), and suppose
\[
\apav{U}\ge \rmp(0)-\Delta.
\]
Then there exists $u_U \in U$ such that $\theta_U := X_{u_U} \in M$ satisfies
\[
    \tfrac{1}{|U|}\sum_{u \in U}\rmp(\D{X_u}{\theta_U}) \ge \rmp(0)-\Delta,
\]
and
\[
\left|\{u\in U:X_u\notin B(\theta_U,R)\}\right|
\le
\frac{\Delta}{\ell_\rmp R}|U| \qquad\text{for every } 0 < R \le r_\rmp.
\]
In addition, if \(U\) is \(1\)-\windowsafe{}, then the same tail bound holds for
every \(R>0\), and moreover
\[
\bar d_U(\theta_U)\le \frac{\Delta}{\ell_\rmp}.
\]
\end{lemma}

\begin{proof}
\step{Good Center}
Set \(m:=|U|\). Since \(\apav{U}\) is the ordered off-diagonal average,
\[
\rmp(0)-\apav{U}
=
\frac{1}{m(m-1)}
\sum_{\substack{u,v\in U\\u\neq v}}
\left[
\rmp(0)-\rmp\!\bigl(\D{X_u}{X_v}\bigr)
\right].
\]
Averaging over the first coordinate, there exists $u_U\in U$ such that
\begin{align*}
\frac{1}{m-1}\sum_{\substack{v\in U\\v\neq u_U}}
\left[\rmp(0)-\rmp\!\bigl(\D{X_{u_U}}{X_v}\bigr)\right]
\le
\rmp(0)-\apav{U}
\le
\Delta.
\end{align*}
Let us simply set $\theta_U:=X_{u_U}$, so that the above average is taken over the distances from $\theta_U$ to the other points in $U$.
Since the term \(u=u_U\) contributes \(\rmp(0)\), this also gives
\[
\frac1m\sum_{u\in U}\rmp(\D{X_u}{\theta_U})
\ge
\rmp(0)-\Delta.
\]
Now relying on the lower Lipschitz bound, we have
$$
    \rmp(0)-\rmp\!\bigl(\D{\theta_U}{X_v}\bigr)
\ge
    \rmp(0) - \rmp\bigl( \min\{\D{\theta_U}{X_v}, r_\rmp\}\bigr)
\ge
    \ell_\rmp \min\{\D{\theta_U}{X_v}, r_\rmp\}\,.
$$
And hence,
\begin{align}
    \label{eq:truncated-average}
    \frac{1}{m-1}\sum_{\substack{v\in U\\v\neq u_U}}
    \min\{\D{\theta_U}{X_v}, r_\rmp\} \le \tfrac{\Delta}{\ell_\rmp}\,.
\end{align}

\step{Markov's inequality}
From the above truncated average, we can deduce a tail bound as long as \(R\le r_\rmp\) via Markov's inequality:
\[
\left|\{u\in U:X_u\notin B(\theta_U,R)\}\right|R
\le
\sum_{u\in U}\min\{\D{X_u}{\theta_U},r_\rmp\}
\le
\sum_{\substack{u\in U\\u\neq u_U}}\min\{\D{X_u}{\theta_U},r_\rmp\}
\le
	|U| \cdot \frac{\Delta}{\ell_\rmp}.
\]

\step{Average radius under \(1\)-window-safety}
Assume that \(U\) is \(1\)-\windowsafe{}. Then every distance
\(\D{\theta_U}{X_v}\) lies in \([0,r_\rmp]\), so the truncated average bound \eqref{eq:truncated-average} is actually an average radius bound.
Therefore
\[
\bar d_U(\theta_U)
=
\frac1m\sum_{v\in U}\D{X_v}{\theta_U}
\le
\frac{m-1}{m}\frac{\Delta}{\ell_\rmp}
\le
\frac{\Delta}{\ell_\rmp}.
\]
\end{proof}

\begin{rem}
\label{rem:internal-average-candidate-cluster}
Lemma~\ref{lem:internal-gap-average-distance} implies that every
\(\Delta\)-\internalaveragecandidate{} \(U\) is an
\[
\left(R,\frac{\Delta}{\ell_\rmp R}\right)\text{-cluster}
\]
with center \(\theta_U\), for every \(0<R\le r_\rmp\).
\end{rem}

The refinement step only needs one deterministic property: the latent link
average \(\acn{U}{\cdot}\) must separate an inner ball from the complement of a
larger ball. We package that property directly.

\begin{defi}[Link-average separation]
\label{def:link-average-separated-seed}
Let \(U\subseteq\bfV\), \(x\in M\), \(0<R_{\rm in}<R_{\rm out}\),
\(\Theta\in\mathbb R\), and \(\gamma>0\). We say that \(U\) is
\[
(x,R_{\rm in},R_{\rm out},\Theta,\gamma)\text{-link-average separated}
\]
if
\[
\D{y}{x}\le R_{\rm in}
\quad\Longrightarrow\quad
\acn{U}{y}\ge \Theta+\gamma,
\]
and
\[
\D{y}{x}\ge R_{\rm out}
\quad\Longrightarrow\quad
\acn{U}{y}\le \Theta-\gamma.
\]
\end{defi}

For an \internalaveragecandidate{} \(U\), the drop of \(\acn{U}{y}\)
from its near-maximal value provides a statistic for how far \(y\) lies from
the representative point \(\theta_U\).

\begin{rem}[Oracle Route and Local Route]
\label{rem:oracle-local-routes}
The \keyterm{Oracle Route} assumes that a \fuzzywindoworacle{} is available and
uses it to restrict the candidate search to oracle-certified, window-safe
candidates. The \keyterm{Local Route} uses no oracle; it treats
internal-average candidates as approximate localized clusters before applying
the same refinement mechanism. We use \({\rm(O)}\) and \(\ora\) for the former,
and \({\rm(L)}\) and \(\loc\) for the latter.
\end{rem}

Below we give a unified statement of \linkaverageseparation{} for the two
routes.

\begin{prop}[Link-average separation]
\label{prop:link-average-separation}
There exists a constant
\[
A_{\tref{prop:link-average-separation}}\ge1,
\]
depending only on \(L_\rmp,\ell_\rmp,M_\rmp\), such that the following holds.
Let \(K\ge1\) be a parameter satisfying
\[
K
\ge
1+\frac{2(A_{\tref{prop:link-average-separation}}+2)}{\ell_\rmp},
\]
and let $r>0$ satisfy
$$
Kr
\le
\tfrac{1}{2}\min\{1,r_\rmp\}.
$$
Assume that \(\widehat{\rmp}(0)\in\mathbb R\) satisfies
\[
|\widehat{\rmp}(0)-\rmp(0)|\le r.
\]
Consider either of the following two cases:
\begin{itemize}
    \item \ORA{}: Suppose that \(U\subseteq\bfV\) is
    \(\lambda_{\win}\)-\windowsafe{} with \(\lambda_{\win}<1/4\), and an
    \(r\)-\internalaveragecandidate{} with representative \(\theta_U\).
    \item \LOC{}: Suppose that \(U\subseteq\bfV\) is an \((r,r)\)-cluster with
    center \(\theta_U\).
\end{itemize}
Then, in either case, \(U\) is
$$
\left(
\theta_U,r,Kr,
\widehat{\rmp}(0)-A_{\tref{prop:link-average-separation}}r,r
\right)\text{-link-average separated}.
$$
\end{prop}

\begin{proof}
This is just the combination of the two route-specific propositions below (Propositions~\ref{prop:window-safe-internal-average-separated-seed} and \ref{prop:localized-cluster-separated-seed}).
Choose \(A_{\tref{prop:link-average-separation}}\) so that
\[
A_{\tref{prop:link-average-separation}}
\ge
\max\left\{
L_\rmp+\frac{L_\rmp}{\ell_\rmp}+2,
2(L_\rmp+M_\rmp)+2
\right\}.
\]
The lower bound above on \(K\) implies the \(K\)-requirements in
Propositions~\ref{prop:window-safe-internal-average-separated-seed} and
\ref{prop:localized-cluster-separated-seed} when their \(A\)-parameters are set
equal to \(A_{\tref{prop:link-average-separation}}\). The scale assumption
\(Kr\le\tfrac12\min\{1,r_\rmp\}\) implies the scale assumption in either
route-specific proposition.

In the \ORA{} case, apply
Proposition~\ref{prop:window-safe-internal-average-separated-seed}. In the
\LOC{} case, apply Proposition~\ref{prop:localized-cluster-separated-seed} with
\(\tau=r\). In both cases, use the same \(A\)-parameter
\(A_{\tref{prop:link-average-separation}}\) and the same outer-radius parameter
\(K\). This gives the claimed \linkaverageseparation{}.
\end{proof}
\subsection{\texorpdfstring{\(1\)-window-safe}{1-window-safe} upgrade and window-safe separation}

\begin{prop}[Window-safe internal-average candidates give link-average separation]
\label{prop:window-safe-internal-average-separated-seed}
Assume \(\lambda_{\win}<1/4\).
Let
$$
    A_{\tref{prop:window-safe-internal-average-separated-seed}}
    \ge L_\rmp+\frac{L_\rmp}{\ell_\rmp}+2
\quad\mbox{and}\quad
K_{\tref{prop:window-safe-internal-average-separated-seed}}
\ge
\frac{2A_{\tref{prop:window-safe-internal-average-separated-seed}}}{\ell_\rmp}
\ge 1,
$$
and let \(r>0\) be small enough so that
$$
    K_{\tref{prop:window-safe-internal-average-separated-seed}}r
    \le  \tfrac{1}{2}r_\rmp.
$$
Let \(U\subseteq\bfV\) be a \(\lambda_{\win}\)-\windowsafe{},
\(r\)-\internalaveragecandidate{} with representative \(\theta_U\). Suppose that
\(\widehat{\rmp}(0)\in\mathbb R\) satisfies
\[
|\widehat{\rmp}(0)-\rmp(0)|\le r.
\]
Then \(U\) is
\[
\bigl(\theta_U,r,K_{\tref{prop:window-safe-internal-average-separated-seed}}r,
\widehat{\rmp}(0)-A_{\tref{prop:window-safe-internal-average-separated-seed}}r,r\bigr)
\text{-link-average separated}.
\]
\end{prop}
\begin{rem}
The condition
\(2A_{\tref{prop:window-safe-internal-average-separated-seed}}/\ell_\rmp\ge1\)
is automatic since
\(L_\rmp\ge\ell_\rmp\). We include it to make clear that
\(K_{\tref{prop:window-safe-internal-average-separated-seed}}r\ge r\).
\end{rem}

\begin{proof}
Since \(U\) is \(\lambda_{\win}\)-\windowsafe{} and \(\theta_U\in X_U\),
every \(u\in U\) satisfies
\[
\D{X_u}{\theta_U}\le \lambda_{\win}r_\rmp.
\]
Also, \(U\) is \(1\)-\windowsafe{}, so
Lemma~\ref{lem:internal-gap-average-distance} gives
\[
\bar d_U(\theta_U)\le \frac{r}{\ell_\rmp}.
\]

If \(\D{y}{\theta_U}\le r\), then for every \(u\in U\),
\[
\D{X_u}{y}
\le
\D{X_u}{\theta_U}+\D{\theta_U}{y}
\le
\lambda_{\win}r_\rmp+r
\le r_\rmp.
\]
The local Lipschitz bound and the triangle inequality give
\[
\acn{U}{y}
\ge
\rmp(\D{\theta_U}{y})-L_\rmp\bar d_U(\theta_U)
\ge
\rmp(0)-L_\rmp r-\frac{L_\rmp}{\ell_\rmp}r
\ge
\widehat{\rmp}(0)-L_\rmp r-\frac{L_\rmp}{\ell_\rmp}r-r
\ge
\widehat{\rmp}(0)-A_{\tref{prop:window-safe-internal-average-separated-seed}}r+r.
\]

Now suppose
\(\D{y}{\theta_U}\ge K_{\tref{prop:window-safe-internal-average-separated-seed}}r\).
If
\(\D{y}{\theta_U}\le (1-\lambda_{\win})r_\rmp\), then the same
link-average approximation gives
\[
\acn{U}{y}
\le
\rmp(\D{\theta_U}{y})+L_\rmp\bar d_U(\theta_U)
\le
\rmp(0)-\ell_\rmp K_{\tref{prop:window-safe-internal-average-separated-seed}}r
+\frac{L_\rmp}{\ell_\rmp}r
\le
\widehat{\rmp}(0)
-\ell_\rmp K_{\tref{prop:window-safe-internal-average-separated-seed}}r
+\frac{L_\rmp}{\ell_\rmp}r
+r.
\]
If instead \(\D{y}{\theta_U}>(1-\lambda_{\win})r_\rmp\), then every
\(u\in U\) satisfies
\[
\D{X_u}{y}
\ge
\D{y}{\theta_U}-\D{X_u}{\theta_U}
>
(1-2\lambda_{\win})r_\rmp.
\]
By monotonicity and the local lower Lipschitz bound,
\[
\acn{U}{y}
\le
\rmp(0)-\ell_\rmp(1-2\lambda_{\win})r_\rmp
\le
\rmp(0)-\ell_\rmp K_{\tref{prop:window-safe-internal-average-separated-seed}}r.
\]
Since \(\rmp(0)\le\widehat{\rmp}(0)+r\), both cases give
\[
\acn{U}{y}
\le
\widehat{\rmp}(0)-\ell_\rmp K_{\tref{prop:window-safe-internal-average-separated-seed}}r
+\frac{L_\rmp}{\ell_\rmp}r+r.
\]
The choices of
\(A_{\tref{prop:window-safe-internal-average-separated-seed}}\) and
\(K_{\tref{prop:window-safe-internal-average-separated-seed}}\) imply
\[
\ell_\rmp K_{\tref{prop:window-safe-internal-average-separated-seed}}
\ge
2A_{\tref{prop:window-safe-internal-average-separated-seed}}
\ge
A_{\tref{prop:window-safe-internal-average-separated-seed}}
+\frac{L_\rmp}{\ell_\rmp}+2,
\]
and hence
\[
\acn{U}{y}
\le
\widehat{\rmp}(0)-A_{\tref{prop:window-safe-internal-average-separated-seed}}r-r.
\]
This is exactly the claimed \linkaverageseparation{}.
\end{proof}

\subsection{Local bi-Lipschitzness case}

\begin{rem}[From internal average to \((\tau,\tau)\)-localization]
\label{rem:internal-average-to-tau-tau}
Let \(U\) be a \(\Delta\)-\internalaveragecandidate{}, and let
\(\theta_U\) be the representative point from
Lemma~\ref{lem:internal-gap-average-distance}. If
\[
\sqrt{\frac{\Delta}{\ell_\rmp}}\le \tau\le \min\{1,r_\rmp\},
\]
then \(U\) is a \((\tau,\tau)\)-cluster with center \(\theta_U\). Indeed,
Remark~\ref{rem:internal-average-candidate-cluster} gives
\[
\frac{|\{u\in U:X_u\notin B(\theta_U,\tau)\}|}{|U|}
\le
\frac{\Delta}{\ell_\rmp \tau}
\le
\tau.
\]
Equivalently, for every \(C\ge1\), a \(\Delta\)-\internalaveragecandidate{} is a
\[
\left(
C\sqrt{\frac{\Delta}{\ell_\rmp}},
C\sqrt{\frac{\Delta}{\ell_\rmp}}
\right)\text{-cluster}
\]
provided \(C\sqrt{\Delta/\ell_\rmp}\le \min\{1,r_\rmp\}\). Thus the
estimates below, which use the \localbilipschitzwindow{}, may be applied to
\internalaveragecandidate{}s at the cost of passing to the square-root scale. If an estimate also uses an outer
radius \(K\tau\), one must additionally require \((K-1)\tau\le r_\rmp\).
\end{rem}

We now record the \linkaverageseparation{} estimate for genuine
\((\tau,\tau)\)-clusters in the same threshold form as the window-safe case.
The corresponding pair-average estimate is used in
Section~\ref{sec:refinement-nets}.

\begin{prop}[Approximate localized clusters give link-average separation]
\label{lem:score-gap-tau-cluster-local}
\label{prop:localized-cluster-separated-seed}
Let
\[
A_{\tref{prop:localized-cluster-separated-seed}}
\ge 2(L_\rmp+M_\rmp)+2,
\qquad
K_{\tref{prop:localized-cluster-separated-seed}}
\ge
1+\frac{2(A_{\tref{prop:localized-cluster-separated-seed}}+2)}{\ell_\rmp}.
\]
Let \(U\subseteq\bfV\) be a \((\tau,\tau)\)-cluster with center \(x\), where
\[
K_{\tref{prop:localized-cluster-separated-seed}}\tau
\le \frac12\min\{1,r_\rmp\}.
\]
Let \(\widehat{\rmp}(0)\in\mathbb R\) satisfy
\[
|\widehat{\rmp}(0)-\rmp(0)|\le \tau.
\]
Then \(U\) is
\[
\bigl(x,\tau,K_{\tref{prop:localized-cluster-separated-seed}}\tau,
\widehat{\rmp}(0)-A_{\tref{prop:localized-cluster-separated-seed}}\tau,\tau\bigr)
\text{-link-average separated}.
\]
\end{prop}

\begin{proof}
Choose \(G\subseteq U\) such that
\[
|G|\ge(1-\tau)|U|,
\qquad
X_u\in B(x,\tau)\quad\text{for all }u\in G.
\]
Let \(C_0:=L_\rmp+M_\rmp\). If \(\D{y}{x}\le\tau\), then for every
\(u\in G\),
\[
\D{X_u}{y}\le2\tau\le r_\rmp,
\]
and hence
\[
\rmp(\D{X_u}{y})\ge \rmp(0)-2L_\rmp\tau.
\]
The remaining vertices contribute at least \(-M_\rmp\), while
\(\rmp(0)\le M_\rmp\). Therefore
\[
\acn{U}{y}
\ge
(1-\tau)(\rmp(0)-2L_\rmp\tau)-\tau M_\rmp
\ge
\rmp(0)-2C_0\tau
\ge
\widehat{\rmp}(0)-(2C_0+1)\tau
\ge
\widehat{\rmp}(0)-A_{\tref{prop:localized-cluster-separated-seed}}\tau+\tau.
\]

Now suppose
\(\D{y}{x}\ge K_{\tref{prop:localized-cluster-separated-seed}}\tau\). For every
\(u\in G\),
\[
\D{X_u}{y}\ge (K_{\tref{prop:localized-cluster-separated-seed}}-1)\tau.
\]
Since \(\rmp\) is non-increasing and bi-Lipschitz on \([0,r_\rmp]\),
\[
\rmp(\D{X_u}{y})
\le
\rmp((K_{\tref{prop:localized-cluster-separated-seed}}-1)\tau)
\le
\rmp(0)-\ell_\rmp(K_{\tref{prop:localized-cluster-separated-seed}}-1)\tau.
\]
The remaining vertices contribute at most \(\rmp(0)\). Hence
\[
\acn{U}{y}
\le
(1-\tau)\bigl(\rmp(0)-\ell_\rmp(K_{\tref{prop:localized-cluster-separated-seed}}-1)\tau\bigr)
+
\tau\rmp(0)
=
\rmp(0)-(1-\tau)\ell_\rmp(K_{\tref{prop:localized-cluster-separated-seed}}-1)\tau.
\]
Since \(\tau\le1/2\) and
\[
K_{\tref{prop:localized-cluster-separated-seed}}
\ge
1+\frac{2(A_{\tref{prop:localized-cluster-separated-seed}}+2)}{\ell_\rmp},
\]
we have
\[
(1-\tau)\ell_\rmp(K_{\tref{prop:localized-cluster-separated-seed}}-1)\tau
\ge
 (A_{\tref{prop:localized-cluster-separated-seed}}+2)\tau.
\]
\[
\acn{U}{y}
\le
\rmp(0)-(A_{\tref{prop:localized-cluster-separated-seed}}+2)\tau
\le
\widehat{\rmp}(0)-A_{\tref{prop:localized-cluster-separated-seed}}\tau-\tau.
\]
This is exactly the claimed \linkaverageseparation{}.
\end{proof}

\section{Internal-average search and candidate families}
\label{sec:internal-average-search}

We now turn the internal-average primitives into a concrete \candidatefamily{}.
The output of this section is only
\[
\bigl(\widehat{\rmp}(0),\mathcal C_\rho\bigr):
\]
an observable estimate of the top link value and a large family of raw
\internalaveragecandidate{}s covering the latent space.  The subsequent
sparsification of \(\mathcal C_\rho\) is handled in
Section~\ref{sec:refinement-nets}.

\begin{lemma}[Estimating \(\rmp(0)\) from the maximal internal average]
\label{lem:p0-from-max-internal-average}
Let \(W\subseteq\bfV\) be a vertex set, and set \(n_\star:=|W|\). Fix
\[
\rho\in(0,\min\{r_\mu,r_\rmp/2\}],
\]
and set
\[
m:=\left\lfloor\frac{\phi(\rho/3)}{2}n_\star\right\rfloor,
\qquad
\widehat{\rmp}(0):=\max\{\pav{U}:U\subseteq W,\ |U|=m\}.
\]
Assume \(m\ge2\). On the event
\[
\Ept{W}{\rho}\cap\Enn{W}{L_\rmp\rho}{m},
\]
we have
\[
\rmp(0)-3L_\rmp\rho
\le
\widehat{\rmp}(0)
\le
\rmp(0)+L_\rmp\rho.
\]
In particular, \(|\widehat{\rmp}(0)-\rmp(0)|\le3L_\rmp\rho\).
\end{lemma}

\begin{proof}
Fix \(x\in M\). On \(\Ept{W}{\rho}\), there exists \(U_x\subseteq W\) with
\(|U_x|=m\) such that
\[
X_u\in B(x,\rho)
\qquad\text{for every }u\in U_x.
\]
Thus \(\D{X_u}{X_v}<2\rho\) for all distinct \(u,v\in U_x\), and hence
\[
\apav{U_x}\ge \rmp(2\rho).
\]
Using \(\Enn{W}{L_\rmp\rho}{m}\),
\[
\widehat{\rmp}(0)\ge\pav{U_x}\ge\apav{U_x}-L_\rmp\rho
\ge \rmp(2\rho)-L_\rmp\rho.
\]
Since \(2\rho\le r_\rmp\) and \(\rmp\) is \(L_\rmp\)-Lipschitz on
\([0,r_\rmp]\),
\[
\rmp(2\rho)\ge \rmp(0)-2L_\rmp\rho,
\]
which gives the lower bound.

For the upper bound, let \(U\subseteq W\) have \(|U|=m\). Since \(\rmp\) is
non-increasing, \(\apav{U}\le\rmp(0)\). Again using
\(\Enn{W}{L_\rmp\rho}{m}\),
\[
\pav{U}\le\apav{U}+L_\rmp\rho\le\rmp(0)+L_\rmp\rho.
\]
Taking the maximum over all such \(U\) proves the claim.
\end{proof}

\begin{prop}[Internal-average search produces candidate families]
\label{prop:inner-density-candidates}
Let \(W\subseteq\bfV\) be a vertex set, and set \(n_\star:=|W|\). Let
\[
C_*\ge \frac{48L_\rmp}{\ell_\rmp},
\qquad
0<\rho\le \min\{r_\mu,r_\rmp/C_*\}.
\]
Set
\[
m_\rho:=\left\lfloor\frac{\phi(\rho/3)}{2}n_\star\right\rfloor,
\qquad
\Delta_\rho:=4L_\rmp\rho,
\]
assume \(m_\rho\ge2\), and define
\[
\widehat{\rmp}(0):=\max\{\pav{U}:U\subseteq W,\ |U|=m_\rho\},
\]
\[
\mathcal C_\rho
:=
\{U\subseteq W:\ |U|=m_\rho,\ \pav{U}\ge \widehat{\rmp}(0)-\Delta_\rho\}.
\]
On the good event
\[
\Ept{W}{\rho}\cap\Enn{W}{L_\rmp\rho}{m_\rho},
\]
the following hold.
\begin{enumerate}
\item The maximal internal average satisfies
\[
|\widehat{\rmp}(0)-\rmp(0)|\le 3L_\rmp\rho.
\]

\item For every \(x\in M\), there exists \(U_x\in\mathcal C_\rho\) such that
\[
X_u\in B(x,\rho)\qquad\text{for every }u\in U_x.
\]

\item Every \(U\in\mathcal C_\rho\) is a \(8L_\rmp\rho\)-\internalaveragecandidate{}; namely, for every such \(U\),
\[
\apav{U}\ge \rmp(0)-8L_\rmp\rho.
\]
\end{enumerate}
\end{prop}

\begin{proof}
By Lemma~\ref{lem:p0-from-max-internal-average},
\[
\rmp(0)-3L_\rmp\rho
\le
\widehat{\rmp}(0)
\le
\rmp(0)+L_\rmp\rho.
\]
This proves the first assertion.

For every \(x\in M\), the event \(\Ept{W}{\rho}\) gives a set
\(U_x\subseteq W\) with \(|U_x|=m_\rho\) and \(X_u\in B(x,\rho)\) for all
\(u\in U_x\). As in the proof of Lemma~\ref{lem:p0-from-max-internal-average},
\[
\pav{U_x}
\ge
\rmp(2\rho)-L_\rmp\rho
\ge
\rmp(0)-3L_\rmp\rho.
\]
Since
\[
\widehat{\rmp}(0)\le \rmp(0)+L_\rmp\rho,
\]
we have
\[
\pav{U_x}
\ge
\widehat{\rmp}(0)-4L_\rmp\rho
=
\widehat{\rmp}(0)-\Delta_\rho.
\]
Thus \(U_x\in\mathcal C_\rho\), proving the second assertion.

Now let \(U\in\mathcal C_\rho\). Then
\[
\pav{U}\ge \widehat{\rmp}(0)-\Delta_\rho.
\]
Using the lower bound
\[
\widehat{\rmp}(0)\ge \rmp(0)-3L_\rmp\rho
\]
and the event \(\Enn{W}{L_\rmp\rho}{m_\rho}\), we get
\[
\apav{U}
\ge
\pav{U}-L_\rmp\rho
	\ge
	\rmp(0)-8L_\rmp\rho.
\]
This proves the third assertion.
\end{proof}

\section{Refinement nets from candidate families}
\label{sec:refinement-nets}

The raw family \(\mathcal C_\rho\) from
Section~\ref{sec:internal-average-search} is too large to use directly. This
section extracts a small subfamily whose representatives cover \(M\), using the
\linkaverageseparation{} properties from
Section~\ref{sec:internal-average-link-separation}. The resulting object is
called a refinement net. As in the previous section, the main result is a
unified statement, Proposition~\ref{prop:candidate-refinement-net}, covering
both the \oracleroute{} and the \localroute{}.

\begin{defi}[Refinement net]
\label{def:refinement-net}
Let \(R_{\rm net},R_{\rm in},R_{\rm out},\gamma>0\), and let \(M_{\rm net}\ge1\).
Let \(\Theta\in\mathbb R\).
A family
\[
\mathfrak N=\{(U,\theta_U):U\in\mathcal N\}
\]
is a refinement net with parameters
\[
\Theta,\quad R_{\rm net},\quad R_{\rm in},\quad R_{\rm out},\quad
\gamma,\quad M_{\rm net},
\]
if:
\begin{enumerate}
\item the centers \(\{\theta_U:U\in\mathcal N\}\) form an \(R_{\rm net}\)-net
of \(M\);
\item for every \(U\in\mathcal N\), \(U\) is
\[
(\theta_U,R_{\rm in},R_{\rm out},\Theta,\gamma)\text{-link-average separated};
\]
\item \(|\mathcal N|\le M_{\rm net}\).
\end{enumerate}
\end{defi}

\begin{prop}[Candidate families produce refinement nets]
\label{prop:candidate-refinement-net}
There exists a constant
\[
K_{\tref{prop:candidate-refinement-net}}>0
\]
depending only on \(L_\rmp,\ell_\rmp,M_\rmp\) such that the following holds.
Work in the setup of Proposition~\ref{prop:inner-density-candidates}, and
assume the good event there.
Let \(K\ge K_{\tref{prop:candidate-refinement-net}}\).

Consider either of the following two cases:
\begin{itemize}
\item \ORA{}: Let \(\mathcal O_{\win}\) be an
\((\alpha_{\win},\lambda_{\win})\)-\fuzzywindoworacle{} with
\[
\lambda_{\win}\le \frac18,
\qquad
\rho\le \frac{\alpha_{\win}}2 r_\rmp.
\]
Let \(r_{\rm ref}>0\) satisfy
\[
K\rho
\le
r_{\rm ref}.
\]
\item \LOC{}: Let \(r_{\rm ref}>0\) satisfy
\[
K\sqrt{\rho}
\le
r_{\rm ref}.
\]
\end{itemize}
Suppose in addition that
\[
K^3 r_{\rm ref}
\le \min\{1,r_\rmp\}.
\]
Then there is a procedure, using \pairaverage{} comparisons between candidates
and the oracle in the \ORA{} case, that outputs a subfamily \(\mathcal N\). For
suitable analysis representatives \(\theta_U\), the family
\[
\mathfrak N=\{(U,\theta_U):U\in\mathcal N\}
\]
is a refinement net with parameters
\[
\begin{gathered}
\Theta=\widehat{\rmp}(0)
-A_{\tref{prop:link-average-separation}}r_{\rm ref},\qquad
R_{\rm net}=
\tfrac13 Kr_{\rm ref},\\
R_{\rm in}=r_{\rm ref},\qquad
R_{\rm out}=Kr_{\rm ref},\qquad
\gamma=r_{\rm ref},\qquad
M_{\rm net}=|W|,
\end{gathered}
\]
where \(W\) is the vertex set used in
Proposition~\ref{prop:inner-density-candidates}.
\end{prop}

\begin{proof}
This is the combination of the two route-specific propositions below. Choose
\(K_{\tref{prop:candidate-refinement-net}}\) larger than \(4\), larger than
\(48L_\rmp/\ell_\rmp\), larger than \(8L_\rmp\), large enough compared with
the constants in Propositions~\ref{prop:oracle-candidate-net} and
\ref{prop:local-link-candidate-net}, and large enough that
\[
K_{\tref{prop:candidate-refinement-net}}
\ge
1+\frac{2(A_{\tref{prop:link-average-separation}}+2)}{\ell_\rmp}.
\]
Let \(K\ge K_{\tref{prop:candidate-refinement-net}}\).
The scale condition above then implies the scale assumptions in the
corresponding route proposition, including
\[
Kr_{\rm ref}
\le \tfrac12\min\{1,r_\rmp\}.
\]
In either case, apply the corresponding route-specific proposition. By the
choice of \(K_{\tref{prop:candidate-refinement-net}}\), the route-specific
outer radius is \(Kr_{\rm ref}\), and the route-specific net radius in either
case is
\[
\tfrac13Kr_{\rm ref}.
\]
\end{proof}

Both route constructions use the same selection mechanism. One first builds a
comparison graph \(H\) on a candidate family: in the \ORA{} case the vertices
are the oracle-certified candidates, while in the \LOC{} case the vertices are
all candidates in \(\mathcal C_\rho\). Edges are determined by observable
information, namely the oracle filter when available and empirical
\pairaverage{} comparisons. A maximal independent set of this graph keeps only
well-separated representatives, while maximality preserves coverage. The
representatives \(\theta_U\) are analysis witnesses; the procedure itself only
uses the candidate family, the oracle in the \ORA{} case, the empirical
\pairaverage{}s, and \(\widehat{\rmp}(0)\). The next lemma isolates this
deterministic selection step.

\begin{lemma}[Comparison graph selection]
\label{lem:comparison-graph-selection}
Let \(\mathcal C\) be a candidate family with representatives \(\theta_U\),
and let \(H\) be any graph on \(\mathcal C\). Write
\(U\stackrel{H}{\sim} W\) when \(U,W\) are adjacent in \(H\). Suppose there are
\(\delta,\rho_{\rm cov}>0\) and \(K_{\rm cmp}\ge1\) such that
\[
\D{\theta_U}{\theta_W}\le \delta
\quad\Longrightarrow\quad
U\stackrel{H}{\sim} W,
\]
\[
U\stackrel{H}{\sim} W
\quad\Longrightarrow\quad
\D{\theta_U}{\theta_W}\le K_{\rm cmp}\delta,
\]
and for every \(x\in M\) there exists \(U_x\in\mathcal C\) with
\[
\D{x}{\theta_{U_x}}\le \rho_{\rm cov}.
\]
Then every maximal independent set \(\mathcal N\subseteq\mathcal C\) satisfies
\[
U\neq W\in\mathcal N
\quad\Longrightarrow\quad
\D{\theta_U}{\theta_W}>\delta,
\]
and its representatives form a \((\rho_{\rm cov}+K_{\rm cmp}\delta)\)-net of
\(M\).
\end{lemma}

\begin{proof}
The first implication gives separation: if two distinct vertices of
\(\mathcal N\) had representative distance at most \(\delta\), they would be
adjacent. For covering, fix \(x\in M\) and choose \(U_x\). By maximality,
either \(U_x\in\mathcal N\), or \(U_x\stackrel{H}{\sim} U\) for some
\(U\in\mathcal N\). In the latter case,
\[
\D{x}{\theta_U}
\le
\D{x}{\theta_{U_x}}+\D{\theta_{U_x}}{\theta_U}
\le
\rho_{\rm cov}+K_{\rm cmp}\delta.
\]
\end{proof}

\subsection{Oracle route}

\begin{prop}[Oracle route produces a refinement net]
\label{prop:oracle-candidate-net}
There exists a constant \(K_{\tref{prop:oracle-candidate-net}}>0\), depending
only on \(L_\rmp,\ell_\rmp,M_\rmp\), such that the following holds. Work in the
setup of Proposition~\ref{prop:inner-density-candidates}, and assume the good
event there.
Let \(\mathcal O_{\win}\) be an
\((\alpha_{\win},\lambda_{\win})\)-\fuzzywindoworacle{} with
\[
\lambda_{\win}\le \frac18,
\qquad
\rho\le \frac{\alpha_{\win}}2 r_\rmp.
\]
For any \(K\ge K_{\tref{prop:oracle-candidate-net}}\) and any
\(r_{\rm ref}>0\) satisfying
\[
K\rho
\le
r_{\rm ref},
\qquad
K^3r_{\rm ref}
\le
\min\{1,r_\rmp\},
\]
define the oracle-filtered candidate family
\[
\mathcal C_\rho^{\ora}
:=
\left\{
U\in\mathcal C_\rho:
\mathcal O_{\win}(u,v)=1
\text{ for every }u,v\in U
\right\}.
\]
For each \(U\in\mathcal C_\rho^{\ora}\), let
\(\theta_U=X_{u_U}\) be the representative obtained from
Lemma~\ref{lem:internal-gap-average-distance}, applied with
\(\Delta=8L_\rmp\rho\).
Define a graph \(H_{\ora}\) on vertex set \(\mathcal C_\rho^{\ora}\) by joining
\(U_1,U_2\) whenever
\[
\pav{U_1,U_2}
\ge
\widehat{\rmp}(0)
-2L_\rmp \sqrt{K}r_{\rm ref}.
\]
Let \(\mathcal N_{\ora}\subseteq\mathcal C_\rho^{\ora}\) be any maximal
independent set of \(H_{\ora}\). Then the family
\[
\mathfrak N_{\ora}
:=
\{(U,\theta_U):
U\in\mathcal N_{\ora}\}
\]
is a refinement net with parameters
\[
\begin{gathered}
\Theta=\widehat{\rmp}(0)
-A_{\tref{prop:link-average-separation}}r_{\rm ref},\qquad
R_{\rm net}= \tfrac{1}{3}Kr_{\rm ref},\\
R_{\rm in}=r_{\rm ref},\qquad
R_{\rm out}=Kr_{\rm ref},\qquad
\gamma=r_{\rm ref},\qquad
M_{\rm net}=|W|.
\end{gathered}
\]
\end{prop}

We prove Proposition~\ref{prop:oracle-candidate-net}. The comparison
graph \(H_{\ora}\) is built from empirical \pairaverage{}s. To apply
Lemma~\ref{lem:comparison-graph-selection}, we need two deterministic
\pairaverage{} facts: one approximates \pairaverage{}s while all relevant
points stay inside the local link window, and the other forces a
\pairaverage{} drop when two window-safe representatives are beyond that
window.

\begin{lemma}[Pair averages inside the link window]
\label{lem:pair-average-link-window}
Let \(U_1,U_2\subseteq\bfV\) be two subsets of size \(m\ge2\). Suppose there
exist \(x_1,x_2\in M\) such that
\[
\bar d_{U_i}(x_i)\le R,
\qquad i=1,2,
\]
and
\[
\operatorname{diam}(X_{U_1}\cup X_{U_2}\cup\{x_1,x_2\})
\le r_\rmp.
\]
Then
\[
\left|
\apav{U_1,U_2}
-
\rmp\!\bigl(\D{x_1}{x_2}\bigr)
\right|
\le 4L_\rmp R.
\]
\end{lemma}
\begin{proof}
For every \((u_1,u_2)\in\mathcal D(U_1,U_2)\), the assumed diameter bound puts
both \(\D{X_{u_1}}{X_{u_2}}\) and \(\D{x_1}{x_2}\) inside \([0,r_\rmp]\).
Therefore the local \(L_\rmp\)-Lipschitz bound applies:
\begin{align*}
\left|\apav{U_1,U_2} - \rmp\!\bigl(\D{x_1}{x_2}\bigr)\right|
\le&
\frac{1}{|\mathcal D(U_1,U_2)|}
\sum_{(u_1,u_2) \in \mathcal D(U_1,U_2)}
L_\rmp \left|\D{X_{u_1}}{X_{u_2}} - \D{x_1}{x_2}\right|\\
\le&
\frac{1}{|\mathcal D(U_1,U_2)|}
\sum_{(u_1,u_2) \in \mathcal D(U_1,U_2)}
L_\rmp(\D{X_{u_1}}{x_1}+\D{X_{u_2}}{x_2})\\
\le&
\frac{1}{|\mathcal D(U_1,U_2)|}
\sum_{u_1 \in U_1,u_2 \in U_2} L_\rmp(\D{X_{u_1}}{x_1}+\D{X_{u_2}}{x_2})\\
\le&
\frac{|U_1||U_2|}{|\mathcal D(U_1,U_2)|}\,2L_\rmp R
\le
4L_\rmp R,
\end{align*}
where in the last step we used the fact that
$$
    |\mathcal D(U_1,U_2)| \ge |U_1||U_2| - \min\{|U_1|,|U_2|\} \ge m(m-1)\,.
$$
\end{proof}

\begin{lemma}[Far representatives force a pair-average drop]
\label{lem:far-representatives-pair-average-drop}
Assume \(\lambda_{\win}<1/4\). Let \(U_1,U_2\subseteq\bfV\), with
\(|\mathcal D(U_1,U_2)|>0\), and let \(x_i\in M\), \(i=1,2\). Suppose
\[
\D{X_u}{x_i}\le \lambda_{\win}r_\rmp
\qquad
(u\in U_i,\ i=1,2).
\]
If
\[
\D{x_1}{x_2}\ge (1-2\lambda_{\win})r_\rmp,
\]
then
\[
\apav{U_1,U_2}
\le
\rmp((1-4\lambda_{\win})r_\rmp)
\le
\rmp(0)-\ell_\rmp(1-4\lambda_{\win})r_\rmp.
\]
\end{lemma}

\begin{proof}
For \(u_i\in U_i\),
\[
\D{X_{u_1}}{X_{u_2}}
\ge
\D{x_1}{x_2}-\D{X_{u_1}}{x_1}-\D{X_{u_2}}{x_2}
\ge
(1-4\lambda_{\win})r_\rmp.
\]
Since \(\rmp\) is non-increasing,
\[
\rmp(\D{X_{u_1}}{X_{u_2}})
\le
\rmp((1-4\lambda_{\win})r_\rmp).
\]
Averaging gives the first inequality. Since
\((1-4\lambda_{\win})r_\rmp\in[0,r_\rmp]\), the local lower
bi-Lipschitz bound gives the second inequality.
\end{proof}

\begin{proof}[Proof of Proposition~\ref{prop:oracle-candidate-net}]
Set
\[
\rho_{\rm sep}:=
\sqrt{K}r_{\rm ref},
\]
which is the separation scale for the comparison graph \(H_{\ora}\).
Taking \(K_{\tref{prop:oracle-candidate-net}}\) large enough, the assumptions on
\(r_{\rm ref}\) imply
\[
8L_\rmp\rho\le r_{\rm ref}.
\]

\step{The oracle-filtered family still covers}
For every \(x\in M\), Proposition~\ref{prop:inner-density-candidates} gives
\(U_x\in\mathcal C_\rho\) such that
\[
X_u\in B(x,\rho)
\qquad\text{for every }u\in U_x.
\]
Thus, for \(u,v\in U_x\),
\[
\D{X_u}{X_v}<2\rho\le \alpha_{\win}r_\rmp,
\]
by the assumption on \(\alpha_{\win}\).
The oracle therefore certifies every pair in \(U_x\), so
\[
U_x\in\mathcal C_\rho^{\ora}.
\]

\step{Average-radius control}
If \(U\in\mathcal C_\rho^{\ora}\), then the oracle implication
\(\mathcal O_{\win}(u,v)=1\Rightarrow \D{X_u}{X_v}\le
\lambda_{\win}r_\rmp\) shows that \(U\) is
\(\lambda_{\win}\)-\windowsafe{} and hence \(1\)-\windowsafe{}. Since
Proposition~\ref{prop:inner-density-candidates} gives
\[
\apav{U}\ge \rmp(0)-8L_\rmp\rho,
\]
Lemma~\ref{lem:internal-gap-average-distance} gives
\begin{align}
    \label{eq:oracle-candidate-average-radius-control}
\bar d_U(\theta_U)
\le
\frac{8L_\rmp}{\ell_\rmp}\rho
\qquad\text{for every }U\in\mathcal C_\rho^{\ora}.
\end{align}
Moreover, since \(\theta_U\in X_U\), every \(u\in U\) satisfies
\[
\D{X_u}{\theta_U}\le \lambda_{\win}r_\rmp.
\]

\step{Close representatives are adjacent}
Let \(U_1,U_2\in\mathcal C_\rho^{\ora}\), and write
\[
d_{12}:=\D{\theta_{U_1}}{\theta_{U_2}}.
\]
If \(d_{12}\le\rho_{\rm sep}\), then
for every \(u_i\in U_i\),
\[
\D{X_{u_1}}{X_{u_2}}
\le
\D{X_{u_1}}{\theta_{U_1}}+d_{12}+\D{\theta_{U_2}}{X_{u_2}}
\le
    \lambda_{\win}r_\rmp + \rho_{\rm sep} + \lambda_{\win}r_\rmp
\le \tfrac{1}{4}r_\rmp + \rho_{\rm sep}
\le r_\rmp,
\]
where the last inequality follows from
\(\rho_{\rm sep}=\sqrt{K}r_{\rm ref}
\le K^{-2.5}\min\{1,r_\rmp\}\).
Since it holds for every $u_i\in U_i$, we have
\[
\operatorname{diam}(X_{U_1}\cup X_{U_2}\cup
\{\theta_{U_1},\theta_{U_2}\})
\le r_\rmp.
\]
Lemma~\ref{lem:pair-average-link-window}
with $\bar d_U(\theta_U)
\le
\frac{8L_\rmp}{\ell_\rmp}\rho$ from \eqref{eq:oracle-candidate-average-radius-control}
and the event
\(\Enn{W}{L_\rmp\rho}{m_\rho}\) give
\[
\left|
\pav{U_1,U_2}
-
\rmp(d_{12})
\right|
\le
\left|
    \pav{U_1,U_2} - \apav{U_1,U_2}
\right|
+
\left|
    \apav{U_1,U_2} - \rmp(d_{12})
\right|
\le 4L_\rmp \cdot \frac{8L_\rmp}{\ell_\rmp}\rho + L_\rmp\rho
\le 0.1L_\rmp \rho_{\rm sep},
\]
provided \(K_{\tref{prop:oracle-candidate-net}}\) is large enough.
Therefore
\[
\pav{U_1,U_2}
\ge
\rmp(d_{12})-0.1 L_\rmp \rho_{\rm sep}
\ge
\rmp(0)-L_\rmp\rho_{\rm sep}-0.1 L_\rmp \rho_{\rm sep}
\ge
\widehat{\rmp}(0)- 3L_\rmp \rho  -L_\rmp\rho_{\rm sep}-0.1 L_\rmp \rho_{\rm sep}
\ge
\widehat{\rmp}(0)- 2L_\rmp\rho_{\rm sep}
\]
where the last inequality holds when \(K_{\tref{prop:oracle-candidate-net}}\)
is large enough.
Thus \(U_1\stackrel{H_{\ora}}{\sim}U_2\).

\step{Adjacent representatives are not too far}
Assume now that
\[
d_{12}>K_{\rm cover}\rho_{\rm sep},
\]
for some \(K_{\rm cover}\) to be chosen later.
We show that \(U_1\) and \(U_2\) are not adjacent.
First suppose
\[
d_{12}\le (1-2\lambda_{\win})r_\rmp.
\]
Then the same triangle-inequality check as above gives
\[
\operatorname{diam}(X_{U_1}\cup X_{U_2}\cup
\{\theta_{U_1},\theta_{U_2}\})
\le r_\rmp,
\]
and the same approximation bound gives
\[
\pav{U_1,U_2}
\le
\rmp(d_{12})+0.1 L_\rmp \rho_{\rm sep}
\le
\rmp(0)-\ell_\rmp d_{12}+0.1 L_\rmp \rho_{\rm sep}
\le
\rmp(0)-\ell_\rmp K_{\rm cover}\rho_{\rm sep}+0.1 L_\rmp \rho_{\rm sep}
\le
\widehat{\rmp}(0)-2L_\rmp\rho_{\rm sep},
\]
where \(K_{\rm cover}\) is chosen large enough depending only on \(L_\rmp\) and
\(\ell_\rmp\), for example,
$$
K_{\rm cover} =  \tfrac{1}{\ell_\rmp} \left(
2L_\rmp + 0.1 L_\rmp + 3L_\rmp
\right)
$$
suffices, due to $\rho \le \rho_{\rm sep}$.

Thus \(U_1\) and \(U_2\) are not adjacent.

It remains to consider
\[
d_{12}>(1-2\lambda_{\win})r_\rmp.
\]
Lemma~\ref{lem:far-representatives-pair-average-drop}, applied with
\(x_i=\theta_{U_i}\), gives
\[
\apav{U_1,U_2}
\le
\rmp(0)-\ell_\rmp(1-4\lambda_{\win})r_\rmp.
\]
Together with \(\Enn{W}{L_\rmp\rho}{m_\rho}\), this yields
\[
\pav{U_1,U_2}
\le
\rmp(0)-\ell_\rmp(1-4\lambda_{\win})r_\rmp+L_\rmp\rho
<
\widehat{\rmp}(0)-2L_\rmp\rho_{\rm sep},
\]
provided \(K_{\tref{prop:oracle-candidate-net}}\) is large enough.
Thus, in all cases, adjacency implies
\[
\D{\theta_{U_1}}{\theta_{U_2}}\le K_{\rm cover}\rho_{\rm sep}.
\]

\step{Apply the comparison graph selection lemma}
The oracle-filtered covering step gives, for every \(x\in M\), a candidate
\(U_x\in\mathcal C_\rho^{\ora}\) with
\[
\D{x}{\theta_{U_x}}\le\rho.
\]
The preceding two steps verify the hypotheses of
Lemma~\ref{lem:comparison-graph-selection} with
\[
\delta=\rho_{\rm sep},
\qquad
\rho_{\rm cov}=\rho,
\qquad
K_{\rm cmp}=K_{\rm cover}.
\]
Thus the selected representatives are \(\rho_{\rm sep}\)-separated and form a
\((\rho+K_{\rm cover}\rho_{\rm sep})\)-net. Since
\(\rho\le\rho_{\rm sep}\), this net radius is at most
\[
(1+K_{\rm cover})\rho_{\rm sep}.
\]
By increasing \(K_{\tref{prop:oracle-candidate-net}}\) if necessary, this is at
most \( \tfrac{K}{3} r_{\rm ref}\).
The separation also implies that the representatives
\(\theta_U=X_{u_U}\), \(U\in\mathcal N_{\ora}\), are all distinct. Hence
the map \(U\mapsto u_U\) injects \(\mathcal N_{\ora}\) into \(W\), and
\(|\mathcal N_{\ora}|\le |W|\).

\step{Separated-seed certificates}
For every selected \(U\), Proposition~\ref{prop:inner-density-candidates} gives
\[
\apav{U}\ge \rmp(0)-8L_\rmp\rho,
\]
and the oracle filter gives \(\lambda_{\win}\)-\windowsafe{}. Also
\[
|\widehat{\rmp}(0)-\rmp(0)|\le 3L_\rmp\rho\le r_{\rm ref}.
\]
Since \(8L_\rmp\rho\le r_{\rm ref}\), \(U\) is an
\(r_{\rm ref}\)-\internalaveragecandidate{}. The choice of
\(K_{\tref{prop:oracle-candidate-net}}\) and the assumption
\(K^3r_{\rm ref}\le \min\{1,r_\rmp\}\) give the \(K\)-requirements in
Proposition~\ref{prop:link-average-separation}. Therefore
Proposition~\ref{prop:link-average-separation}, in the \ORA{} case, applied
with \(r=r_{\rm ref}\), the estimate for \(\widehat{\rmp}(0)\), and this value
of \(K\), gives the required
refinement-net certificate.
\end{proof}

\subsection{Local route}

In the \localroute{}, each \(U\in\mathcal C_\rho\) is viewed as an
\((r,r)\)-cluster with center \(\theta_U\), at the cost of passing from
\(r\asymp \rho\) to \(r\asymp \sqrt{\rho}\), as explained in
Remark~\ref{rem:internal-average-to-tau-tau}.

\begin{prop}[Local route produces a refinement net]
\label{prop:local-link-candidate-net}
There exists a constant \(K_{\tref{prop:local-link-candidate-net}}>0\)
depending only on \(L_\rmp,\ell_\rmp,M_\rmp\) such that the following holds.
Work in the setup of Proposition~\ref{prop:inner-density-candidates}, and
assume the good event there.
For any \(K\ge K_{\tref{prop:local-link-candidate-net}}\) and any
\(r_{\rm ref}>0\) satisfying
\[
K\sqrt{\rho}
\le
r_{\rm ref},
\qquad
K^3r_{\rm ref}
\le
\min\{1,r_\rmp\},
\]
do the following. For each \(U\in\mathcal C_\rho\), let
\(\theta_U=X_{u_U}\) be a representative
obtained from Lemma~\ref{lem:internal-gap-average-distance}, applied with
\(\Delta=8L_\rmp\rho\).
Define a graph \(H_{\loc}\) on vertex set \(\mathcal C_\rho\) by joining
\(U_1,U_2\) whenever
\[
\pav{U_1,U_2}
\ge
\widehat{\rmp}(0)
-2L_\rmp \sqrt{K}r_{\rm ref}.
\]
Let \(\mathcal N_{\loc}\subseteq\mathcal C_\rho\) be any maximal independent set
of \(H_{\loc}\). Then the family
\[
\mathfrak N_{\loc}
:=
\{(U,\theta_U):
U\in\mathcal N_{\loc}\}
\]
is a refinement net with parameters
\[
\begin{gathered}
\Theta=\widehat{\rmp}(0)
-A_{\tref{prop:link-average-separation}}r_{\rm ref},\qquad
R_{\rm net}=
\tfrac13 Kr_{\rm ref},\\
R_{\rm in}=r_{\rm ref},\qquad
R_{\rm out}=Kr_{\rm ref},\qquad
\gamma=r_{\rm ref},\qquad
M_{\rm net}=|W|.
\end{gathered}
\]
\end{prop}

We prove Proposition~\ref{prop:local-link-candidate-net}. Here all candidates
are first viewed as approximate localized clusters. The only additional
ingredient needed for Lemma~\ref{lem:comparison-graph-selection} is the
following pair-average estimate, which relates the empirical comparison graph
\(H_{\loc}\) to distances between the cluster centers.

\begin{lemma}[Pair averages for approximate local clusters]
\label{lem:pair-average-tau-clusters-local}
Let \(U_1,U_2\subseteq\bfV\) with
\[
|U_1|=|U_2|=:m\ge2.
\]
For \(i=1,2\), suppose that \(U_i\) is a \((\tau,\tau)\)-cluster with center
\(x_i\in M\), where \(0<\tau\le1/2\). Then
\[
\rmp\!\bigl(\D{x_1}{x_2}+2\tau\bigr)-8M_\rmp\tau
\le
\apav{U_1,U_2}
\le
\rmp\!\bigl((\D{x_1}{x_2}-2\tau)_+\bigr)+8M_\rmp\tau.
\]
\end{lemma}

\begin{proof}
For \(i=1,2\), choose \(G_i\subseteq U_i\) such that
\[
|G_i|\ge(1-\tau)m,
\qquad
X_u\in B(x_i,\tau)\quad\text{for all }u\in G_i,
\]
and set \(B_i:=U_i\setminus G_i\). Let
\[
d_0:=\D{x_1}{x_2},
\qquad
\underline A:=\rmp(d_0+2\tau),
\qquad
\overline A:=\rmp((d_0-2\tau)_+),
\qquad
D_{12}:=|\mathcal D(U_1,U_2)|.
\]
For every \((u_1,u_2)\in G_1\times G_2\), \(u_1\neq u_2\), the triangle
inequality gives
\[
(d_0-2\tau)_+\le \D{X_{u_1}}{X_{u_2}}\le d_0+2\tau.
\]
Hence the corresponding \(\rmp\)-values lie between \(\underline A\) and
\(\overline A\). Let \(\mathcal B\subseteq\mathcal D(U_1,U_2)\) be the set of
ordered pairs for which at least one endpoint lies in \(B_1\cup B_2\). Then
\[
|\mathcal B|\le |B_1|m+|B_2|m\le 2\tau m^2,
\qquad
D_{12}\ge m(m-1),
\]
and therefore, since \(m\ge2\) and \(\tau\le1/2\),
\[
\frac{|\mathcal B|}{D_{12}}\le 4\tau.
\]
Because all \(\rmp\)-values have absolute value at most \(M_\rmp\), replacing
the good-pair bounds by arbitrary bad-pair values can change the average by at
most \(2M_\rmp|\mathcal B|/D_{12}\). Consequently,
\[
\apav{U_1,U_2}\ge
\underline A-2M_\rmp\frac{|\mathcal B|}{D_{12}}
\ge
\underline A-8M_\rmp\tau
\]
and
\[
\apav{U_1,U_2}\le
\overline A+2M_\rmp\frac{|\mathcal B|}{D_{12}}
\le
\overline A+8M_\rmp\tau.
\]
\end{proof}

\begin{proof}[Proof of Proposition~\ref{prop:local-link-candidate-net}]
Set
\[
\rho_{\rm sep}:=
\sqrt{K}r_{\rm ref}.
\]
By Proposition~\ref{prop:inner-density-candidates} and
Remark~\ref{rem:internal-average-to-tau-tau}, every \(U\in\mathcal C_\rho\) is a
\((r_{\rm ref},r_{\rm ref})\)-cluster with center \(\theta_U\). Here we use
the scale assumptions above and take
\(K_{\tref{prop:local-link-candidate-net}}\) large enough so that
\[
\sqrt{\frac{8L_\rmp\rho}{\ell_\rmp}}\le r_{\rm ref},
\qquad
3L_\rmp\rho\le r_{\rm ref}.
\]

Let \(U_1,U_2\in\mathcal C_\rho\), and write
\[
d_{12}:=\D{\theta_{U_1}}{\theta_{U_2}}.
\]
By Lemma~\ref{lem:pair-average-tau-clusters-local} and
\(\Enn{W}{L_\rmp\rho}{m_\rho}\),
\[
\pav{U_1,U_2}
\ge
\rmp(d_{12}+2r_{\rm ref})-8M_\rmp r_{\rm ref}-L_\rmp\rho
\]
and
\[
\pav{U_1,U_2}
\le
\rmp((d_{12}-2r_{\rm ref})_+)+8M_\rmp r_{\rm ref}+L_\rmp\rho.
\]
Also, Proposition~\ref{prop:inner-density-candidates} gives
\[
\rmp(0)-3L_\rmp\rho
\le
\widehat{\rmp}(0)
\le
\rmp(0)+L_\rmp\rho.
\]
By choosing \(K_{\tref{prop:local-link-candidate-net}}\) large enough, the
scale assumptions ensure
\[
L_\rmp\rho_{\rm sep}
\ge
2L_\rmp r_{\rm ref}+8M_\rmp r_{\rm ref}+2L_\rmp\rho,
\qquad
\ell_\rmp\rho_{\rm sep}
\ge
8M_\rmp r_{\rm ref}+5L_\rmp\rho,
\qquad
\rho_{\rm sep}\ge \max\{r_{\rm ref},\rho\}.
\]
If
\[
d_{12}\le \rho_{\rm sep},
\]
then \(d_{12}+2r_{\rm ref}\le 3\rho_{\rm sep}\le r_\rmp\), and the Lipschitz
bound gives
\[
\pav{U_1,U_2}
\ge
\rmp(0)-L_\rmp(d_{12}+2r_{\rm ref})-8M_\rmp r_{\rm ref}-L_\rmp\rho
\ge
\widehat{\rmp}(0)-2L_\rmp\rho_{\rm sep}.
\]
Thus \(U_1\stackrel{H_{\loc}}{\sim}U_2\).

If
\[
d_{12}>
2r_{\rm ref}+
\left(1+2\frac{L_\rmp}{\ell_\rmp}\right)\rho_{\rm sep},
\]
then \(d_{12}-2r_{\rm ref}>(1+2L_\rmp/\ell_\rmp)\rho_{\rm sep}\). Since
\((1+2L_\rmp/\ell_\rmp)\rho_{\rm sep}\le r_\rmp\), monotonicity and the lower
Lipschitz bound give
\begin{align*}
\pav{U_1,U_2}
&\le
\rmp\!\left(\left(1+2\frac{L_\rmp}{\ell_\rmp}\right)\rho_{\rm sep}\right)
+8M_\rmp r_{\rm ref}+L_\rmp\rho\\
&\le
\rmp(0)-\ell_\rmp\rho_{\rm sep}-2L_\rmp\rho_{\rm sep}
+8M_\rmp r_{\rm ref}+L_\rmp\rho\\
&<
\widehat{\rmp}(0)-2L_\rmp\rho_{\rm sep}.
\end{align*}
Thus \(U_1\) and \(U_2\) are not adjacent. Therefore adjacency implies
\begin{align}
\label{eq:local-candidate-adjacent-representatives-close}
d_{12}\le
2r_{\rm ref}+\left(1+2\frac{L_\rmp}{\ell_\rmp}\right)\rho_{\rm sep}.
\end{align}
\step{Apply the comparison graph selection lemma}
Proposition~\ref{prop:inner-density-candidates} supplies, for every \(x\in M\),
a candidate \(U_x\in\mathcal C_\rho\) with
\[
\D{x}{\theta_{U_x}}\le\rho.
\]
Since \(r_{\rm ref}\le\rho_{\rm sep}\),
\eqref{eq:local-candidate-adjacent-representatives-close} gives
\[
U_1\stackrel{H_{\loc}}{\sim}U_2
\quad\Longrightarrow\quad
d_{12}\le
\left(3+2\frac{L_\rmp}{\ell_\rmp}\right)\rho_{\rm sep}.
\]
Lemma~\ref{lem:comparison-graph-selection}, applied with
\[
\delta=\rho_{\rm sep},
\qquad
\rho_{\rm cov}=\rho,
\qquad
K_{\rm cmp}=3+2\frac{L_\rmp}{\ell_\rmp},
\]
shows that the selected representatives are \(\rho_{\rm sep}\)-separated and
form a net of radius at most
\[
\rho+\left(3+2\frac{L_\rmp}{\ell_\rmp}\right)\rho_{\rm sep}
\le
4\left(1+\frac{L_\rmp}{\ell_\rmp}\right)\rho_{\rm sep}.
\]
By increasing \(K_{\tref{prop:local-link-candidate-net}}\) if necessary, this is
at most \(\tfrac13 Kr_{\rm ref}\).
The separation also implies that the representatives
\(\theta_U=X_{u_U}\), \(U\in\mathcal N_{\loc}\), are all distinct. Hence
the map \(U\mapsto u_U\) injects \(\mathcal N_{\loc}\) into \(W\), and
\(|\mathcal N_{\loc}|\le |W|\).

\step{Separated-seed certificates}
Each selected \(U\) is a \((r_{\rm ref},r_{\rm ref})\)-cluster with center
\(\theta_U\).
Furthermore,
\[
|\widehat{\rmp}(0)-\rmp(0)|\le 3L_\rmp\rho\le r_{\rm ref}.
\]
By the choice of \(K_{\tref{prop:local-link-candidate-net}}\), and since
\(K^3r_{\rm ref}\le\min\{1,r_\rmp\}\), the \(K\)-requirements in
Proposition~\ref{prop:link-average-separation} hold.
Proposition~\ref{prop:link-average-separation}, in the \LOC{} case, applied
with \(r=r_{\rm ref}\), the estimate for \(\widehat{\rmp}(0)\), and
this value of \(K\), gives the required
refinement-net certificate.
\end{proof}

\section{Link-average threshold refinement}
\label{sec:link-average-threshold-refinement}

The refinement step consumes a link-average separated set and tests it against
a fresh vertex block. Thresholding at the associated level \(t\) keeps the inner
ball and rejects vertices outside the outer ball, up to the empirical
link-average fluctuation. The main statement below is the form used later:
it gives readable sufficient conditions for applying the refinement step
simultaneously over a whole refinement net.

\begin{prop}[Simultaneous threshold refinement]
\label{prop:readable-simultaneous-refinement}
\label{prop:simultaneous-refinement-net}
\label{cor:readable-simultaneous-refinement}
\label{cor:simultaneous-refinement-net}
Let
\[
\mathfrak N=\{(U,\theta_U):U\in\mathcal N\}
\]
be a refinement net with parameters
\[
(\Theta,\quad R_{\rm net},\quad R_{\rm in},\quad R_{\rm out},\quad
\gamma,\quad M_{\rm net}).
\]
Suppose \(|U|\ge m\) for every
\(U\in\mathcal N\). Let \(V\subseteq\bfV\) be
disjoint from every seed in \(\mathcal N\), and set \(n_\star:=|V|\ge2\).
Assume \(\Ept{V}{R_{\rm in}}\). Let \(L_{\rm fail}>0\), and define
\[
\varphi_{\rm in}:=\phi(R_{\rm in}/3).
\]
Assume \(R_{\rm out}\le1\), \(\varphi_{\rm in}>0\), and
\begin{align}
\label{eq:sim-refinement-failure-budget}
L_{\rm fail}
&\le
c_{\tref{prop:readable-simultaneous-refinement}}
R_{\rm out}\varphi_{\rm in}n_\star,
\\
\label{eq:sim-refinement-approx-concentration}
\sp m\min\{\gamma^2,1\}
&\ge
C_{\tref{prop:readable-simultaneous-refinement}}
\log\!\left(\frac{e}{R_{\rm out}\varphi_{\rm in}}\right),
\end{align}
where the constants
\[
c_{\tref{prop:readable-simultaneous-refinement}},
C_{\tref{prop:readable-simultaneous-refinement}}>0
\]
depend only on \(c_{\rm link},C_{\rm link}\). For each \(U\in\mathcal N\),
define
\[
\widehat U:=\{v\in V:\cn{U}{v}\ge \Theta\}.
\]
Then, conditionally on all latent positions of the seeds and of \(V\), with
probability at least \(1-M_{\rm net}\exp(-L_{\rm fail})\), every
\(\widehat U\), \(U\in\mathcal N\), satisfies
\[
|\widehat U|
\ge
\frac{\phi(R_{\rm in}/3)}{4}n_\star
\]
and is an \((R_{\rm out},R_{\rm out})\)-cluster with center \(\theta_U\). If, in
addition,
\begin{align}
\label{eq:sim-refinement-exact-concentration}
\sp m\min\{\gamma^2,1\}
&\ge
C_{\tref{prop:readable-simultaneous-refinement}}
\bigl(L_{\rm fail}+\log n_\star\bigr),
\end{align}
then, with the same probability bound, every \(U\in\mathcal N\) also satisfies
\[
V\cap B(\theta_U,R_{\rm in})
\subseteq
\widehat U
\subseteq
V\cap B(\theta_U,R_{\rm out}).
\]
\end{prop}

For \(m\ge1\) and \(\gamma>0\), define
\[
q_{\rm ref}(m,\gamma)
:=
C_{\rm link}
\exp\!\left(
-c_{\rm link}
\sp m\min\{\gamma^2/4,1\}
\right),
\]
the fixed-vertex tail bound from Lemma~\ref{lem:fixed-v-navigation} with
deviation level \(\gamma/2\). The following lemma is the
technical input for \cref{prop:readable-simultaneous-refinement}.

\begin{lemma}[link-average threshold refinement]
\label{prop:abstract-score-threshold-refinement}
\label{lem:abstract-score-threshold-refinement}
Let \(U,V\subseteq\bfV\) be disjoint, with \(m:=|U|\) and
\(n_\star:=|V|\ge2\). Fix realizations \(X_U=x_U\) and \(X_V=x_V\). Suppose
that \(U\) is
\[
(x_0,R_{\rm in},R_{\rm out},t,\gamma)\text{-link-average separated}
\]
and that \(\Ept{V}{R_{\rm in}}\) holds. Define
\[
\widehat U:=\{v\in V:\cn{U}{v}\ge t\}.
\]
Let \(L_{\rm fail}>0\), and assume
\begin{align}
\label{eq:single-refinement-bad-budget}
b_{\rm ref}:=2\bigl(q_{\rm ref}(m,\gamma)n_\star+L_{\rm fail}\bigr)
&\le
\frac{\phi(R_{\rm in}/3)}{4}n_\star.
\end{align}
Then, conditionally on \(X_U=x_U\) and \(X_V=x_V\), with probability at least
\(1-\exp(-L_{\rm fail})\), the following holds:
\begin{itemize}
\item \[
|\widehat U|
\ge
\frac{\phi(R_{\rm in}/3)}{4}n_\star.
\]
\item \(\widehat U\) is an
\[
\left(
R_{\rm out},
\frac{4b_{\rm ref}}{\phi(R_{\rm in}/3)n_\star}
\right)\text{-cluster}
\]
with center \(x_0\).
\item If, in addition,
\begin{align}
\label{eq:single-refinement-exact-condition}
n_\star q_{\rm ref}(m,\gamma)\le \exp(-L_{\rm fail}),
\end{align}
then
\[
V\cap B(x_0,R_{\rm in})
\subseteq
\widehat U
\subseteq
V\cap B(x_0,R_{\rm out}).
\]
In particular, \(\widehat U\) is an \((R_{\rm out},0)\)-cluster with center
\(x_0\).
\end{itemize}
\end{lemma}

\begin{proof}
Define
\[
V_{\rm in}:=\{v\in V:X_v\in B(x_0,R_{\rm in})\},
\qquad
V_{\rm out}:=\{v\in V:X_v\notin B(x_0,R_{\rm out})\},
\]
and
\[
\mathcal I_{\rm ref}
:=
\left\{
v\in V:
|\cn{U}{v}-\acn{U}{v}|>\gamma/2
\right\}.
\]
Set \(q:=q_{\rm ref}(m,\gamma)\) and \(b:=b_{\rm ref}\). Conditionally on
\(X_U=x_U\) and \(X_V=x_V\), the indicators
\(\mathbf 1_{\{v\in\mathcal I_{\rm ref}\}}\), \(v\in V\), are independent, and
Lemma~\ref{lem:fixed-v-navigation} bounds each success probability by \(q\).
Since
\[
\left(\sqrt{qn_\star}+\sqrt{L_{\rm fail}}\right)^2\le b,
\]
Remark~\ref{rem:chernoff-form} gives
\[
\mathbb P\left(
|\mathcal I_{\rm ref}|>b
\ \middle|\ X_U=x_U,\ X_V=x_V
\right)
\le
\exp(-L_{\rm fail}).
\]
Work on the event \(|\mathcal I_{\rm ref}|\le b\).

If \(v\in V_{\rm in}\setminus\mathcal I_{\rm ref}\), then
\(\acn{U}{v}\ge t+\gamma\), and hence
\[
\cn{U}{v}\ge t+\gamma-\gamma/2>t.
\]
Thus \(V_{\rm in}\setminus\mathcal I_{\rm ref}\subseteq\widehat U\). Using
\(\Ept{V}{R_{\rm in}}\),
\[
|\widehat U\cap V_{\rm in}|
\ge
\frac{\phi(R_{\rm in}/3)}{2}n_\star-b.
\]
Similarly, if \(v\in V_{\rm out}\setminus\mathcal I_{\rm ref}\), then
\(\acn{U}{v}\le t-\gamma\), and
\[
\cn{U}{v}\le t-\gamma+\gamma/2<t,
\]
so \(v\notin\widehat U\). Therefore
\[
\widehat U\cap V_{\rm out}\subseteq\mathcal I_{\rm ref},
\qquad
|\widehat U\cap V_{\rm out}|\le b.
\]
The bound \eqref{eq:single-refinement-bad-budget} gives
\[
|\widehat U|
\ge
\frac{\phi(R_{\rm in}/3)}{4}n_\star,
\]
and then
\[
\frac{|\widehat U\cap V_{\rm out}|}{|\widehat U|}
\le
\frac{4b_{\rm ref}}{\phi(R_{\rm in}/3)n_\star}.
\]
This is exactly the asserted cluster bound.

For the exact-inclusion clause, \eqref{eq:single-refinement-exact-condition}
and the union bound give
\[
\mathbb P\left(
\mathcal I_{\rm ref}\neq\varnothing
\ \middle|\ X_U=x_U,\ X_V=x_V
\right)
\le
\exp(-L_{\rm fail}).
\]
On \(\mathcal I_{\rm ref}=\varnothing\), the two pointwise inclusions become
\[
V_{\rm in}\subseteq\widehat U
\qquad\text{and}\qquad
\widehat U\cap V_{\rm out}=\varnothing,
\]
which is the desired exact inclusion.
\end{proof}

\begin{proof}[Proof of Proposition~\ref{prop:readable-simultaneous-refinement}]
Since
\[
\min\{\gamma^2/4,1\}\ge \frac14\min\{\gamma^2,1\},
\]
assumption \eqref{eq:sim-refinement-approx-concentration} implies, after increasing
\(C_{\tref{prop:readable-simultaneous-refinement}}\), that
\[
q_{\rm ref}(m,\gamma)
\le
\frac{1}{16}R_{\rm out}\varphi_{\rm in}.
\]
Assumption \eqref{eq:sim-refinement-failure-budget} gives
\[
L_{\rm fail}\le \frac{1}{16}R_{\rm out}\varphi_{\rm in}n_\star
\]
after decreasing \(c_{\tref{prop:readable-simultaneous-refinement}}\). Hence
\begin{align}
\label{eq:sim-refinement-bad-budget}
b_{\rm ref}:=
2\bigl(q_{\rm ref}(m,\gamma)n_\star+L_{\rm fail}\bigr)
&\le
\frac{R_{\rm out}\varphi_{\rm in}}4 n_\star
\le
\frac{\varphi_{\rm in}}4 n_\star.
\end{align}
For each \(U\in\mathcal N\), apply
Proposition~\ref{prop:abstract-score-threshold-refinement} with
\[
x_0=\theta_U,\qquad t=\Theta.
\]
Since \(|U|\ge m\), monotonicity of \(q_{\rm ref}\) in the first argument gives
\(q_{\rm ref}(|U|,\gamma)\le q_{\rm ref}(m,\gamma)\). Hence the corresponding
single-seed value
\[
b_U:=2\bigl(q_{\rm ref}(|U|,\gamma)n_\star+L_{\rm fail}\bigr)
\]
satisfies \(b_U\le b_{\rm ref}\). The bound
\eqref{eq:sim-refinement-bad-budget}
implies the size conclusion and the cluster error bound
\[
\frac{4b_U}{\varphi_{\rm in}n_\star}\le R_{\rm out}
\]
for every seed. Union bounding over \(|\mathcal N|\le M_{\rm net}\) gives the
claimed probability.

For the exact-inclusion conclusion, assumption
\eqref{eq:sim-refinement-exact-concentration} gives
\[
q_{\rm ref}(m,\gamma)\le n_\star^{-1}e^{-L_{\rm fail}}
\]
after increasing \(C_{\tref{prop:readable-simultaneous-refinement}}\) again.
Thus, again using \(|U|\ge m\),
\[
n_\star q_{\rm ref}(m,\gamma)\le e^{-L_{\rm fail}},
\]
so the exact-inclusion clause of
Proposition~\ref{prop:abstract-score-threshold-refinement} applies to every
seed. The same union bound over the refinement net proves the simultaneous
exact-inclusion claim.
\end{proof}

\section{Three-block extraction}
\label{sec:three-block-extraction}

The extraction statement has two routes. With a \fuzzywindoworacle{}, the
first block uses the oracle-filtered \candidatenet{} at internal scale
\(\rho_\circ\asymp r\). Without such an
oracle, the first block uses the smaller internal-average scale
\(\rho_\circ\asymp r^2\) and converts \internalaveragecandidate{}s
into localized clusters. In both routes, \(\rho_\circ\) is chosen so that the
unified \candidatenet{} statement produces a first-block refinement net at the
theorem scale \(r\). After that, the two routes pass through the same two
refinement blocks.

\begin{theor}[Three-block extraction]
\label{thm:three-block-extraction}
There exist constants
\[
C_{\rm ext},\Lambda_0>0
\]
depending only on \(L_\rmp,\ell_\rmp,M_\rmp\), such that the following holds.
Let \(K_\star\) be a parameter satisfying
$$
    K_\star \ge \max\{3,2K_{\tref{prop:candidate-refinement-net}}\}
$$
and let \(\Lambda\ge \Lambda_0\). Assume the vertex set $\bfV$ has three
disjoint blocks \(V_1,V_2,V_3\) of equal size \(n\). Let $r>0$. Assume one of
the following two cases.
\begin{itemize}
    \item \ORA: Suppose an \((\alpha_{\win},\lambda_{\win})\)-\fuzzywindoworacle{} \(\mathcal O_{\win}\) is available on the first block \(V_1\), with
\[
0<\alpha_{\win}<\lambda_{\win}\le 1/8,
\qquad
\rho_\circ:=K_\star^{-1}r,
\qquad
\rho_\circ\le \frac{\alpha_{\win}}2 r_\rmp.
\]
\item \LOC: Suppose no \fuzzywindoworacle{} is used, and set
\[
\rho_\circ:=K_\star^{-2}r^2.
\]
\end{itemize}
Assume
\[
0<r\le K_\star^{-3}\min\{1,r_\mu,r_\rmp\},
\tag{E0}
\label{eq:uniform-extraction-range-refactored}
\]
and
\[
\sp n\,\phi(\rho_\circ/3)\rho_\circ^2
\ge
C_{\rm ext}\Lambda\log n.
\tag{E1}
\label{eq:uniform-extraction-scale-condition-refactored}
\]
Then there is a three-block procedure which, with probability \(1-o(1)\), outputs

\begin{itemize}
\item for every \(v\in V_3\), a set \(v \in U_v\subseteq V_3\) such that
\[
|U_v|\ge \frac12\phi(r)n
\quad \mbox{and}\quad
U_v\text{ is a }(K_\star^2r,0)\text{-cluster with center }X_v,
\]
	    \item an estimator \(\widehat{\rmp}(0)\) satisfying \[
|\widehat{\rmp}(0)-\rmp(0)|\le r\,;
\]
\end{itemize}
Moreover, for a constant \(C>0\) depending only on the model parameters, the
number of distinct sets in $\{U_v\}_{v\in V_3}$ is at most
\[
\exp\left\{
C\,
\log^2\!\left(\frac{3}{r\phi(\rho_\circ/3)}\right)
\frac{1}
{\sp\rho_\circ^2}
\right\},
\]
meaning that many $U_v=U_{v'}$ may coincide.
Finally, the running time is bounded by
\[
\exp\left\{
C\,
\log^2\!\left(\frac{3}{r\phi(\rho_\circ/3)}\right)
\frac{1}
{\sp\rho_\circ^2}
\right\}
n\Lambda\log n.
\]
\end{theor}

\begin{rem}[Procedure behind Theorem~\ref{thm:three-block-extraction}]
\label{rem:three-block-procedure}
The statistical statement uses three ambient blocks of size \(n\), but the
implementation may work with smaller internal subsets. In the proof we choose
\(V_1'\subseteq V_1\) and \(V_2'\subseteq V_2\). The first-stage candidate
search and \candidatenet{} construction are performed only inside \(V_1'\). The
first refinement is performed into \(V_2'\). The final refinement and the final
output use the full block \(V_3\).

The first block produces a small net of seed sets. In the \oracleroute{} this net
is built from candidates certified by the oracle at internal scale
\(\rho_\circ=K_\star^{-1}r\).
In the \localroute{} the net is built from \internalaveragecandidate{}s at
\(\rho_\circ\asymp r^2\), which are then viewed as localized clusters. The
important point is that \(r\) is not the raw search scale in both routes:
\(\rho_\circ\) is the internal candidate-search scale, while \(r\) is the first
usable refinement scale. We choose \(\rho_\circ\) so that the unified
\candidatenet{} statement can be applied with \(r_{\rm ref}=r\). The selected
seeds are refined once into \(V_2'\), producing intermediate clusters, and then
refined again into \(V_3\), producing exact clusters.
\end{rem}

\subsection*{Working scales and constant choices}
Set
\[
r_1:=r,
\qquad
r_2:=K_\star r,
\qquad
r_{\rm net}:=\tfrac13K_\star r,
\qquad
r_3:=\tfrac12K_\star r_2.
\]
Choose the auxiliary block sizes by
\begin{align}
\label{eq:three-block-n1-choice}
n_1
&:=
\min\left\{
N\in\mathbb N:
\begin{array}{l}
\lceil\log\log n\rceil\le N\le n,\\
\sp N\phi(\rho_\circ/3)\rho_\circ^2
\ge C_{\rm ext}\Lambda\log N
\end{array}
\right\},
\\
\label{eq:three-block-n2-choice}
n_2
&:=
\min\left\{
N\in\mathbb N:
\begin{array}{l}
2\le N\le n,\\
\sp N\phi(r_1/3)r_1^2
\ge C_{\rm ext}\Lambda\log n
\end{array}
\right\}.
\end{align}
The sets in \eqref{eq:three-block-n1-choice} and
\eqref{eq:three-block-n2-choice} are nonempty under the scale assumption
\eqref{eq:uniform-extraction-scale-condition-refactored} from
\cref{thm:three-block-extraction}: \(N=n\) is admissible in both definitions.

\begin{lemma}[First-block construction of a refinement net]
\label{lem:first-block-refinement-net}
\label{lem:seed-net-score-gap-refactored}
Under the assumptions of Theorem~\ref{thm:three-block-extraction},
there is a first-block procedure using a subset \(V_1'\subseteq V_1\) of size
\(n_1\). Within the event
\[
\Ept{V_1'}{\rho_\circ}
\cap
\Enn{V_1'}{L_\rmp\rho_\circ}{m_1}\mbox{ with } m_1:=\left\lfloor\frac{\phi(\rho_\circ/3)}2n_1\right\rfloor,
\]
which holds with probability at least \(1-n_1^{-c\Lambda}\), the procedure
outputs an estimator \(\widehat{\rmp}(0)\), a threshold \(\Theta_2\), and a
refinement net
\[
\mathfrak N=\{(U,\theta_U):U\in\mathcal N\}
\]
such that
\[
|\widehat{\rmp}(0)-\rmp(0)|\le r,
\]
and \(\mathfrak N\) is a refinement net with parameters
\[
\begin{gathered}
\Theta=\Theta_2,\qquad
R_{\rm net}=r_{\rm net},\qquad
R_{\rm in}=r_1,\\
R_{\rm out}=r_2,\qquad
\gamma=r_1,\qquad
M_{\rm net}=n_1.
\end{gathered}
\]
\end{lemma}

\begin{proof}

Let \(V_1'\subseteq V_1\) have size \(n_1\), where \(n_1\) is chosen in
\eqref{eq:three-block-n1-choice}.

\step{Probability of the first-block event}
The Uniform lower-occupancy Lemma \ref{lem:uniform-lower-occupancy} gives
$$
\mathbb P\left(\Ept{V_1'}{\rho_\circ}\right) \ge 1-n_1^{- \tfrac{1}{16}\Lambda}.
$$
Also, we apply the Uniform concentration of pair averages Lemma \ref{lem:uniform-pair-average} with Remark \ref{rem:uniform-pair-average-condition} with
$$
    \lambda = \rho_\circ, \qquad m = m_1, \qquad n = n_1, \qquad \phi = \tfrac{\phi(\rho_\circ/3)}{2},
$$
together with the assumption that \(C_{\rm ext}\) is large enough. The defining
inequality \eqref{eq:three-block-n1-choice} gives the condition required by
the lemma:
$$
    \sp \tfrac{n_1}{\Lambda \log n_1} \cdot \tfrac{\phi(\rho_\circ/3)}{2} \cdot \rho_\circ^2 \ge C_{\tref{lem:uniform-pair-average}}.
$$
Then, we have
$$
\mathbb P\left(\Enn{V_1'}{L_\rmp\rho_\circ}{m_1}\, \vert\, X_{V_1'}\right) \ge 1 -
\exp \left\{
 - c_{\tref{lem:uniform-pair-average}} \phi(\rho_\circ/3) n_1 \log( e/ \phi(\rho_\circ/3) )
\right\}
\ge
1 - n_1^{-c_{\tref{lem:uniform-pair-average}} \Lambda},
$$
where the last inequality holds from our assumption on $n_1$, $\sp \le 1$, and $\log(e/\phi(\rho_\circ/3)) \ge \log e =1$. Therefore, we conclude that for $C_{\rm ext}$ and $\Lambda$ greater than some constants depending on the model parameters, the event holds with
$$
    1-n_1^{-c\Lambda} = 1-o(1).
$$

\step{Construction of the refinement net}
Within the event, Proposition~\ref{prop:inner-density-candidates} gives
\[
\widehat{\rmp}(0):=\max\{\pav{U}:U\subseteq V_1',\ |U|=m_1\},
\qquad
|\widehat{\rmp}(0)-\rmp(0)|\le 3L_{\rmp}\rho_\circ\le r,
\]
where the last inequality follows from the definition of \(\rho_\circ\), the
range assumption, and the lower bound on \(K_\star\).
The same proposition also gives the candidate family
\[
\mathcal C_{\rho_\circ}:=
\{U\subseteq V_1':\ |U|=m_1,\ \pav{U}\ge
\widehat{\rmp}(0)-4L_{\rmp}\rho_\circ\}.
\]

We now select a refinement net from \(\mathcal C_{\rho_\circ}\). Apply
Proposition~\ref{prop:candidate-refinement-net} in the appropriate route with
\(K=K_\star\).
The range assumption and the lower bound on \(K_\star\) imply that
\[
\rho_\circ\le \min\{r_\mu,r_\rmp/C_*\},
\qquad
K_\star^3r
\le
\min\{1,r_\rmp\}.
\]
In the \ORA{} case, the theorem assumes
\[
\rho_\circ\le \frac{\alpha_{\win}}2r_\rmp,
\]
and
\[
K_\star\rho_\circ\le r.
\]
In the \LOC{} case,
\[
K_\star\sqrt{\rho_\circ}\le r.
\]
Thus the hypotheses of Proposition~\ref{prop:candidate-refinement-net} hold
with \(r_{\rm ref}=r\). This gives a refinement net
\[
\mathfrak N=\{(U,\theta_U):U\in\mathcal N\}
\]
with threshold
\[
\Theta_2=
\widehat{\rmp}(0)-A_{\tref{prop:link-average-separation}}r
\]
and parameters
\[
\begin{gathered}
\Theta= \widehat{\rmp}(0)-A_{\tref{prop:link-average-separation}}r,\qquad
R_{\rm net}=r_{\rm net},\qquad
R_{\rm in}=r_1,\\
R_{\rm out}=r_2,\qquad
\gamma=r_1,\qquad
M_{\rm net}=n_1.
\end{gathered}
\]
\end{proof}

\begin{lemma}[Two-step refinement from a refinement net]
\label{lem:two-step-refinement-refactored}
Let
\[
\mathfrak N=\{(U,\theta_U):U\in\mathcal N\}
\]
be the refinement net from Lemma~\ref{lem:first-block-refinement-net}, with
threshold \(\Theta_2\), and let \(\widehat{\rmp}(0)\) be the estimator produced
in the same first-block step.
Recall that every first-block seed has size
\[
m_1=\left\lfloor\frac{\phi(\rho_\circ/3)}2n_1\right\rfloor.
\]
Let \(V_2'\subseteq V_2\) have size \(n_2\), where \(n_2\) is chosen in
\eqref{eq:three-block-n2-choice}. Then, with probability at least
\(1-n_2^{-c\Lambda}-n^{-\Omega(\Lambda)}\), for every \(U\in\mathcal N\) the
two refinement rounds produce sets
\[
U^{(2)}\subseteq V_2',
\qquad
U^{(3)}\subseteq V_3,
\]
such that \(U^{(2)}\) is an \((r_2,r_2)\)-cluster with center
\(\theta_U\), and
\[
V_3\cap B(\theta_U,r_2)
\subseteq
U^{(3)}
\subseteq
V_3\cap B(\theta_U,r_3).
\]
\end{lemma}

\begin{proof}
For the first refinement, \eqref{eq:three-block-n2-choice} gives
\[
\frac{\Lambda\log n}{n_2}
\le
\frac{\sp}{C_{\rm ext}}\phi(r_1/3)r_1^2.
\]
Since \(\sp\le1\), \(r_1\le1\), and \(r_1\le r_2\), increasing
\(C_{\rm ext}\) gives
\begin{align}
\label{eq:first-refinement-failure-budget-check}
\Lambda\log n
&\le
c_{\tref{prop:readable-simultaneous-refinement}}
r_2\phi(r_1/3)n_2.
\end{align}
Also, \eqref{eq:three-block-n1-choice}, the identity
\(m_1=\lfloor \phi(\rho_\circ/3)n_1/2\rfloor\), and the relation
\(\rho_\circ\le r_1\) imply, after increasing \(C_{\rm ext}\) again, that
\begin{align}
\label{eq:first-refinement-approx-concentration-check}
\sp m_1\min\{r_1^2,1\}
&\ge
C_{\tref{prop:readable-simultaneous-refinement}}
\log\!\left(\frac{e}{r_2\phi(r_1/3)}\right).
\end{align}
The range assumption gives \(r_2\le1\).
Thus \eqref{eq:first-refinement-failure-budget-check} and
\eqref{eq:first-refinement-approx-concentration-check} verify the conditions
\eqref{eq:sim-refinement-failure-budget} and
\eqref{eq:sim-refinement-approx-concentration} from
\cref{prop:simultaneous-refinement-net}, respectively, for
\[
R_{\rm in}=r_1,\qquad R_{\rm out}=r_2,\qquad \gamma=r_1,
\qquad n_\star=n_2.
\]
On \(\Ept{V_2'}{r_1}\), Proposition~\ref{prop:simultaneous-refinement-net}, applied
to \(\mathfrak N\) with \(L_{\rm fail}=\Lambda\log n\), gives for every
\(U\in\mathcal N\) a set \(U^{(2)}\subseteq V_2'\) with
\[
|U^{(2)}|\ge \frac{\phi(r_1/3)}4 n_2
\]
which is an \((r_2,r_2)\)-cluster with center \(\theta_U\). The event
\(\Ept{V_2'}{r_1}\) holds with probability \(1-n_2^{-c\Lambda}\), and the
simultaneous refinement failure is at most \(n_1\exp(-\Lambda\log n)\), which
is harmless once \(\Lambda\ge\Lambda_0\).

For the second refinement, set
\[
m_2:=\frac{\phi(r_1/3)}4 n_2.
\]
Each \(U^{(2)}\) has size at least \(m_2\), and is an \((r_2,r_2)\)-cluster
with center \(\theta_U\).
Since \(K_\star/2\ge K_{\tref{prop:candidate-refinement-net}}\), the parameter
\(K_\star/2\) is admissible in
Proposition~\ref{prop:link-average-separation}.
Proposition~\ref{prop:link-average-separation}, in the \LOC{} case, applied
with \(r=r_2\) and \(K=K_\star/2\), gives the common threshold
\[
\Theta_3:=
\widehat{\rmp}(0)-
A_{\tref{prop:link-average-separation}}r_2
\]
such that
\[
\mathfrak N^{(2)}
:=
\{(U^{(2)},\theta_U):U\in\mathcal N\}
\]
is a refinement net with parameters
\[
\begin{gathered}
\Theta=\Theta_3,\qquad
R_{\rm net}=r_{\rm net},\qquad
R_{\rm in}=r_2,\\
R_{\rm out}=r_3,\qquad
\gamma=r_2,\qquad
M_{\rm net}=n_1.
\end{gathered}
\]
The range assumption \eqref{eq:uniform-extraction-range-refactored} from
\cref{thm:three-block-extraction} ensures the local-window condition
\[
\tfrac12K_\star r_2=r_3
\le \frac12\min\{1,r_\rmp\},
\]
which in particular gives \(r_3\le1\), and the condition
\[
|\widehat{\rmp}(0)-\rmp(0)|\le 3L_\rmp\rho_\circ\le r_2
\]
needed to apply that proposition.

\eqref{eq:three-block-n2-choice} and the definition of \(m_2\) imply
\[
\sp m_2\min\{r_2^2,1\}
\ge
C_{\tref{prop:readable-simultaneous-refinement}}
\bigl(\Lambda\log n+\log n\bigr),
\]
after increasing \(C_{\rm ext}\) and \(\Lambda_0\), since \(r_2=K_\star r_1\)
and \(\Lambda\ge\Lambda_0\). This verifies the exact-concentration condition
\eqref{eq:sim-refinement-exact-concentration} from
\cref{prop:simultaneous-refinement-net}. It also verifies the
approximate-concentration condition
\eqref{eq:sim-refinement-approx-concentration} from
\cref{prop:simultaneous-refinement-net}: indeed, \eqref{eq:three-block-n2-choice}, \(n_2\le n\), and
\(\phi(r_1/3)\le\phi(r_2/3)\) imply, after adjusting constants, that
\[
r_3\phi(r_2/3)\ge n^{-1} \quad \Rightarrow \quad
\log\!\left(\frac{e}{r_3\phi(r_2/3)}\right)
\le 2\log n.
\]
Finally, because \(n_2\le n\), \eqref{eq:three-block-n2-choice} gives
\[
\frac{\Lambda\log n}{n}
\le
\frac{\Lambda\log n}{n_2}
\le
\frac{\sp}{C_{\rm ext}}\phi(r_1/3)r_1^2.
\]
Using \(\phi(r_1/3)\le\phi(r_2/3)\) and \(r_1^2\le r_3\), and increasing
\(C_{\rm ext}\) again, we get
\[
\Lambda\log n
\le
c_{\tref{prop:readable-simultaneous-refinement}}
r_3\phi(r_2/3)n.
\]
This verifies the failure-budget condition
\eqref{eq:sim-refinement-failure-budget} from
\cref{prop:simultaneous-refinement-net}. Hence all three conditions
\eqref{eq:sim-refinement-failure-budget},
\eqref{eq:sim-refinement-approx-concentration}, and
\eqref{eq:sim-refinement-exact-concentration} from
\cref{prop:simultaneous-refinement-net} hold for
\[
R_{\rm in}=r_2,\qquad R_{\rm out}=r_3,\qquad
\gamma=r_2,\qquad n_\star=n.
\]
Therefore the exact-inclusion clause of
Proposition~\ref{prop:simultaneous-refinement-net}, applied to
\(\mathfrak N^{(2)}\) on \(\Ept{V_3}{r_2}\) with
\(L_{\rm fail}=\Lambda\log n\), gives
\[
V_3\cap B(\theta_U,r_2)
\subseteq
U^{(3)}
\subseteq
V_3\cap B(\theta_U,r_3).
\]
The occupancy event \(\Ept{V_3}{r_2}\) holds with probability
\(1-n^{-c\Lambda}\), and the remaining failures are absorbed by the same union
bound over the seed family.
\end{proof}

\begin{proof}[Proof of Theorem~\ref{thm:three-block-extraction}]
Choose \(r_1,r_2,r_{\rm net},r_3\) as above and \(n_1,n_2\) by
\eqref{eq:three-block-n1-choice}--\eqref{eq:three-block-n2-choice}. Apply
Lemma~\ref{lem:first-block-refinement-net} on \(V_1'\) to obtain
\(\widehat{\rmp}(0)\), the threshold \(\Theta_2\), and a refinement net
\(\mathfrak N\). Then apply Lemma~\ref{lem:two-step-refinement-refactored} to
refine every selected seed first into \(V_2'\) and then into \(V_3\).

For each \(v\in V_3\), choose \(U\in\mathcal N\) with
\[
\D{X_v}{\theta_U}\le r_{\rm net}<r_2,
\]
and set \(U_v:=U^{(3)}\). The exact inclusion from
Lemma~\ref{lem:two-step-refinement-refactored} gives \(v\in U_v\), while
\[
U_v\subseteq B(\theta_U,r_3)
\subseteq B(X_v,r_{\rm net}+r_3)
\subseteq B(X_v,K_\star^2r).
\]
On \(\Ept{V_3}{r_2}\),
\[
|U_v|
\ge
|V_3\cap B(\theta_U,r_2)|
\ge
\frac{\phi(r_2/3)}2n
\ge
\frac{\phi(r)}2n.
\]
The estimate for \(\rmp(0)\) was produced in the first-block step. The number
of distinct output sets is bounded by \(|\mathcal N|\le n_1\le\binom{n_1}{m_1}\),
and the entropy bound in the running-time remark gives the stated estimate.
The probability bound follows
from the union bound over the first-block event, the two occupancy events, and
the two refinement rounds.
\end{proof}

\begin{rem}[Running time]
The constant \(C\) below depends only on the model parameters and may increase
from line to line.
The first-stage exhaustive search enumerates
\[
\binom{n_1}{m_1}
\le
\exp\{n_1H(\phi(\rho_\circ/3)/2)\}
\le
\exp\left\{
C\,
\frac{\log^2\!\left(\frac{3}{r\phi(\rho_\circ/3)}\right)}
{\sp\rho_\circ^2}
\right\}.
\]
The refinement part uses
\[
n_2\asymp
\frac{\Lambda\log n}{\sp\,\phi(r/3)r^2},
\]
so its cost is absorbed into
\[
\exp\left\{
C\,
\frac{\log^2\!\left(\frac{3}{r\phi(\rho_\circ/3)}\right)}
{\sp\rho_\circ^2}
\right\}
n\Lambda\log n.
\]
\end{rem}

\section{Two-round extraction via a coarse fuzzy window oracle}
\label{sec:two-round-extraction}

To prove the main theorem, we run Theorem~\ref{thm:three-block-extraction}
twice. The first run is a coarse \localroute{} extraction at scale
\(\sqrt r\). Its only purpose is to build a \fuzzywindoworacle{} on a fresh
block. The second run is the final \oracleroute{} extraction at scale \(r\).
The next lemma packages the only new ingredient: coarse exact clusters and a
coarse estimate of \(\rmp(0)\) generate the oracle needed by the second run.

\begin{lemma}[Coarse three-block output generates a fuzzy window oracle]
\label{lem:coarse-output-generates-window-oracle}
Fix constants \(C_0,c_0>0\). There exist constants
\(c_{\ora},C_{\ora},\Lambda_{\ora}>0\), depending only on
\(C_0,c_0,L_\rmp,\ell_\rmp,M_\rmp,\Ksg\), such that the following holds
whenever \(\Lambda\ge\Lambda_{\ora}\).

Let \(V_0,W\subseteq\bfV\) be disjoint vertex blocks of size at most $n$.
Suppose that we are given an estimator \(\widehat{\rmp}(0)\) and, for every
\(z\in V_0\), a set
\(
A_z\subseteq V_0
\)
such that
\[
z\in A_z,
\qquad
A_z\text{ is a }(C_0r,0)\text{-cluster with center }X_z,
\quad \mbox{and}
\quad
|A_z|\ge \frac12\phi(c_0r)n.
\]
Assume also that the centers \(\{X_z:z\in V_0\}\) form a \(C_0r\)-net of \(M\):
\[
\forall x\in M,\quad \exists z\in V_0
\quad\text{such that}\quad
\D{x}{X_z}\le C_0r,
\]
and that
\[
|\widehat{\rmp}(0)-\rmp(0)|\le C_0r.
\]

Assume
\[
r\le c_{\ora}r_\rmp
\qquad\text{and}\qquad
\sp n\phi(c_0r)r_\rmp^2\ge C_{\ora}\Lambda\log n.
\]
For each \(z\in V_0\), define
\[
W_z
:=
\left\{
v\in W:
\cn{A_z}{v}
\ge
\widehat{\rmp}(0)-\frac{\ell_\rmp}{32}r_\rmp
\right\}.
\]
Define an oracle on \(W\) by
\[
\mathcal O_{\win}(v,w)=1
\]
if there exists \(z\in V_0\) such that \(v,w\in W_z\), and otherwise set
\[
\mathcal O_{\win}(v,w)=0.
\]
Then, conditionally on the latent points of \(V_0\cup W\) and on the given sets
\(\{A_z:z\in V_0\}\), with probability at least
\(1-n^{-\Omega(\Lambda)}\), \(\mathcal O_{\win}\) is an
\((\alpha_{\win},\lambda_{\win})\)-\fuzzywindoworacle{} on \(W\),
with
\[
\alpha_{\win}:=\frac{\ell_\rmp}{128L_\rmp},
\qquad
\lambda_{\win}:=\frac18.
\]
\end{lemma}

\begin{proof}
See Appendix~\ref{app:coarse-window-oracle-proof}.
\end{proof}

\begin{theor}[Two-round extraction via a coarse fuzzy window oracle]
\label{thm:two-round-extraction}
There exist constants
\[
c_{\tref{thm:two-round-extraction}},
\qquad
C_{\tref{thm:two-round-extraction}},
\qquad
C'_{\tref{thm:two-round-extraction}},
\qquad
\Lambda_{\tref{thm:two-round-extraction}}>0,
\]
depending only on the model parameters, such that the following holds whenever
\(\Lambda\ge\Lambda_{\tref{thm:two-round-extraction}}\). Split the vertex set
into six disjoint blocks of equal size,
\[
\bfV
=
V_1\sqcup V_2\sqcup V_3
\sqcup V_4\sqcup V_5\sqcup V_6,
\qquad
|V_i|=n.
\]
Let \(r>0\) be the target scale. Assume
\[
0<\sqrt r
\le
c_{\tref{thm:two-round-extraction}}
\min\{1,r_\mu,r_\rmp\},
\tag{T0}
\label{eq:two-round-range}
\]
and
\[
\sp n\,
\phi\!\left(c_{\tref{thm:two-round-extraction}}r\right)
r^2
\ge
C_{\tref{thm:two-round-extraction}}\Lambda\log n.
\tag{T1}
\label{eq:two-round-scale}
\]
Then there is a six-block procedure which, with probability \(1-o(1)\), runs in
time
$$
\exp\left(
        C_{\tref{thm:two-round-extraction}}
        \log^2( \tfrac{3}{r\phi(c_{\tref{thm:two-round-extraction}}r)}) \tfrac{1}{\sp r^2}
    \right) n \Lambda\log n ,
$$
and outputs an estimator \(\widehat{\rmp}(0)\) and, for every \(v\in V_6\), a set
\(U_v\subseteq V_6\) such that
\[
v\in U_v,
\qquad
U_v
\text{ is a }
\left(C'_{\tref{thm:two-round-extraction}}r,0\right)
\text{-cluster with center }X_v,
\]
and
\[
|U_v|
\ge
\frac12
\phi(r)n.
\]
Further, the number of sets $\{U_v:v\in V_6\}$ without counting multiplicity is at most
\[
\exp\left(
        C_{\tref{thm:two-round-extraction}}
        \log^2( \tfrac{3}{r\phi(c_{\tref{thm:two-round-extraction}}r)}) \tfrac{1}{\sp r^2}
    \right).
\]
Moreover,
\[
|\widehat{\rmp}(0)-\rmp(0)|
\le
r.
\]
\end{theor}

\begin{proof}
Set
\[
r_0:=\sqrt r.
\]
Fix once and for all an admissible value
\[
K_\star\ge\max\{3,2K_{\tref{prop:candidate-refinement-net}}\}
\]
depending only on the model parameters, and use it for both applications of
Theorem~\ref{thm:three-block-extraction}.
All constants below are allowed to depend on this fixed \(K_\star\), hence only
on the model parameters.

\step{Round 1: coarse local-route extraction}
Apply Theorem~\ref{thm:three-block-extraction} in the \localroute{}
to \(V_1,V_2,V_3\) at target scale \(r_0\). The internal-average scale in that
run is
\[
\rho_\circ
=
K_\star^{-2}r_0^2
=
K_\star^{-2}r.
\]
After decreasing \(c_{\tref{thm:two-round-extraction}}\) and increasing
\(C_{\tref{thm:two-round-extraction}}\), the assumptions
\eqref{eq:two-round-range} and \eqref{eq:two-round-scale} imply the range and
scale conditions of Theorem~\ref{thm:three-block-extraction} for this local
run. Thus, with probability \(1-o(1)\), it outputs
\(\widehat{\rmp}_{\rm coarse}(0)\) and sets \(A_z\subseteq V_3\) such that, for
every \(z\in V_3\),
\[
z\in A_z,
\qquad
A_z\text{ is a }(K_\star^2r_0,0)\text{-cluster with center }X_z,
\qquad
|A_z|\ge \frac12\phi(r_0)n,
\qquad
|\widehat{\rmp}_{\rm coarse}(0)-\rmp(0)|
\le r_0.
\]

\step{Construct the fuzzy window oracle}
Set \(C_0:=K_\star^2\) and choose any fixed \(c_0\le1\). The event
\(\Ept{V_3}{C_0r_0}\) also holds with probability \(1-o(1)\), by
Lemma~\ref{lem:uniform-lower-occupancy}. Indeed, after decreasing
\(c_{\tref{thm:two-round-extraction}}\), the radius \(C_0r_0\) is in range and
\(\phi(C_0r_0)\ge \phi(c_{\tref{thm:two-round-extraction}}r)\); then
\eqref{eq:two-round-scale}, together with \(\sp r^2\le1\), gives the needed
occupancy lower bound. Hence the centers \(\{X_z:z\in V_3\}\) form a
\(C_0r_0\)-net of \(M\).

The scale requirements in
Lemma~\ref{lem:coarse-output-generates-window-oracle} are checked in the same
way. The bound \(r_0\le c_{\ora}r_\rmp\) follows from
\eqref{eq:two-round-range}. For the concentration requirement, decrease
\(c_{\tref{thm:two-round-extraction}}\) so that
\(\phi(c_0r_0)\ge \phi(c_{\tref{thm:two-round-extraction}}r)\). Since
\eqref{eq:two-round-range} gives \(r_\rmp^2\gtrsim r_0^2=r\ge r^2\),
\eqref{eq:two-round-scale} implies
\[
\sp n\phi(c_0r_0)r_\rmp^2\ge C_{\ora}\Lambda\log n
\]
after increasing \(C_{\tref{thm:two-round-extraction}}\). Applying the lemma
with \(V_0=V_3\), \(W=V_4\), and coarse scale \(r_0\), we obtain a
\fuzzywindoworacle{} \(\mathcal O_{\win}\) on \(V_4\), with fixed margins
\[
\alpha_{\win}:=\frac{\ell_\rmp}{128L_\rmp},
\qquad
\lambda_{\win}:=\frac18.
\]

\step{Round 2: fine oracle extraction}
Apply Theorem~\ref{thm:three-block-extraction} in the \oracleroute{} to
\(V_4,V_5,V_6\) at target scale \(r\), using this oracle on \(V_4\). The
internal-average scale is
\[
\rho_\circ=K_\star^{-1}r.
\]
The range condition of Theorem~\ref{thm:three-block-extraction} follows from
\eqref{eq:two-round-range}, since \(r\le r_0\). Its scale condition follows
from \eqref{eq:two-round-scale} after decreasing
\(c_{\tref{thm:two-round-extraction}}\), because then
\(\phi(\rho_\circ/3)\ge \phi(c_{\tref{thm:two-round-extraction}}r)\) and
\(\rho_\circ^2=K_\star^{-2}r^2\). Therefore, with probability \(1-o(1)\), the
fine run outputs an estimator \(\widehat{\rmp}(0)\) and, for every
\(v\in V_6\), a set \(U_v\subseteq V_6\) satisfying
\[
v\in U_v,
\qquad
U_v
\text{ is a }(K_\star^2 r,0)\text{-cluster with center }X_v,
\]
\[
|U_v|\ge \frac12\phi(r)n,
\qquad
|\widehat{\rmp}(0)-\rmp(0)|\le r.
\]
	Taking \(C'_{\tref{thm:two-round-extraction}}\ge K_\star^2\) gives the stated
	cluster-radius form. The running time and the number of distinct sets are the
	sum of the two three-block costs and the oracle-construction cost, computed
	over the distinct coarse clusters. Since the coarse run has internal scale
	comparable to \(r\), and the fine run has internal scale comparable to \(r\),
	while \(r\le r_0\le1\), these costs are absorbed into the complexity stated in
	the theorem.

\step{Probability and independence}
The first coarse extraction, the occupancy event on \(V_3\), the oracle
construction using edges from \(V_3\) to \(V_4\), and the final oracle
extraction each fail with probability \(o(1)\). The edge sets used in these
steps are disjoint: the oracle uses \(V_3\)-to-\(V_4\) edges, while the fine
oracle extraction uses internal \(V_4\) edges, then \(V_4\)-to-\(V_5\) edges,
and finally \(V_5\)-to-\(V_6\) edges. Thus the conditional concentration
arguments from the three-block theorem apply without interference. A union
bound gives total success probability \(1-o(1)\).
\end{proof}

\begin{rem}
The first round is used only to construct the \fuzzywindoworacle{} on \(V_4\).
The final clusters are produced entirely by the second, oracle-assisted run on
\(V_4,V_5,V_6\). The choice \(r_0=\sqrt r\) is what makes the local-route coarse
condition at scale \(r_0\) match the oracle-route condition at scale \(r\).
\end{rem}

A repacking of the two-round procedure gives simultaneous cluster extraction for every vertex in the graph, as follows.

\begin{theor}\label{theor:main}
There exist constants
\[
c_{\tref{theor:main}},
\qquad
C_{\tref{theor:main}},
\qquad
C'_{\tref{theor:main}},
\qquad
\Lambda_{\tref{theor:main}}>0,
\]
depending only on the model parameters, such that the following holds whenever
\(\Lambda\ge\Lambda_{\tref{theor:main}}\). Let \(\bfV\) be a vertex set of
size \(n\), and let \(r>0\). Assume
\[
0<\sqrt r
\le
c_{\tref{theor:main}}
\min\{1,r_\mu,r_\rmp\},
\]
and
\[
\sp n\,\phi(r)r^2
\ge
C_{\tref{theor:main}}\Lambda\log n.
\]
Then there is a procedure which, with probability \(1-o(1)\), runs in time
$$
\exp\left(
        C_{\tref{theor:main}}
        \log^2( \tfrac{3}{r\phi(r)}) \tfrac{1}{\sp r^2}
    \right)n \Lambda\log n ,
$$
outputs an
estimator \(\widehat{\rmp}(0)\) and, for every \(v\in\bfV\), a set
\(U_v\subseteq\bfV\) such that
\[
v\in U_v,
\qquad
U_v
\text{ is a }
\left(C'_{\tref{theor:main}}r,0\right)
\text{-cluster with center }X_v,
\]
and
\[
|U_v|
\ge
\frac{1}{13}\phi(r)n.
\]
Further, the number of sets $U_v$ without counting multiplicity is at most
$$
\exp\left(
        C_{\tref{theor:main}}
        \log^2( \tfrac{3}{r\phi(r)}) \tfrac{1}{\sp r^2}
    \right)\,.
$$
Moreover,
\[
|\widehat{\rmp}(0)-\rmp(0)|
\le
C'_{\tref{theor:main}}r.
\]
\end{theor}
\begin{proof}
We may assume \(c_{\tref{thm:two-round-extraction}}\le1\), and set
\[
\bar r:=r/c_{\tref{thm:two-round-extraction}}.
\]
Split \(\bfV\) into six blocks of sizes differing by at most one. Apply
Theorem~\ref{thm:two-round-extraction} six times, cyclically permuting the
blocks so that each block is used once as the output block \(V_6\), and use
target scale \(\bar r\) in each run. The assumptions of
Theorem~\ref{thm:two-round-extraction} follow from the assumptions of this theorem
after decreasing \(c_{\tref{theor:main}}\) and increasing
\(C_{\tref{theor:main}}\): indeed
\[
c_{\tref{thm:two-round-extraction}}\bar r=r,
\]
so the scale condition there is exactly the main scale condition, up to the
constant change.

The six success events have total failure probability \(o(1)\). The runtime of
each call is bounded by the displayed main-theorem runtime because
\(c_{\tref{thm:two-round-extraction}}\bar r=r\), \(\bar r\ge r\), and the factor
\(\bar r^{-2}\) differs from \(r^{-2}\) only by a model-dependent constant.
The same comparison applies to the number of distinct output sets.

For a vertex in an output block of size at least \(n/6-1\),
Theorem~\ref{thm:two-round-extraction} gives a cluster of size at least
\[
\frac12\phi(\bar r)(n/6-1)
\ge
\frac1{13}\phi(r)n
\]
for all large \(n\), after increasing constants if necessary, since
\(\bar r\ge r\) and \(\phi\) is nondecreasing. The cluster radius and estimator
error are \(O(\bar r)=O(r)\), so they are absorbed into
\(C'_{\tref{theor:main}}r\). Use the estimator from any one of the six runs.
\end{proof}

\newpage

\bibliographystyle{alpha}
\bibliography{ref,references_metric_denoising}

\newpage

\appendix

\section{Lower bound}
\label{sec:lower-bound}

\begin{prop}[SBM partition lower bound]
\label{prop:sbm-partition-lower-bound}
There are universal constants \(c_{\rm sbm},C_{\rm sbm},c_{\rm fail}>0\) such
that the following holds for all sufficiently large \(n\). Consider the
symmetric stochastic block model with \(k\) blocks, where \(k\) is divisible by
\(4\): the latent labels \(X_1,\ldots,X_n\) are i.i.d. uniform on \([k]\), and
conditional on the labels the edges are independent Bernoulli random variables
with probability \(p\) within a block and \(q\) across distinct blocks. If
\[
    0<q<p\le c_{\rm sbm},
    \qquad
    p-q\le c_{\rm sbm}p,
    \qquad
    p\frac{n}{k}\ge C_{\rm sbm},
\]
and
\[
    \frac{(p-q)^2}{p}\frac{n}{k}\le c_{\rm sbm}\log n,
\]
then every estimator \(\widehat\Pi({\bf A})\) of the induced partition
\(\Pi(X_{[n]})\), where
\[
    \Pi(x_{[n]})
    :=
    \bigl\{\{u\in[n]:x_u=a\}:a\in[k]\bigr\}\setminus\{\emptyset\},
\]
based on the observed adjacency matrix \({\bf A}\), satisfies
\[
    \Prob\bigl(\widehat\Pi({\bf A})\neq\Pi(X_{[n]})\bigr)\ge c_{\rm fail}.
\]
\end{prop}

\begin{proof}[Proof of Lemma~\ref{lem:lower-bound}]
Set \(\eta=\phi(r_0)\). Choose \(k\) divisible by \(4\) such that
\[
    \frac{1}{8\eta}\le k\le \frac{1}{4\eta};
\]
this is possible after taking \(c_\phi\) small enough. Let
\[
    M=[k],
    \qquad
    {\rm d}(x,y)=r_0{\bf 1}\{x\neq y\},
    \qquad
    \mu(x)=1/k.
\]
For \(0<t\le r_0\), every open ball \(B(x,t)\) is the singleton \(\{x\}\), and
\[
    \mu(B(x,t))=1/k\ge \eta\ge \phi(t).
\]
Thus \((M,{\rm d},\mu)\) is lower-\(\phi\)-regular up to scale \(r_0\).

Fix small universal constants \(a,\gamma>0\), to be chosen below, and define
\[
    \rmp(t):=a+\gamma(r_0-t)_+.
\]
This link is non-increasing and locally bi-Lipschitz on \([0,r_0]\), with
bi-Lipschitz constants equal to \(\gamma\). Taking \(a,\gamma\) small also
ensures that \(0\le\rmp\le1\). Use the Bernoulli observation model
\[
    \mathbf F(t,u)={\bf 1}\{u\le \rmp(t)\}.
\]
Then the observed edge indicator has probability
\[
    p:=\mathsf{s}\rmp(0)=\mathsf{s}(a+\gamma r_0)
    \qquad\text{within a block}
\]
and
\[
    q:=\mathsf{s}\rmp(r_0)=\mathsf{s}a
    \qquad\text{between distinct blocks}.
\]
Choosing \(a,\gamma\) sufficiently small, with \(\gamma\ll a\), gives
\[
    0<q<p\le c_{\rm sbm},
    \qquad
    p-q\le c_{\rm sbm}p.
\]
Moreover, since \(k\asymp1/\eta\),
\[
    p\frac{n}{k}\asymp \mathsf{s}n\eta,
\]
and
\[
    \frac{(p-q)^2}{p}\frac{n}{k}
    \asymp
    \mathsf{s}n\eta r_0^2,
\]
with universal implicit constants depending only on \(a,\gamma\). Therefore,
after increasing \(C_0\) and decreasing \(c_0\), the assumptions of
Proposition~\ref{prop:sbm-partition-lower-bound} hold.

Let \(\widehat d\) be any distance estimator. It induces a partition estimator
\(\widehat\Pi\) by taking the connected components of the graph on \([n]\) with
an edge between \(u\) and \(v\) whenever \(\widehat d(u,v)<r_0/2\). If
\[
    \max_{u,v}|\widehat d(u,v)-\D{X_u}{X_v}|<r_0/2,
\]
then this thresholded partition is exactly \(\Pi(X_{[n]})\). Hence
\[
    \Prob\!\left(
    \max_{u,v}|\widehat d(u,v)-\D{X_u}{X_v}|\ge r_0/2
    \right)
    \ge
    \Prob\bigl(\widehat\Pi\neq\Pi(X_{[n]})\bigr)
    \ge c_{\rm fail}.
\]
Taking \(c_1=c_{\rm fail}\) completes the proof.
\end{proof}

The rest of this section is devoted to the proof of
 Proposition~\ref{prop:sbm-partition-lower-bound}. 
Throughout this section, \(c\) and \(C\) denote positive universal constants
whose values may change from line to line. For a partition \(\pi\) of \([n]\),
write \(\pi^{(v)}\) for the partition of \([n]\setminus\{v\}\) obtained by
deleting \(v\) from its block and then removing the empty block if one is
created. We also write
\[
    {\cal P}_{n,\le k}
    :=
    \{\pi:\pi\text{ is a partition of }[n]\text{ into at most }k
    \text{ nonempty blocks}\}.
\]
A partition estimator is a measurable map from the observed adjacency matrix to
\({\cal P}_{n,\le k}\). The global maximum a posteriori (MAP) partition
estimator is any measurable choice
\begin{align}
    \label{eq:def-sbm-partition-map}
    \widehat\Pi^{\rm MAP}({\bf A})
    \in
    \arg\max_{\pi\in{\cal P}_{n,\le k}}
    \Prob\left(\Pi(X_{[n]})=\pi\mid{\bf A}\right).
\end{align}
Since \(\widehat\Pi^{\rm MAP}\) maximizes the conditional success probability
given \({\bf A}\), it suffices to prove that this estimator fails with
probability bounded away from zero.

For \(x\in[k]\), put
\[
    B_x:=\{u\in[n]:X_u=x\},
    \qquad
    N_x:=|B_x|.
\]

\subsection{Good blocks}

We first isolate a constant-probability event on which many block sizes are close
to \(n/k\).

\begin{lemma}[Good blocks]
    \label{lem:sbm-good-blocks}
    There is a universal constant \(C_{\rm G}\ge1\) such that, whenever
    \(n/k\ge C_{\rm G}\), the event
    \[
        {\cal E}_{\cal G}:=\{|{\cal G}|\ge k/2\},
        \qquad
        {\cal G}:=
        \left\{x\in[k]:
        |N_x-n/k|\le\sqrt{C_{\rm G}n/k}\right\},
    \]
    has probability at least \(1/2\).
\end{lemma}
\begin{proof}
Each \(N_x\) is binomial with parameters \(n\) and \(1/k\). By Chernoff's
inequality, if \(n/k\ge C_{\rm G}\), then
\[
    \Prob\left(|N_x-n/k|>\sqrt{C_{\rm G}n/k}\right)
    \le 2\exp(-C_{\rm G}/3).
\]
Hence \(\E|{\cal G}^c|\le 2k\exp(-C_{\rm G}/3)\). Markov's inequality gives
\[
    \Prob\left(|{\cal G}^c|>4k\exp(-C_{\rm G}/3)\right)\le \frac12.
\]
Taking \(C_{\rm G}\) large enough that \(4\exp(-C_{\rm G}/3)\le1/2\) proves the
claim.
\end{proof}

In the rest of the proof we work on \({\cal E}_{\cal G}\), which has probability
at least \(1/2\). We choose \(C_{\rm sbm}\) large enough relative to \(C_{\rm G}\)
so that on \({\cal E}_{\cal G}\) every block in \({\cal G}\) contains at least
two vertices. On \({\cal E}_{\cal G}\), fix disjoint sets
\[
    I_1,I_2\subseteq{\cal G},
    \qquad
    |I_1|=|I_2|=k/4,
\]
and define
\[
    V_*:=\{v\in[n]:X_v\in I_1\}.
\]
For \(C_{\rm sbm}\) sufficiently large, \(|V_*|\ge cn\) on \({\cal E}_{\cal G}\).

\subsection{One-genie comparison}

Fix \(v\in[n]\), and write
\[
    [n]^{(v)}:=[n]\setminus\{v\},
    \qquad
    X^{(v)}:=(X_i)_{i\in[n]^{(v)}}.
\]
We condition on a realization \(x^{(v)}\) of \(X^{(v)}\). If this information is
revealed, then the MAP estimator of \(X_v\) is
\[
    \widehat x_v({\bf A},x^{(v)})
    \in
    \arg\max_{x_v\in[k]}
    \Prob(X_v=x_v\mid{\bf A},X^{(v)}=x^{(v)}).
\]
Since \(X_v\) is independent of \(X^{(v)}\) and uniformly distributed on \([k]\),
Bayes' rule gives
\begin{align*}
    \Prob(X_v=x_v\mid{\bf A},X^{(v)}=x^{(v)})
    &=
    \frac{\Prob(X_v=x_v\mid X^{(v)}=x^{(v)})}
         {\Prob({\bf A}\mid X^{(v)}=x^{(v)})}
    \Prob({\bf A}\mid X_v=x_v,X^{(v)}=x^{(v)})\\
    &=
    \frac{1}{k\Prob({\bf A}\mid X^{(v)}=x^{(v)})}
    \Prob({\bf A}\mid X_v=x_v,X^{(v)}=x^{(v)}).
\end{align*}
Thus the genie MAP maximizes the conditional likelihood.

For \(x\in[k]\), define
\[
    N_x^{(v)}
    :=
    |\{i\in[n]^{(v)}:X_i=x\}|,
    \qquad
    M_x^{(v)}
    :=
    \sum_{\substack{i\in[n]^{(v)}\\ X_i=x}} A_{v,i},
\]
and write \(\vec M^{(v)}=(M_x^{(v)})_{x\in[k]}\). By conditional independence of
the edges, and since edges not incident to \(v\) contribute a multiplicative
factor independent of the candidate label \(x_v\),
\[
\Prob({\bf A}\mid X_v=x_v,X^{(v)}=x^{(v)})
\propto
\prod_{i\in[n]^{(v)}}
\Prob(A_{v,i}\mid X_v=x_v,X_i=x_i).
\]
The product on the right is
\begin{align*}
\prod_{x\in[k]}
q^{M_x^{(v)}}(1-q)^{N_x^{(v)}-M_x^{(v)}}
\cdot
\left(\frac pq\right)^{M_{x_v}^{(v)}}
\cdot
\left(\frac{1-p}{1-q}\right)^{N_{x_v}^{(v)}-M_{x_v}^{(v)}}
\propto
\exp\bigg(
    \log\left(\frac{p(1-q)}{q(1-p)}\right)
    \left(M_{x_v}^{(v)}-\alpha N_{x_v}^{(v)}\right)
\bigg),
\end{align*}
where the proportionality hides only factors independent of \(x_v\), and
\[
    \alpha
    :=
    \alpha(p,q)
    :=
    \frac{\log\left(\frac{1-q}{1-p}\right)}
         {\log\left(\frac{p(1-q)}{q(1-p)}\right)}.
\]
The logarithmic prefactor is positive since \(p>q\). Hence, for
\(\vec m=(m_x)_{x\in[k]}\), the normalized score
\[
    S_x^{(v)}(\vec m):=m_x-\alpha N_x^{(v)}
\]
satisfies
\[
    \widehat x_v({\bf A},x^{(v)})
    \in
    \arg\max_{x\in[k]} S_x^{(v)}(\vec M^{(v)}).
\]

We also need one partition-specific prior identity. If \(\rho\) is a partition
of \([n]^{(v)}\), \(B\in\rho\), and
\[
    \rho\oplus_v B:=(\rho\setminus\{B\})\cup\{B\cup\{v\}\},
\]
then
\begin{align}
    \label{eq:sbm-existing-block-prior}
    \Prob\left(
        \Pi(X_{[n]})=\rho\oplus_v B
        \mid
        \Pi(X_{[n]})^{(v)}=\rho
    \right)
    =
    \frac1k.
\end{align}
Indeed, after \(\Pi(X_{[n]})^{(v)}=\rho\) is revealed, joining any fixed existing
block means that \(X_v\) equals the latent label of that block, which has
conditional probability \(1/k\). Thus all existing-block alternatives have the
same conditional prior weight.

\begin{lemma}[Score normalization]
    \label{lem:sbm-alpha-bound}
    After decreasing \(c_{\rm sbm}\), if
    \(0<q<p\le c_{\rm sbm}\) and \(p-q\le c_{\rm sbm}p\), then
    \[
        \frac p2\le \alpha\le p.
    \]
\end{lemma}
\begin{proof}
Write \(\Delta=p-q\). We use
\[
    \frac r{1+r}\le \log(1+r)\le r,
    \qquad r\ge0.
\]
Applying this with \(r=\Delta/(1-p)\) gives
\[
    \Delta
    \le
    \log\left(\frac{1-q}{1-p}\right)
    \le
    \frac{\Delta}{1-p}.
\]
Similarly,
\[
    \frac{\Delta}{p}
    \le
    \log\left(\frac pq\right)
    \le
    \frac{\Delta}{q}
    \le
    \frac{\Delta}{(1-c_{\rm sbm})p},
\]
where the last inequality uses \(q\ge(1-c_{\rm sbm})p\). Since
\[
    \log\left(\frac{p(1-q)}{q(1-p)}\right)
    =
    \log\left(\frac pq\right)
    +
    \log\left(\frac{1-q}{1-p}\right),
\]
the upper bound follows from
\[
    \alpha
    \le
    \frac{\Delta/(1-p)}
         {\Delta/p+\Delta/(1-p)}
    =
    p.
\]
For the lower bound,
\[
    \alpha
    \ge
    \frac{\Delta}
         {\Delta/(1-p)+\Delta/((1-c_{\rm sbm})p)}
    =
    p\frac{(1-c_{\rm sbm})(1-p)}{1-c_{\rm sbm}p}
    \ge \frac p2
\]
after decreasing \(c_{\rm sbm}\).
\end{proof}

\subsection{Bad rows}

For \(v\in V_*\), define the bad row event \({\cal E}_v\) by
\[
    M_{X_v}^{(v)}\le pN_{X_v}^{(v)}
\]
and
\[
    \max_{j\in I_2}\bigl(M_j^{(v)}-\alpha N_j\bigr)
    >
    pN_{X_v}^{(v)}-\alpha N_{X_v}^{(v)}.
\]

\begin{lemma}[Bad row gives a better partition]
    \label{lem:sbm-bad-row-better-partition}
    Fix a realization \(x_{[n]}^{\rm true}\) of \(X_{[n]}\) satisfying
    \({\cal E}_{\cal G}\), and let \(v\in V_*\). On \({\cal E}_v\), the true
    partition \(\Pi(x_{[n]}^{\rm true})\) is not a global partition MAP solution.
\end{lemma}
\begin{proof}
On \({\cal E}_v\), choose \(x\in I_2\) such that
\[
    M_x^{(v)}-\alpha N_x
    >
    pN_{X_v}^{(v)}-\alpha N_{X_v}^{(v)}.
\]
Since \(x\in I_2\) and \(X_v\in I_1\), we have \(N_x^{(v)}=N_x\). Together with
\(M_{X_v}^{(v)}\le pN_{X_v}^{(v)}\), this gives
\[
    M_x^{(v)}
    >
    M_{X_v}^{(v)}
    +
    \alpha\bigl(N_x^{(v)}-N_{X_v}^{(v)}\bigr),
\]
or equivalently
\[
    S_x^{(v)}(\vec M^{(v)})
    >
    S_{X_v}^{(v)}(\vec M^{(v)}).
\]
Let
\[
    \rho:=\Pi(x_{[n]}^{\rm true})^{(v)},
    \qquad
    B_{\rm true}:=\{i\in[n]^{(v)}:X_i=X_v\},
    \qquad
    B_x:=\{i\in[n]^{(v)}:X_i=x\}.
\]
Both \(B_{\rm true}\) and \(B_x\) are nonempty blocks of \(\rho\), because
\(X_v,x\in{\cal G}\) and blocks in \({\cal G}\) contain at least two vertices.
The true partition is \(\rho\oplus_v B_{\rm true}\). Let
\[
    \pi':=\rho\oplus_v B_x.
\]
The strict score inequality says that, with \(X^{(v)}\) fixed, the likelihood is
larger when \(v\) is assigned to the block \(B_x\) than when it is assigned to
its realized block \(B_{\rm true}\). The two alternatives have the same
restriction \(\rho\) on \([n]^{(v)}\), and \eqref{eq:sbm-existing-block-prior}
gives them the same conditional prior weight. Hence
\[
    \Prob(\Pi(X_{[n]})=\pi'\mid{\bf A})
    >
    \Prob(\Pi(X_{[n]})=\Pi(x_{[n]}^{\rm true})\mid{\bf A}),
\]
so the true partition is not a global MAP solution.
\end{proof}

\begin{lemma}[A local binomial point mass]
    \label{lem:sbm-binomial-point-mass}
    There are universal constants \(c_{\rm pm},C_{\rm pm},c'_{\rm pm}>0\) with
    the following property. Fix a realization of \(X_{[n]}\). Let \(v\in[n]\),
    set \(y:=X_v\), and fix \(x\neq y\). Suppose
    \[
        |N_x^{(v)}-n/k|\le 2\sqrt{C_{\rm G}n/k},
        \qquad
        |N_y^{(v)}-n/k|\le 2\sqrt{C_{\rm G}n/k}.
    \]
    If \(0<q<p\le c'_{\rm pm}\), \(p-q\le c'_{\rm pm}p\),
    \(n/k\ge C_{\rm pm}C_{\rm G}\), and \(p(n/k)\ge C_{\rm pm}\), then any
    integer \(t\) satisfying
    \[
        \left|t-\left(pN_y^{(v)}+4p\sqrt{C_{\rm G}n/k}\right)\right|\le2
    \]
    also satisfies
    \[
        \Prob(M_x^{(v)}=t\mid X_{[n]})
        \ge
        \frac{c_{\rm pm}}{\sqrt{p\,n/k}}
        \exp\left[
            -C_{\rm pm}\left(\frac{(p-q)^2}{p}\frac{n}{k}+C_{\rm G}\right)
        \right].
    \]
\end{lemma}
\begin{proof}
Conditional on \(X_{[n]}\), \(M_x^{(v)}\sim{\rm Bin}(N_x^{(v)},q)\). Let
\(\mu:=qN_x^{(v)}\). The assumptions imply \(q\ge(1-c'_{\rm pm})p\) and
\(N_x^{(v)}\ge(n/k)/2\), after increasing \(C_{\rm pm}\). Hence
\(\mu\ge c p(n/k)\).

The standard local binomial lower bound gives, for every integer \(t\) with
\(|t-\mu|\le \mu/2\),
\[
    \Prob(M_x^{(v)}=t\mid X_{[n]})
    \ge
    \frac{c}{\sqrt{\mu}}
    \exp\left(-C\frac{(t-\mu)^2}{\mu}\right).
\]
Using the assumptions on \(N_x^{(v)}\) and \(N_y^{(v)}\),
\[
    |t-\mu|
    \le
    (p-q)\frac{n}{k}
    +8p\sqrt{C_{\rm G}n/k}
    +2.
\]
Increasing \(C_{\rm pm}\), using \(p(n/k)\ge C_{\rm pm}\), and decreasing
\(c'_{\rm pm}\) gives \(|t-\mu|\le \mu/2\) and
\[
    \frac{(t-\mu)^2}{\mu}
    \le
    C\left(\frac{(p-q)^2}{p}\frac{n}{k}+C_{\rm G}\right).
\]
Since \(\mu\asymp p(n/k)\), the claimed estimate follows.
\end{proof}

\subsection{Second moment}

For the second moment estimates, condition on a realization of \(X_{[n]}\)
satisfying \({\cal E}_{\cal G}\). For \(v\in V_*\), define
\[
    {\cal E}'_v
    :=
    \left\{
        \max_{j\in I_2}\bigl(M_j^{(v)}-\alpha N_j\bigr)
        >
        pN_{X_v}^{(v)}-\alpha N_{X_v}^{(v)}
    \right\},
    \qquad
    \beta_v:=\Prob({\cal E}'_v\mid X_{[n]}).
\]
Thus \({\cal E}_v\) is the intersection of the within-block lower-tail event
\(\{M_{X_v}^{(v)}\le pN_{X_v}^{(v)}\}\) and the competing-block event
\({\cal E}'_v\). Let
\[
    Y:=\sum_{v\in V_*}{\bf 1}_{{\cal E}_v}.
\]

\begin{lemma}[Second moment reduction]
    \label{lem:sbm-second-moment-reduction}
    On \({\cal E}_{\cal G}\),
    \[
        \E[Y^2\mid X_{[n]}]
        \le
        \E[Y\mid X_{[n]}]^2
        +
        \E[Y\mid X_{[n]}]
        +
        C\sum_{v\in V_*}\beta_v^2.
    \]
    Consequently, if
    \[
        \sum_{v\in V_*}\beta_v^2
        =
        o\left(\E[Y\mid X_{[n]}]^2\right)
        +
        O\left(\E[Y\mid X_{[n]}]\right)
    \]
    and \(\E[Y\mid X_{[n]}]\to\infty\), then
    \[
        \Prob(Y>0\mid X_{[n]})\to1.
    \]
\end{lemma}
\begin{proof}
If \(v,w\in V_*\) and \(X_v\neq X_w\), then \({\cal E}_v\) and \({\cal E}_w\)
depend on disjoint sets of edges. Indeed, \({\cal E}_v\) only uses edges from
\(v\) to \(B_{X_v}\setminus\{v\}\) and to \(\bigcup_{j\in I_2}B_j\), and
similarly for \(w\). Since \(X_v,X_w\in I_1\) and \(I_1\cap I_2=\emptyset\), the
edge \(A_{vw}\) is not used by either event. Hence these events are conditionally
independent given \(X_{[n]}\).

The only nontrivial dependence comes from pairs \(v\neq w\) in the same selected
block. For such a pair, condition on \(A_{vw}\). The remaining edge sets used by
\({\cal E}_v\) and \({\cal E}_w\) are disjoint. Define
\[
    p_1(v,w):=\Prob({\cal E}_v\mid X_{[n]},A_{vw}=1),
    \qquad
    p_0(v,w):=\Prob({\cal E}_v\mid X_{[n]},A_{vw}=0).
\]
By symmetry inside the fixed block, the same two conditional probabilities apply
to \({\cal E}_w\). Therefore the excess over the independent product is exactly
\[
    p(1-p)(p_1(v,w)-p_0(v,w))^2.
\]

Let \(m:=N_{X_v}^{(v)}\) and \(t:=\lfloor pm\rfloor\). Write
\[
    M_{X_v}^{(v)}=A_{vw}+R_v,
    \qquad
    R_v\sim{\rm Bin}(m-1,p).
\]
The event \({\cal E}'_v\) is independent of both \(A_{vw}\) and \(R_v\). Thus
\[
    p_1(v,w)=\Prob(R_v\le t-1)\beta_v,
    \qquad
    p_0(v,w)=\Prob(R_v\le t)\beta_v,
\]
and hence
\[
    |p_1(v,w)-p_0(v,w)|=\beta_v\Prob(R_v=t).
\]
The standard local binomial upper bound gives
\[
    \Prob(R_v=t)\le \frac{C}{\sqrt{mp(1-p)}}.
\]
Thus
\[
    p(1-p)(p_1(v,w)-p_0(v,w))^2\le \frac{C\beta_v^2}{m}.
\]
Since \(m\asymp n/k\) on \({\cal E}_{\cal G}\), summing over ordered same-block
pairs gives an excess bounded by \(C\sum_{v\in V_*}\beta_v^2\). Combining this
with the different-block factorization yields the displayed second-moment bound.
The final conclusion follows from Paley--Zygmund:
\[
    \Prob(Y>0\mid X_{[n]})
    \ge
    \frac{\E[Y\mid X_{[n]}]^2}{\E[Y^2\mid X_{[n]}]}.
\]
\end{proof}

\begin{cor}
    \label{cor:sbm-second-moment-condition}
    Work on \({\cal E}_{\cal G}\), and assume the hypotheses of
    Proposition~\ref{prop:sbm-partition-lower-bound}, with \(C_{\rm sbm}\)
    sufficiently large and \(c_{\rm sbm}\) sufficiently small. Then
    \[
        \sum_{v\in V_*}\beta_v^2
        \le
        C\E[Y\mid X_{[n]}],
    \]
    and
    \[
        \E[Y\mid X_{[n]}]
        \ge
        cn\min\left\{
            1,\,
            \frac{k}{\sqrt{p\,n/k}}
            \exp\left[
                -C\left(\frac{(p-q)^2}{p}\frac{n}{k}+C_{\rm G}\right)
            \right]
        \right\}.
    \]
    In particular, \(\E[Y\mid X_{[n]}]\to\infty\), and therefore
    \[
        \Prob(Y>0\mid X_{[n]})\to1.
    \]
\end{cor}
\begin{proof}
The two parts of \({\cal E}_v\) depend on disjoint sets of edges and are
independent conditional on \(X_{[n]}\). Since
\(M_{X_v}^{(v)}\sim{\rm Bin}(N_{X_v}^{(v)},p)\), the binomial median bound gives
\[
    \Prob(M_{X_v}^{(v)}\le pN_{X_v}^{(v)}\mid X_{[n]})\ge c.
\]
Therefore
\[
    \Prob({\cal E}_v\mid X_{[n]})\ge c\beta_v.
\]
Since \(0\le\beta_v\le1\),
\[
    \sum_{v\in V_*}\beta_v^2
    \le
    \sum_{v\in V_*}\beta_v
    \le
    C\E[Y\mid X_{[n]}].
\]

It remains to lower bound \(\E[Y\mid X_{[n]}]\). Fix \(v\in V_*\) and
\(j\in I_2\). Since \(X_v,j\in{\cal G}\), Lemma~\ref{lem:sbm-binomial-point-mass}
applies after adjusting universal constants. By Lemma~\ref{lem:sbm-alpha-bound},
\(\alpha\le p\). On \({\cal E}_{\cal G}\),
\[
    N_j-N_{X_v}^{(v)}
    \le
    4\sqrt{C_{\rm G}n/k},
\]
after increasing \(C_{\rm G}\) and using \(p(n/k)\ge C_{\rm sbm}\). Hence
\[
    pN_{X_v}^{(v)}+\alpha\bigl(N_j-N_{X_v}^{(v)}\bigr)
    \le
    pN_{X_v}^{(v)}+4p\sqrt{C_{\rm G}n/k}.
\]
Choose an integer \(t_{v,j}\) such that
\[
    t_{v,j}
    >
    pN_{X_v}^{(v)}+\alpha\bigl(N_j-N_{X_v}^{(v)}\bigr),
    \qquad
    \left|t_{v,j}-\left(pN_{X_v}^{(v)}+4p\sqrt{C_{\rm G}n/k}\right)\right|\le2.
\]
Then \(\{M_j^{(v)}=t_{v,j}\}\subseteq{\cal E}'_v\). Lemma
\ref{lem:sbm-binomial-point-mass} gives
\[
    \Prob(M_j^{(v)}=t_{v,j}\mid X_{[n]})\ge \varrho,
\]
where
\[
    \varrho
    :=
    \frac{c}{\sqrt{p\,n/k}}
    \exp\left[
        -C\left(\frac{(p-q)^2}{p}\frac{n}{k}+C_{\rm G}\right)
    \right].
\]
The events \(\{M_j^{(v)}=t_{v,j}\}\), \(j\in I_2\), depend on disjoint edge sets,
so
\[
    \beta_v
    \ge
    1-(1-\varrho)^{k/4}
    \ge
    c\min\{1,k\varrho\}.
\]
Together with \(|V_*|\ge cn\) and
\(\Prob({\cal E}_v\mid X_{[n]})\ge c\beta_v\), this proves the stated lower
bound on \(\E[Y\mid X_{[n]}]\).

Finally, the assumption
\[
    \frac{(p-q)^2}{p}\frac{n}{k}\le c_{\rm sbm}\log n
\]
implies, for \(c_{\rm sbm}\) sufficiently small, that the displayed lower bound
for \(\E[Y\mid X_{[n]}]\) tends to infinity. Lemma
\ref{lem:sbm-second-moment-reduction} then gives
\(\Prob(Y>0\mid X_{[n]})\to1\).
\end{proof}

\begin{proof}[Proof of Proposition~\ref{prop:sbm-partition-lower-bound}]
By Lemma~\ref{lem:sbm-good-blocks}, the event \({\cal E}_{\cal G}\) has
probability at least \(1/2\). On this event, fix \(I_1,I_2\) and \(V_*\) as
above. By Corollary~\ref{cor:sbm-second-moment-condition}, conditional on any
realization of \(X_{[n]}\) satisfying \({\cal E}_{\cal G}\),
\[
    \Prob(Y>0\mid X_{[n]})\to1.
\]
On \(\{Y>0\}\), there exists \(v\in V_*\) such that \({\cal E}_v\) occurs. By
Lemma~\ref{lem:sbm-bad-row-better-partition}, the true partition is not a global
partition MAP solution. Therefore
\[
    \Prob\left(\widehat\Pi^{\rm MAP}({\bf A})\neq\Pi(X_{[n]})\right)
    \ge
    \Prob({\cal E}_{\cal G})
    \inf_{x_{[n]}\text{ satisfying }{\cal E}_{\cal G}}
    \Prob(Y>0\mid X_{[n]}=x_{[n]})
    \ge c
\]
for all sufficiently large \(n\). Since the global partition MAP is Bayes-optimal
for partition recovery, the same lower bound holds for every partition estimator.
This proves the proposition, after setting \(c_{\rm fail}=c\).
\end{proof}
\section{Proof of the coarse fuzzy-window oracle lemma}
\label{app:coarse-window-oracle-proof}

\begin{proof}[Proof of Lemma~\ref{lem:coarse-output-generates-window-oracle}]
Set
\[
m_0:=\frac12\phi(c_0r)n,
\qquad
\varepsilon_{\rm link}:=
\sqrt{\frac{\Lambda\log n}{\sp m_0}}.
\]
The assumption
\[
\sp n\phi(c_0r)r_\rmp^2\ge C_{\ora}\Lambda\log n
\]
implies, after increasing \(C_{\ora}\), that
\[
\varepsilon_{\rm link}\le c_{\ora}r_\rmp.
\]
By Lemma~\ref{lem:fixed-v-navigation}, a union bound over
\(|V_0||W|\le n^2\) gives, with conditional probability
\(1-n^{-\Omega(\Lambda)}\),
\[
\left|\cn{A_z}{v}-\acn{A_z}{v}\right|
\le
\varepsilon_{\rm link}
\qquad
(z\in V_0,\ v\in W).
\]
We work on this event.

We first record the deterministic inclusions
\[
\left\{
v\in W:
\D{X_v}{X_z}\le \frac{\ell_\rmp}{64L_\rmp}r_\rmp
\right\}
\subseteq
W_z
\subseteq
\left\{
v\in W:
\D{X_v}{X_z}<\frac{1}{16}r_\rmp
\right\}
\qquad (z\in V_0).
\]
Fix \(z\in V_0\). Since \(A_z\) is a \((C_0r,0)\)-cluster with center \(X_z\),
every \(u\in A_z\) satisfies \(\D{X_u}{X_z}<C_0r\).
If
\[
\D{X_v}{X_z}\le \frac{\ell_\rmp}{64L_\rmp}r_\rmp,
\]
then, after taking \(c_{\ora}\) small enough, all distances
\(\D{X_u}{X_v}\) lie in the local bi-Lipschitz window and
\[
\acn{A_z}{v}
\ge
\rmp(0)-L_\rmp C_0r-\frac{\ell_\rmp}{64}r_\rmp.
\]
Using \(|\widehat{\rmp}(0)-\rmp(0)|\le C_0r\) and the concentration event,
\[
\cn{A_z}{v}
\ge
\widehat{\rmp}(0)
-(L_\rmp C_0+C_0)r-\varepsilon_{\rm link}
-\frac{\ell_\rmp}{64}r_\rmp.
\]
The choices of \(c_{\ora}\) and \(C_{\ora}\) ensure that the last three error
terms are at most \(\ell_\rmp r_\rmp/32\), so \(v\in W_z\).

Conversely, if
\[
\D{X_v}{X_z}\ge \frac{1}{16}r_\rmp,
\]
then, again taking \(c_{\ora}\) small enough,
\[
0<
\frac{1}{16}r_\rmp-C_0r
\le
r_\rmp.
\]
For every \(u\in A_z\), monotonicity and the lower local bi-Lipschitz bound give
\[
\rmp(\D{X_u}{X_v})
\le
\rmp(0)
-
\ell_\rmp\left(\frac{1}{16}r_\rmp-C_0r\right).
\]
Averaging and using the estimator and concentration errors,
\[
\cn{A_z}{v}
\le
\widehat{\rmp}(0)
-\frac{\ell_\rmp}{16}r_\rmp
+(\ell_\rmp C_0+C_0)r+\varepsilon_{\rm link}.
\]
With \(c_{\ora}\) small and \(C_{\ora}\) large, the final two error terms are at
most \(\ell_\rmp r_\rmp/64\), and hence
\[
\cn{A_z}{v}
<
\widehat{\rmp}(0)-\frac{\ell_\rmp}{32}r_\rmp.
\]
Thus \(v\notin W_z\), proving the displayed inclusions.

We now verify the fuzzy-oracle property. If \(\mathcal O_{\win}(v,w)=1\), then
there is \(z\in V_0\) with \(v,w\in W_z\). The outer inclusion gives
\[
\D{X_v}{X_z}<\frac1{16}r_\rmp,
\qquad
\D{X_w}{X_z}<\frac1{16}r_\rmp,
\]
so \(\D{X_v}{X_w}<r_\rmp/8=\lambda_{\win}r_\rmp\). Therefore
\[
\D{X_v}{X_w}>\lambda_{\win}r_\rmp
\quad\Longrightarrow\quad
\mathcal O_{\win}(v,w)=0.
\]

Conversely, suppose
\[
\D{X_v}{X_w}\le
\alpha_{\win}r_\rmp
=
\frac{\ell_\rmp}{128L_\rmp}r_\rmp.
\]
Choose \(z\in V_0\) with \(\D{X_v}{X_z}\le C_0r\), using the \(C_0r\)-net
assumption. For \(c_{\ora}\) small enough,
\[
C_0r\le \frac{\ell_\rmp}{128L_\rmp}r_\rmp,
\]
and hence
\[
\D{X_v}{X_z}\le \frac{\ell_\rmp}{64L_\rmp}r_\rmp,
\qquad
\D{X_w}{X_z}\le \frac{\ell_\rmp}{64L_\rmp}r_\rmp.
\]
The inner inclusion gives \(v,w\in W_z\), so
\(\mathcal O_{\win}(v,w)=1\). This proves that \(\mathcal O_{\win}\) is an
\(\bigl(\ell_\rmp/(128L_\rmp),1/8\bigr)\)-\fuzzywindoworacle{}.
\end{proof}

\section{From cluster extraction to distance estimation} \label{sec:cluster-to-distance}
We now explain how the clusters from Theorem~\ref{theor:main} yield distance
estimates. The arguments in this section are simply variants from what appeared in \cite{HuangJiradilokMossel2024GeometryRGG,HuangJiradilokMossel2025RiemannianMetrics,FeffermanMartyRen2025NoisyDistances,HuangJiradilokMossel2026BeyondVolumetric}. Here we give a self-contained presentation, with the necessary modifications to fit the current setting.

If \(U_v\) and \(U_w\) are exact clusters
around \(X_v\) and \(X_w\), then every pair \(u\in U_v\), \(u'\in U_w\) has
\[
\D{X_u}{X_{u'}}
=
\D{X_v}{X_w}+O(r).
\]
Thus the \pairaverage{} \(\pav{U_v,U_w}\) estimates
\(\rmp(\D{X_v}{X_w})\), up to the cluster radius and the fluctuation of the
edge average.  If \(\rmp\) is known and invertible on the relevant range, this
gives distance estimates.

We use this in conjunction with the uniform pair-average event
\[
\Enn{\bfV}{\lambda_{\rm pair}}{m_{\rm pair}},
\]
where
\[
m_{\rm pair}:=\left\lceil \frac{1}{13}\phi(r)n\right\rceil,
\qquad
\lambda_{\rm pair}:=r.
\]
Under the scale condition of Theorem~\ref{theor:main}, Lemma~\ref{lem:uniform-pair-average} imply that
\[
\Enn{\bfV}{\lambda_{\rm pair}}{m_{\rm pair}}
\]
holds with probability \(1-o(1)\), after increasing the constant in
Theorem~\ref{theor:main}.  Thus, with high probability, both the cluster
extraction guarantee and this pair-average event hold simultaneously.

\begin{lemma}[Pair averages from exact clusters]
\label{lem:pair-averages-from-exact-clusters}
Let \(W\subseteq\bfV\), and suppose that for every \(v\in W\) we are given
\(U_v\subseteq W\) satisfying
\[
|U_v|\ge m,
\qquad
X_u\in B(X_v,r_{\rm cl})
\quad (u\in U_v).
\]
Assume the event
\[
\Enn{W}{\lambda}{m}
\]
holds.  Set
\[
\varepsilon_{\rm pair}:=2L_\rmp r_{\rm cl}+\lambda.
\]
Then, for every \(v,w\in W\) such that
\[
\D{X_v}{X_w}+2r_{\rm cl}\le r_\rmp,
\]
we have
\[
\left|
\pav{U_v,U_w}
-
\rmp(\D{X_v}{X_w})
\right|
\le
\varepsilon_{\rm pair}.
\]
If \(\rmp\) is globally \(L_\rmp\)-Lipschitz on
\([0,\operatorname{diam}(M)]\), then the same bound holds for every
\(v,w\in W\).
\end{lemma}

\begin{proof}
Fix \(v,w\in W\), and set
\[
d_{vw}:=\D{X_v}{X_w}.
\]
For every \((u,u')\in\mathcal D(U_v,U_w)\), the triangle inequality gives
\[
\left|
\D{X_u}{X_{u'}}-d_{vw}
\right|
\le
\D{X_u}{X_v}+\D{X_{u'}}{X_w}
<
2r_{\rm cl}.
\]
If \(d_{vw}+2r_{\rm cl}\le r_\rmp\), then all these distances lie in the local
bi-Lipschitz window. Hence the local Lipschitz bound gives
\[
\left|
\rmp(\D{X_u}{X_{u'}})-\rmp(d_{vw})
\right|
\le
2L_\rmp r_{\rm cl}.
\]
Averaging over \(\mathcal D(U_v,U_w)\), we obtain
\[
\left|
\apav{U_v,U_w}
-
\rmp(d_{vw})
\right|
\le
2L_\rmp r_{\rm cl}.
\]
On \(\Enn{W}{\lambda}{m}\),
\[
\left|
\pav{U_v,U_w}-\apav{U_v,U_w}
\right|
\le
\lambda.
\]
Combining the two inequalities proves the claim.  In the globally Lipschitz
case, the same argument applies without the restriction
\(d_{vw}+2r_{\rm cl}\le r_\rmp\).
\end{proof}

\begin{cor}[Distance recovery with known global link]
\label{cor:distance-recovery-known-global-link}
Assume the hypotheses of Theorem~\ref{theor:main}.  Suppose in addition that
\(\rmp\) is known and bi-Lipschitz on
\([0,\operatorname{diam}(M)]\).  Then, with probability \(1-o(1)\), one can
construct estimates \(\widehat d(v,w)\) for all \(v,w\in\bfV\) such that
\[
\left|
\widehat d(v,w)-\D{X_v}{X_w}
\right|
\le
C r
\qquad (v,w\in\bfV),
\]
where \(C\) depends only on the model parameters.
\end{cor}

\begin{proof}
Work on the event where Theorem~\ref{theor:main} holds and where
\(\Enn{\bfV}{r}{m_{\rm pair}}\) holds.  Then every output cluster satisfies
\[
U_v\text{ is a }(C_{\tref{theor:main}}'r,0)\text{-cluster with center }X_v,
\]
and
\[
|U_v|\ge m_{\rm pair}.
\]
Applying Lemma~\ref{lem:pair-averages-from-exact-clusters} with
\[
r_{\rm cl}=C_{\tref{theor:main}}'r,
\qquad
\lambda=r,
\]
gives
\[
\left|
\pav{U_v,U_w}
-
\rmp(\D{X_v}{X_w})
\right|
\le
C r
\qquad (v,w\in\bfV).
\]
Since \(\rmp\) is known and globally bi-Lipschitz, its inverse is
\(1/\ell_\rmp\)-Lipschitz on its range.  Define \(\widehat d(v,w)\) by
projecting \(\pav{U_v,U_w}\) onto \(\rmp([0,\operatorname{diam}(M)])\) and
then applying \(\rmp^{-1}\).  Projection can only decrease the error to the
range, and inversion gives the stated bound.
\end{proof}

\begin{cor}[Local distance recovery with known local link]
\label{cor:local-distance-recovery-known-link}
Assume the hypotheses of Theorem~\ref{theor:main}.  Suppose that \(\rmp\) is
known and bi-Lipschitz on \([0,r_\rmp]\).  Then, with probability \(1-o(1)\),
one can construct a set of certified local pairs \(\mathcal P_{\loc}\subseteq
\bfV^2\) and estimates \(\widehat d_{\loc}(v,w)\) for
\((v,w)\in\mathcal P_{\loc}\) such that
\[
\left|
\widehat d_{\loc}(v,w)-\D{X_v}{X_w}
\right|
\le
C r
\qquad ((v,w)\in\mathcal P_{\loc}),
\]
and
\[
\D{X_v}{X_w}\le \frac12 r_\rmp
\quad\Longrightarrow\quad
(v,w)\in\mathcal P_{\loc}.
\]
Moreover, every \((v,w)\in\mathcal P_{\loc}\) satisfies
\[
\D{X_v}{X_w}\le \frac12 r_\rmp+C r.
\]
\end{cor}

\begin{proof}
Again work on the intersection of the extraction event and
\(\Enn{\bfV}{r}{m_{\rm pair}}\).  Set
\[
r_{\rm cl}:=C'_{\tref{theor:main}}r,
\qquad
\varepsilon_{\rm pair}:=2L_\rmp r_{\rm cl}+r.
\]
By decreasing \(c_{\tref{theor:main}}\), we may assume
\[
2r_{\rm cl}+\frac{2\varepsilon_{\rm pair}}{\ell_\rmp}
\le
\frac14 r_\rmp.
\]
Define
\[
\mathcal P_{\loc}
:=
\left\{
(v,w):v\neq w,\
\pav{U_v,U_w}
\ge
\rmp(r_\rmp/2)-\varepsilon_{\rm pair}
\right\}.
\]
For \((v,w)\in\mathcal P_{\loc}\), first note that
\[
\D{X_v}{X_w}
\le
\frac12 r_\rmp+\frac{2\varepsilon_{\rm pair}}{\ell_\rmp}.
\]
Indeed, if the distance were larger, then every pair
\((u,u')\in\mathcal D(U_v,U_w)\) would have distance at least slightly above
\(r_\rmp/2\), and the lower bi-Lipschitz bound together with
\(\Enn{\bfV}{r}{m_{\rm pair}}\) would force
\[
\pav{U_v,U_w}<\rmp(r_\rmp/2)-\varepsilon_{\rm pair},
\]
a contradiction.  Thus
\[
\D{X_v}{X_w}+2r_{\rm cl}\le r_\rmp,
\]
and Lemma~\ref{lem:pair-averages-from-exact-clusters} applies:
\[
\left|
\pav{U_v,U_w}
-
\rmp(\D{X_v}{X_w})
\right|
\le
\varepsilon_{\rm pair}.
\]
Project \(\pav{U_v,U_w}\) onto \(\rmp([0,r_\rmp])\) and invert \(\rmp\) on
\([0,r_\rmp]\).  Since the inverse is \(1/\ell_\rmp\)-Lipschitz, the distance
error is at most
\[
\frac{\varepsilon_{\rm pair}}{\ell_\rmp}
\le Cr.
\]

Finally, if
\[
\D{X_v}{X_w}\le r_\rmp/2,
\]
then the same pair-average estimate gives
\[
\pav{U_v,U_w}
\ge
\rmp(\D{X_v}{X_w})-\varepsilon_{\rm pair}
\ge
\rmp(r_\rmp/2)-\varepsilon_{\rm pair},
\]
so \((v,w)\in\mathcal P_{\loc}\).  This proves the claim.
\end{proof}

\begin{defi}[Chain property at scale \(\rho_0\)]
\label{def:chain-property-rho0}
Let \((\mathcal X,\rho)\) be a metric space, let
\[
0<\rho_0\le{\rm diam}(\mathcal X),
\]
and let \(\eta\ge0\). We say that \((\mathcal X,\rho)\) satisfies the
\((\rho_0,\eta)\)-chain property if, for every \(x,y\in\mathcal X\) with
\(\rho(x,y)>\rho_0\), there is a chain
\[
p_0=x,\ p_1,\ldots,p_k=y
\]
such that
\[
\rho(p_i,p_{i+1})\le\rho_0
\qquad (0\le i<k),
\]
and
\[
\left|
\sum_{i=0}^{k-1}\rho(p_i,p_{i+1})-\rho(x,y)
\right|
\le\eta.
\]
\end{defi}

\begin{lemma}[Metric extension from local estimates]
\label{lem:metric-extension-local-estimates}
Let \((\mathcal X,\rho)\) be a finite metric space with
\[
0<\rho_0\le{\rm diam}(\mathcal X).
\]
Let \(\eta,\varepsilon\ge0\). Assume that \((\mathcal X,\rho)\) satisfies the
\((\rho_0,\eta)\)-chain property. Let
\(\mathcal P\subseteq\mathcal X\times\mathcal X\) be symmetric and suppose that
for every \((x,y)\in\mathcal P\) we are given an estimate
\(\widehat\rho(x,y)\) satisfying
\[
\left|\widehat\rho(x,y)-\rho(x,y)\right|\le\varepsilon.
\]
Assume also that
\[
\rho(x,y)\le\rho_0
\quad\Longrightarrow\quad
(x,y)\in\mathcal P.
\]
Define the weighted graph \(G_{\mathcal P}\) on \(\mathcal X\) by joining
\((x,y)\in\mathcal P\) with edge weight
\[
\widehat\rho(x,y)+\varepsilon.
\]
Let \(\rho_{\rm sp}\) be the shortest-path metric on \(G_{\mathcal P}\). Then,
for every \(x,y\in\mathcal X\),
\[
\rho(x,y)
\le
\rho_{\rm sp}(x,y)
\le
\rho(x,y)+\eta+C\frac{{\rm diam}(\mathcal X)}{\rho_0}\varepsilon,
\]
where \(C>0\) is a universal constant.
\end{lemma}

\begin{proof}
This is the standard shortest-path extension from local metric estimates; see
\cite[Lemma~8.1]{HuangJiradilokMossel2025RiemannianMetrics}. The added \(\varepsilon\) in each edge weight makes
every edge length an upper bound on the true distance, while the chain property
supplies a path whose accumulated local errors are controlled by
\({\rm diam}(\mathcal X)/\rho_0\).
\end{proof}

\begin{cor}[Global distances by chaining local estimates]
\label{cor:global-distance-by-chaining}
Assume the conclusions of
Corollary~\ref{cor:local-distance-recovery-known-link}, and suppose that
\((M, {\rm d})\) satisfies the \((r_\rmp/2,\eta)\)-\chainproperty{}.  Define a weighted graph on
\(\bfV\) with edge set \(\mathcal P_{\loc}\) and edge weights
\[
\widehat d_{\loc}(v,w)+Cr,
\]
where \(C\) is the constant from
Corollary~\ref{cor:local-distance-recovery-known-link}.  Let
\(\widehat d_{\rm sp}\) be the induced shortest-path metric.  Then
\[
\D{X_v}{X_w}
\le
\widehat d_{\rm sp}(v,w)
\le
\D{X_v}{X_w}
+
\eta
+
C\frac{\operatorname{diam}(M)}{r_\rmp}r
\qquad (v,w\in\bfV).
\]
In particular, if \(M\) is a geodesic space and the sampled points are
sufficiently dense at scale \(r\), the chain-property error \(\eta\) is controlled by
the sampling resolution, and local distance recovery extends to global distance
recovery.
\end{cor}
\begin{proof}
One thing that is subtle besides directly applying the above lemma together with Corollary~\ref{cor:local-distance-recovery-known-link} is to justify that the sampled points $X_{\bfV}$ also have the chain property. This is implied by $\Ept{\bfV}(r)$, which holds with high probability under the scale condition of Theorem~\ref{theor:main}, after adjusting constants. This automatically implies the chain property for the sampled points with $\eta$ increased by $C\tfrac{{\rm diam}(M)}{r_\rmp} r$, because every point along the chain has a sampled point within distance \(r\).
\end{proof}

\section{Auxiliary probability tools and proofs}
\label{app:auxiliary-probability-proofs}

\begin{lemma}[Chernoff inequality for \(0\)-\(1\) random variables]
\label{lem:chernoff-01-tform}
Let \(X_1,\dots,X_N\) be independent random variables taking values in
\(\{0,1\}\), and set
\[
S_N:=\sum_{i=1}^N X_i,
\qquad
\mu:=\mathbb E S_N.
\]
Then for every \(t\ge0\),
\[
\mathbb P(S_N\ge \mu+t)
\le
\exp\!\left(-\frac{t^2}{2\mu+t}\right),
\]
and for every \(0\le t\le\mu\),
\[
\mathbb P(S_N\le \mu-t)
\le
\exp\!\left(-\frac{t^2}{2\mu}\right).
\]
\end{lemma}
This standard form may be found, for instance, in \cite{Ver18}.

\begin{rem}
\label{rem:chernoff-form}
Suppose we want an upper-tail failure probability at most \(\exp(-L)\),
where \(L>0\). A convenient sufficient condition is
\[
\mathbb P\left(S_N\ge(\sqrt{\mu}+\sqrt L)^2\right)\le \exp(-L).
\]
Indeed, the upper-tail Chernoff bound shows that it is enough to take
\[
t\ge \frac{L+\sqrt{L^2+8\mu L}}{2},
\]
and \(t=L+2\sqrt{\mu L}\) is a valid choice.
\end{rem}

\begin{lemma}[Bernstein's inequality]
\label{lem:bernstein}
Let \(Y_1,\dots,Y_N\) be independent mean-zero random variables such that
\(|Y_i|\le M\) almost surely for all \(i\), and let
\[
\sigma^2:=\sum_{i=1}^N\operatorname{Var}(Y_i).
\]
Then for every \(t>0\),
\[
\mathbb P\!\left(\left|\sum_{i=1}^N Y_i\right|>t\right)
\le
2\exp\!\left(-\frac{t^2}{2(\sigma^2+Mt/3)}\right).
\]
\end{lemma}
This is the standard Bernstein inequality; see, for instance,
\cite[Theorem 2.9.5]{Ver18}.

\begin{lemma}[Net cardinality from lower regularity]
\label{lem:appendix-net-cardinality}
Let \((M,{\rm d},\mu)\) satisfy \lowerphiregularity{} up to scale \(r_\mu\). If
\(\delta\in(0,r_\mu]\) and \(\mathcal N\subset M\) is a maximal
\(\delta\)-separated set, then \(\mathcal N\) is a \(\delta\)-net of \(M\) and
\[
|\mathcal N|\le \frac1{\phi(\delta/2)}.
\]
\end{lemma}

\begin{proof}
Maximality gives the net property: otherwise a point at distance \(>\delta\)
from all points of \(\mathcal N\) could be added to \(\mathcal N\). If
\(\mathcal N=\{x_1,\ldots,x_N\}\), then the balls \(B(x_i,\delta/2)\) are
pairwise disjoint. Hence
\[
1=\mu(M)
\ge
\sum_{i=1}^N \mu(B(x_i,\delta/2))
\ge
N\phi(\delta/2),
\]
which proves the cardinality bound.
\end{proof}

\begin{proof}[Proof of Lemma~\ref{lem:uniform-lower-occupancy}]
Let \(\mathcal N\) be a maximal \(r/3\)-separated set. By
Lemma~\ref{lem:appendix-net-cardinality},
\[
|\mathcal N|\le \frac1{\phi(r/6)}.
\]
For \(x_i\in\mathcal N\), let
\[
S_i:=|\{v\in W:X_v\in B(x_i,r/3)\}|.
\]
Then \(S_i\) is binomial with mean at least \(n_\star\phi(r/3)\). The lower
Chernoff bound gives
\[
\mathbb P\left(S_i<\frac12 n_\star\phi(r/3)\right)
\le
\exp\left(-\frac{n_\star\phi(r/3)}8\right).
\]
Taking a union bound over \(\mathcal N\) and using
\(\phi(r/3)\ge\phi(r/6)\ge \Lambda \log n_\star/n_\star\), we get
\begin{align*}
\mathbb P\left(
\exists x_i\in\mathcal N:S_i<\frac12 n_\star\phi(r/3)
\right)
\le&
\frac1{\phi(r/6)}
\exp\left(-\frac{n_\star\phi(r/3)}8\right) \\
\le &
\frac{n_\star}{\Lambda\log n_\star}
\exp\left(-\frac{\Lambda}{8}\log(n_\star)\right)\\
\le &
\exp\left(
    \log(n_\star) - \tfrac{1}{8}\Lambda \log(n_\star)
\right)
\le
\exp\left(-\frac{\Lambda}{16}\log(n_\star)\right)\,,
\end{align*}
provided \(\Lambda\) is large enough than some universal constant.

On the complementary event, fix \(x\in M\). Since \(\mathcal N\) is an
\(r/3\)-net, choose \(x_i\in\mathcal N\) with \(\D{x}{x_i}\le r/3\). Then
\[
B(x_i,r/3)\subseteq B(x,r),
\]
because the balls are open. Therefore
\[
|\{v\in W:X_v\in B(x,r)\}|
\ge
S_i
\ge
\frac12 n_\star\phi(r/3).
\]
This is exactly \(\Ept{W}{r}\).
\end{proof}

\begin{proof}[Proof of Lemma~\ref{lem:fixed-v-navigation}]
Fix \(X_U=x_U\) and \(X_v=x_v\), and write \(m:=|U|\),
\[
a_u:=\rmp(\D{x_u}{x_v}),
\qquad
\xi_u:=\widetilde Z_{u,v}-a_u,
\qquad
B_u:=B_{u,v}.
\]
Then the \(B_u\)'s are independent Bernoulli\((\sp)\), the \(\xi_u\)'s are
independent centered \(\Ksg\)-subgaussian variables, and these two families are
independent. Also \(|a_u|\le M_\rmp\). Since
\[
Z_{u,v}=B_u(a_u+\xi_u),
\]
we have
\[
\cn{U}{v}-\acn{U}{v}
=
\frac1{\sp m}\sum_{u\in U}(Z_{u,v}-\sp a_u)
=T_1+T_2,
\]
where
\[
T_1:=\frac1{\sp m}\sum_{u\in U}(B_u-\sp)a_u,
\qquad
T_2:=\frac1{\sp m}\sum_{u\in U}B_u\xi_u.
\]

For \(T_1\), set \(Y_u=(B_u-\sp)a_u\). Then the \(Y_u\)'s are independent,
mean-zero, \(|Y_u|\le M_\rmp\), and
\[
\sum_{u\in U}\operatorname{Var}(Y_u)\le \sp m M_\rmp^2.
\]
Bernstein's inequality with \(s=t\sp m/2\) gives
\[
\mathbb P\left(|T_1|>\frac t2\mid X_U=x_U,X_v=x_v\right)
\le
2\exp\bigl(-c_1\sp m\min\{t^2,1\}\bigr).
\]

For \(T_2\), let \(N_B:=\sum_{u\in U}B_u\). By Chernoff,
\[
\mathbb P(N_B>2\sp m\mid X_U=x_U,X_v=x_v)
\le
\exp(-c\sp m).
\]
On the event \(N_B\le2\sp m\), conditioning on the \(B_u\)'s, the sum
\(\sum_u B_u\xi_u\) is subgaussian with parameter at most
\(\Ksg\sqrt{2\sp m}\). Thus, again with \(s=t\sp m/2\),
\[
\mathbb P\left(|T_2|>\frac t2\mid X_U=x_U,X_v=x_v\right)
\le
3\exp\bigl(-c_2\sp m\min\{t^2,1\}\bigr).
\]
Combining the two bounds and adjusting constants proves the claim.
\end{proof}

\begin{proof}[Proof of Lemma~\ref{lem:uniform-navigation-event}]
For each fixed \(v\in V\), Lemma~\ref{lem:fixed-v-navigation} gives
\[
\mathbb P\left(
\left|\cn{U}{v}-\acn{U}{v}\right|>\epsnav{U}{n_\star}
\ \middle|\ X_U=x_U,\ X_V=x_V
\right)
\le
C_{\rm link}
\exp\left(
-c_{\rm link}\sp |U|
\min\{\epsnav{U}{n_\star}^2,1\}
\right).
\]
Since \(\sp |U|\ge \Lambda \log n_\star\), we have
\(\epsnav{U}{n_\star}\le1\), and the exponent is at most
\(-c_{\rm link}\Lambda \log n_\star\). A union bound over
\(|V|\le n_\star\) gives
\[
\mathbb P\left(
\Enavi{U}{V}{n_\star}^{\,c}
\ \middle|\ X_U=x_U,X_V=x_V
\right)
\le
e^{- \tfrac{1}{2}c_{\rm link}\Lambda \log n_\star}\,,
\]
provided that $\Lambda$ is larger than some universal constants depending on \(C_{\rm link}\) and \(c_{\rm link}\).
\end{proof}

\begin{proof}[Proof of Lemma~\ref{lem:uniform-pair-average}]
Fix a realization \(X_W=x_W\). Let \(U_1,U_2\subseteq W\) have
\[
|U_1|=a,\qquad |U_2|=b,\qquad a,b\ge m.
\]
For each unordered pair \(e=\{u,v\}\subseteq W\), define
\[
w_e:=
\mathbf 1_{\{u\in U_1,\ v\in U_2\}}
+
\mathbf 1_{\{v\in U_1,\ u\in U_2\}}.
\]
Then \(w_e\in\{0,1,2\}\), and
\[
D:=\sum_{e\subseteq W}w_e
=|U_1||U_2|-|U_1\cap U_2|
=|\mathcal D(U_1,U_2)|.
\]
For \(e=\{u,v\}\), write
\[
a_e:=\rmp(\D{x_u}{x_v}),
\qquad
\xi_e:=\widetilde Z_{u,v}-a_e,
\qquad
B_e:=B_{u,v}.
\]
As in the fixed-vertex proof,
\[
\pav{U_1,U_2}-\apav{U_1,U_2}
=
\frac1{\sp D}\sum_{e\subseteq W}w_e(Z_e-\sp a_e)
=T_1+T_2,
\]
where
\[
T_1:=\frac1{\sp D}\sum_{e\subseteq W}w_e(B_e-\sp)a_e,
\qquad
T_2:=\frac1{\sp D}\sum_{e\subseteq W}w_eB_e\xi_e.
\]

For \(T_1\), set \(Y_e=w_e(B_e-\sp)a_e\). Then
\[
|Y_e|\le2M_\rmp,
\qquad
\sum_{e\subseteq W}\operatorname{Var}(Y_e)
\le
\sp M_\rmp^2\sum_{e\subseteq W}w_e^2
\le
2\sp D M_\rmp^2.
\]
Bernstein's inequality with \(s=\lambda\sp D/2\) gives
\[
\mathbb P\left(|T_1|>\frac\lambda2\mid X_W=x_W\right)
\le
2\exp(-c_1\lambda^2\sp D).
\]

For \(T_2\), let \(W_B:=\sum_{e\subseteq W}w_eB_e\). Since
\(\mathbb E[W_B\mid X_W=x_W]=\sp D\) and \(0\le w_eB_e\le2\), Bernstein's
inequality gives
\[
\mathbb P(W_B>2\sp D\mid X_W=x_W)
\le
\exp(-c_2\sp D).
\]
On the event \(W_B\le2\sp D\), conditioning on the \(B_e\)'s,
\[
\sum_{e\subseteq W}w_eB_e\xi_e
\]
is subgaussian with parameter at most
\[
\Ksg\left(\sum_{e\subseteq W}w_e^2B_e\right)^{1/2}
\le
2\Ksg\sqrt{\sp D}.
\]
Therefore
\[
\mathbb P\left(|T_2|>\frac\lambda2\mid X_W=x_W\right)
\le
3\exp(-c_3\lambda^2\sp D).
\]
Combining the two bounds,
\[
\mathbb P\left(
\bigl|\pav{U_1,U_2}-\apav{U_1,U_2}\bigr|>\lambda
\ \middle|\ X_W=x_W
\right)
\le
4\exp(-c_4\lambda^2\sp D).
\]

It remains to union-bound over \(U_1,U_2\). Since \(m\ge2\),
\[
D=ab-|U_1\cap U_2|\ge ab-\min\{a,b\}\ge\frac{ab}{2}.
\]
For \(\alpha=a/n_\star\) and \(\beta=b/n_\star\), the number of pairs
\((U_1,U_2)\) with these cardinalities is at most
\[
\binom{n_\star}{a}\binom{n_\star}{b}
\le
\exp\left((\alpha+\beta)n_\star\log(e/\varphi)\right),
\]
because \(\alpha,\beta\ge\varphi\). Also
\(\alpha+\beta\le2\alpha\beta/\varphi\). Thus, if
\[
\sp n_\star\ge \frac{C\log(e/\varphi)}{\varphi\lambda^2}
\]
with \(C\) large enough, then each cardinality class contributes at most
\[
4\exp(-c_5\lambda^2\sp\alpha\beta n_\star^2).
\]
Summing over at most \(n_\star^2\) cardinality pairs and using
\(\alpha\beta\ge\varphi^2\), we obtain
\[
\mathbb P\left(
\Enn{W}{\lambda}{m}^{\,c}
\ \middle|\ X_W=x_W
\right)
\le
4n_\star^2
\exp(-c_5\lambda^2\sp\varphi^2n_\star^2).
\]
The assumed lower bound on \(\sp n_\star\), together with \(m\ge2\), implies
that the right-hand side is at most
\[
\exp\{-c\varphi n_\star\log(e/\varphi)\}
\]
after increasing \(C\) and decreasing \(c\). This proves the lemma.
\end{proof}

\begin{proof}[High-probability consequence of Lemma~\ref{lem:uniform-pair-average}]
The failure probability in Lemma~\ref{lem:uniform-pair-average} is
\[
\exp\{-c\varphi n_\star\log(e/\varphi)\}.
\]
If \(\varphi n_\star\log(e/\varphi)\gg\log n_\star\), this is
\(n_\star^{-\omega(1)}\).
\end{proof}
\end{document}